\documentclass[11pt]{article}
\usepackage{gmac_macros}
\renewcommand{\edit}[1]{\textcolor{black}{#1}}
\usepackage{fullpage}

\begin{document}

\title{Many-User Multiple Access with Random User Activity: \\ Achievability Bounds and Efficient Schemes}
  \author{Xiaoqi Liu
   \qquad
  Pablo Pascual Cobo
  \qquad 
  Ramji Venkataramanan}
  \date{}

\maketitle

\makeatletter
\def\blfootnote{\gdef\@thefnmark{}\@footnotetext} 
\makeatother

\blfootnote{X.\@ Liu was  supported  by a Schlumberger Cambridge International Scholarship and by the  UKRI 
  under the UK government’s Horizon Europe Guarantee 
  (grant number EP/Y028333/1). P.\@ Pascual Cobo was supported by an EPSRC Doctoral Training Partnership Award.  
 This paper was presented in part at the 2024 IEEE International Symposium on Information Theory and accepted to the IEEE Transactions on Information Theory. Author emails:  \texttt{\{xl394,pp423,rv285\}@cam.ac.uk}. }

\begin{abstract}
 We study the Gaussian multiple access channel with random user activity, in the regime where the number of users is proportional to the code length. The receiver may know some statistics about the number of active users, but does not know the exact number nor the identities of the active users. We  derive two  achievability bounds on the probabilities of \edit{missed detection}, false alarm, and active user error, and propose an efficient CDMA-type scheme whose performance can be compared against these bounds.   
 The first  bound  is a finite-length result based on  Gaussian random codebooks and  maximum-likelihood decoding. The second is an asymptotic bound, 
 established using  spatially coupled Gaussian codebooks and approximate message passing (AMP) decoding. These bounds can be used to compute an achievable \edit{tradeoff} between the active user density and energy-per-bit, for a fixed user payload and target error rate. 
The efficient CDMA scheme uses a spatially coupled signature matrix and AMP decoding, and we give rigorous asymptotic guarantees on its error performance. Our analysis provides the first state evolution result for spatially coupled AMP with matrix-valued iterates, which may be of independent interest. Numerical experiments demonstrate the promising error performance of the CDMA scheme for both small and large user payloads, when compared with the two  achievability bounds.
\end{abstract}

\tableofcontents
\section{Introduction}
We study the  Gaussian multiple access channel (GMAC), with  output of the form 
\begin{equation}\label{eq:standard_gmac}
\by = \sum_{\ell=1}^L\bc_\ell+\bveps,    
\end{equation}
where $L$ is the number of users,  $\bc_\ell\in\reals^n$ is the codeword of user $\ell$, $\bveps\sim \normal_n(\0, \sigma^2\bI)$ is random channel noise with spectral density $N_0=2\sigma^2$, and $n$ is the number of channel uses.
Motivated by applications such as the Internet of Things, there has been much interest in the \emph{many-user} regime,  where the number of users $L$ grows with  the code length $n$  \cite{chen2014many, chen2017capacity,polyanskiy2017perspective, zadik2019improved,  RaviK22}.

In this paper, we consider the  many-user regime where the number of users $L$  grows proportionally with $n$; the ratio $\mu:=L/n$ is called the user density \cite{ polyanskiy2017perspective}. Each user transmits a fixed  number of bits $k$  (payload) under a constant energy-per-bit constraint $\|\bc_\ell\|^2_2/k\le E_b$.  A key question in this regime is to characterize the tradeoff between user density $\mu$, the signal-to-noise ratio $E_b/N_0$ and the error probability in decoding the codewords $\{\bc_\ell\}$ from $\by$. The standard error metric is the per-user probability of error 
$\PUPE :=\frac{1}{L} \sum_{\ell=1}
^L \prob(\bc_\ell\neq \widehat{\bc_\ell})$
where $\widehat{\bc_\ell}$ denotes the decoder estimate of $\bc_\ell$. 
Achievability  and converse bounds were derived in \cite{polyanskiy2017perspective, zadik2019improved, kowshik2022improved}  in terms of the minimum $E_b/N_0$ needed to achieve a given target $\PUPE$ for a given user density $\mu$. Efficient coding schemes were proposed in \cite{hsieh2022near} in an attempt to approach the converse bound.

In practical  GMAC settings, the users are seldom consistently active. Instead, they are active in a sporadic and uncoordinated manner. Motivated by recent work in this direction \cite{chen2014many, Ordentlich2017low, Yavas2021gaussian, Yavas2021random, ngo2022unsourced, Ngo2024unsourced_alarm, gao2024unsourced},  we study the GMAC with \emph{random user activity}.
Letting $\Ka \le L$ denote the (random) number of active users, we  consider the proportional regime with $\E[\Ka]/L\to \alpha$ 
\edit{where $\alpha$ is a fixed constant. This implies that $\E[\Ka]/n \to \alpha\mu=: \mu_\a$ where $\mu_\a$ is also a  fixed constant.} 
The receiver may know the distribution of $K_\a$ or some statistics, but it does not know the identities of the active users  nor the exact value of $K_\a$. We study the tradeoff between the \emph{active} user density $\mu_\a$,  the signal-to-noise ratio $E_b/N_0$ and the decoding performance, measured in terms of the probabilities of \edit{missed detection} ($\MD$), false alarm ($\FA$), and active user error ($\AUE$) (see Section \ref{sec:3_errors}).

We emphasize that our GMAC setting is different from unsourced random-access \cite{polyanskiy2017perspective, Amalladinne2020Coded,fengler2021SPARCS,amalladinne2022unsourced, ngo2022unsourced, Ngo2024unsourced_alarm, gao2024unsourced}, where  all the users share the same codebook, 
a subset of them are active, \edit{and the decoder  recovers the set of transmitted messages (without necessarily recovering the identities of the users who sent the messages).} \edit{In contrast, our setting assigns each user a separate codebook,}
\edit{and the decoder can recover both the transmitted messages and  the sender identities.} While unsourced random-access is particularly relevant for  grant-free communication systems,  coding schemes based on sparse regression for the standard AWGN channel and the GMAC \cite{venkataramanan19monograph}  form the basis of several state-of-the-art schemes for unsourced  random-access
\cite{fengler2021SPARCS, amalladinne2022unsourced}. \edit{Analogously, the coding schemes  proposed in this paper can potentially be extended to the unsourced random-access setting.}

\subsection{Main Contributions}
In this work, we establish new  information-theoretic bounds and propose an efficient scheme for the  many-user GMAC with random user activity. 

\begin{enumerate}[i)]
    \item In Theorem \ref{thm:ngo_ach}, we derive finite-length achievability bounds on the probabilities of  \edit{missed detection} ($\MD$), false alarm ($\FA$), and active user error ($\AUE$) \edit{for an arbitrary distribution of $\Ka$ denoted $p_{\Ka}$}; these error probabilities 
    are defined in Section \ref{sec:3_errors}. The bounds are obtained by analyzing random i.i.d.\@ Gaussian codebooks with 
   joint  maximum-likelihood decoding.  

 \item Theorem \ref{thm:sparc_asymp_errors} provides an asymptotic  achievability bound for the GMAC with random user activity in the regime where $n,L\to\infty$ proportionally. The bound is derived by analyzing a scheme with a spatially coupled Gaussian codebook and iterative Approximate Message Passing (AMP) decoding. \edit{For simplicity the bound assumes $p_{\Ka}$ follows the Binomial distribution $\text{Bin}(\alpha, L)$, but our approach can be generalized to arbitrary $p_{\Ka}$.} 
 \edit{In the deterministic setting  (i.e., when $\alpha=1$),} the  bound in Theorem \ref{thm:sparc_asymp_errors} gives a strictly larger achievable region  than the bounds by Kowshik in  \cite{kowshik2022improved} and Zadik et  al.\@ \cite{zadik2019improved},  the best existing bounds for moderate to large $k$ (payload). See Corollary \ref{cor:asymp_error_alpha=1} and the discussion below it  for details, and Fig.~\ref{fig:alpha=1_kowshik_v_kuan_v_ours} for an example with $k=60$.

 The schemes used to derive the bound in Theorem \ref{thm:ngo_ach} uses joint maximum-likelihood decoding and is infeasible, even for small payloads $k$. The scheme for Theorem \ref{thm:sparc_asymp_errors} uses AMP decoding and can be implemented for small values of $k$, but the codebook size for each user (and hence the decoding complexity) scales exponentially with $k$.

 \item To  handle payloads of all sizes,
we propose an efficient CDMA-type  scheme, and characterize its asymptotic error performance  in Theorem \ref{thm:MetricsSE}, in terms of $E_b/N_0$, $\mu_\a$ and  $\alpha$. \edit{The scheme has small storage overhead since each user is assigned only a signature sequence rather than an entire codebook.} 
    The scheme uses a spatially coupled Gaussian design matrix (whose columns are the users' signature sequences) and an  AMP decoder tailored to handle random user activity.   We propose two different  choices for the denoising function used within the AMP decoder: a Bayes-optimal denoiser for small payloads, and a novel  thresholding denoiser that is computationally efficient for larger payloads up to hundreds of bits. \edit{Theorem \ref{thm:MetricsSE} assumes $p_{\Ka}=\text{Bin}(\alpha, L)$ for simplicity, but can be extended to general $p_{\Ka}$.}

    The AMP decoder, which aims to recover a matrix-valued signal (with a known prior) from noisy linear observations, can be viewed as a generalization of the spatially coupled AMP algorithm in \cite{donoho2013information} for vector-valued signals. To our knowledge, this is the first application of spatial coupling to random linear models with matrix-valued signals, which may be of independent interest. 
    Theorem \ref{thm:gen_sc_gamp_SE} provides the asymptotic distributional characterization of the AMP algorithm in this setting. 

 In Section \ref{subsec:num_results_CDMA}, we   compare the two achievability bounds   with the efficient   CDMA-type scheme. Our numerical results  demonstrate the promising performance of the  scheme for both small and larger payloads. 
\end{enumerate}

\edit{We summarize the assumptions and the details of coding schemes used for each of the above results in Table \ref{tab:summary}.}

  \begin{table}[t]
\centering
\begin{tabular}{|p{3.5cm}|p{3.2cm}|p{3.2cm}|p{3.5cm}|}
\hline
 & \centering Theorem \ref{thm:ngo_ach} & \centering 
 Theorem \ref{thm:sparc_asymp_errors} & \parbox{3cm}{\vspace{-0.05cm}\centering Theorem \ref{thm:MetricsSE}} \\
\hline
Assumptions & \parbox{3cm}{\centering \vspace{0.3cm}Fixed $n, L$,  \\ arbitrary $p_{\Ka}$ \vspace{0.3cm}} & \multicolumn{2}{c|}{\parbox{6cm}{\centering $n,L\to \infty$, \ $L/n =\mu$,\\  $p_\Ka=\text{Bin}(\alpha, L)$ but can be extended to arbitrary $p_{\Ka}$}} \\
\hline
Codebook & \parbox{3.5cm}{\centering \vspace{0.05cm}Random codebook\\with i.i.d.\@ design} & \parbox{3.5cm}{\centering \vspace{0.05cm}Random codebook\\with SC design}  & \parbox{3.5cm}{\centering \vspace{0.05cm}CDMA with SC design}\\
\hline
\vspace{-0.08cm}Decoding algorithm & \centering Maximum-likelihood & \multicolumn{2}{c|}{\raisebox{-0.3cm}{\centering AMP decoding}}\\
\hline
Memory and computational complexity &\multicolumn{2}{c|}{\raisebox{-0.3cm}{Exponential in $k$ }}& \raisebox{-0.3cm}{\centering  \qquad Linear in $k$}\\
\hline
\end{tabular}
\caption{Summary of main results, where  i.i.d.\@ (SC)  design refers to i.i.d.\@ (spatially coupled) Gaussian codebooks,  and ``Memory and computational complexity'' refers to the complexity of the coding scheme 
analyzed to obtain the result.}
\label{tab:summary}
\end{table}

\subsection{Related Work}  

\paragraph{Unsourced random access} Finite-length achievability bounds for  unsourced random-access, under various settings,  have been established in  \cite{Yavas2021gaussian, Yavas2021random,ngo2022unsourced, Ngo2024unsourced_alarm, gao2024unsourced}. The main results in \cite{Yavas2021gaussian, Yavas2021random}  are  stated in terms of   the joint-user error probability, with 
\cite{Yavas2021random} also providing bounds on the $\PUPE$. The probabilities of $\MD$ and $\FA$ are quantified separately in \cite{ngo2022unsourced, Ngo2024unsourced_alarm} and \cite{gao2024unsourced}.  Our finite length achievability bounds in Theorem 
  \ref{thm:ngo_ach} are  similar to those in \cite{ngo2022unsourced}, but a key difference is that in our setting the users have distinct codebooks, and we derive bounds on the probabilities of $\MD, \FA$, and $\AUE$ separately.

  \paragraph{Approximate Message Passing}   AMP is a class of first-order iterative algorithms initially proposed for estimation in random linear models  \cite{Kab03, donoho2009message, bayati2011dynamics, Krz12}. It has since  been applied to a variety of high-dimensional estimation problems such as low-rank matrix estimation \cite{fletcher2018matrix, montanari2021lowrank, barbier2020matrix} and generalized linear models \cite{rangan2011generalized, Mai20, Mon21}. Two appealing features of AMP are: i) it can be tailored to take advantage of the signal prior, and  ii)  under suitable model assumptions, its estimation performance in the high-dimensional limit can be characterized by a simple deterministic recursion called state evolution.  For a variety of high-dimensional statistical estimation problems, AMP has been shown to be optimal among first-order methods  \cite{celentano2020estimation}, and  is conjectured to achieve the optimal asymptotic estimation error among polynomial-time algorithms  \cite{celentano2022fundamental}.  We refer the reader to \cite{feng2022unifying}  for a survey on AMP algorithms for various statistical models.

  \paragraph{Spatial coupling} Coding schemes based on spatially coupled random linear models with  AMP decoding have been shown to be capacity-achieving for point-to-point channels \cite{barbier2017approximate,barbier2019universal,rush2021capacity},  and give  the best achievable tradeoffs for the many-user GMAC \cite{hsieh2022near,kowshik2022improved,liu2024LDPC}. These works show that spatially coupled designs significantly improve on the performance of i.i.d.\@ designs in certain regimes.
More generally, spatially coupled designs with AMP decoding have been shown to achieve the Bayes-optimal estimation error for both linear models \cite{donoho2013information} and generalized linear models \cite{cobo2023bayes} (with vector-valued signals). To  our knowledge, Theorem  \ref{thm:gen_sc_gamp_SE} presents the first state evolution result for AMP applied to spatially coupled linear models with a matrix-valued signal.

\paragraph{CDMA-based schemes} The asymptotic error performance   of CDMA with i.i.d. signature sequences 
    has been studied in a number of works in the absence of random user activity  
\cite{verduShamai99,shamaiVerdu01,CaireGuemRomyVerdu2004,tanaka2002cdma,guo2005randomly}. Assuming the signature sequences are i.i.d.\@ sub-Gaussian, AMP is the best known  decoding algorithm \cite{donoho2009message, bayati2011dynamics}. Recently, variants of CDMA have been    studied  for  activity detection in multiple-antenna networks \cite{chen2018sparse,liu2018massive, cakmak2024joint}. The  decoding task in these settings is an instance of the multiple measurement vector (MMV) problem, where the goal is to recover a matrix signal (with a specific prior) from  noisy linear observations. The AMP algorithm for the MMV problem
 \cite{Ziniel2013efficient} is used for decoding in \cite{chen2018sparse,liu2018massive}. However, the design matrix in all these works is i.i.d., whereas we use a spatially coupled design, which requires a novel AMP algorithm.

 \subsection{Performance Metrics}\label{sec:3_errors}
Let $\mc{C^{(\ell)}}=\{\bc_1^{(\ell)}, \bc_2^{(\ell)}, \dots, \bc_M^{(\ell)}\}$ denote the codebook of user $\ell$. Each active user transmits $k$ bits, so  $M=2^k$.  Let $w_\ell\in \{\emptyset, 1, 2, \dots, M\}$ denote the  index of the codeword chosen by user $\ell$, where $w_\ell=\emptyset$ indicates that the user is silent (not active). 
We use index pair $(i,j)$ to refer to
the $j$th codeword of the $i$th codebook. 
Then  the set of transmitted codewords can be defined as $\W:=\{(\ell, w_\ell): w_\ell\neq \emptyset\}$,  and the GMAC channel output can be written as
\begin{equation}
    \by = \sum_{\ell:(\ell,w_\ell)\in \W}\bc_{w_\ell}^{(\ell)} +\bveps,
\end{equation} 
where the number of active users is  $\Ka=\abs{\W}\le L$.

Given the channel output $\by$ and the codebooks $\{\mc{C}^{(\ell)}\}$, the decoder aims to recover the set of transmitted codewords. 
Let $\hw_\ell\in \{\emptyset, 1, 2, \dots, M\}$ denote the decoded codeword  in the $\ell$th codebook, and let $\hW:=\{(\ell, \hw_\ell): \hw_\ell\neq \emptyset\}$ denote the set of decoded codewords, or the decoded set in short, with size $|\hW|=: \hKa$. 

Such a decoder can make three types of errors, which we call \edit{missed detection} ($\MD$), false alarm ($\FA$)  and active user error ($\AUE$). $\MD$ refer to an active user declared silent, and  $\FA$ to a silent user declared active. An $\AUE$  occurs when an active users is correctly declared active, but the decoded codeword is incorrect. Given the codebooks $ \mc{C}^{(1)},\dots, \mc{C}^{(L)}$, the error probabilities corresponding to these events are defined as follows:
\begin{align}
     \pMD&:=\E\left[\ind\{\Ka\neq 0\}\cdot  \frac{1}{\Ka}\sum_{\ell: (\ell,w_\ell)\in \W}\ind\{\widehat{w_\ell} = \emptyset\}\right] ,\label{eq:def_my_pMD} \\
     \pFA &:=\E\left[\ind\{\widehat{\Ka}\neq 0\}\cdot  \frac{1}{\widehat{\Ka}}\sum_{\ell:(\ell, \hw_\ell)\in\hW}\ind\{w_\ell=\emptyset\}\right],
     \label{eq:def_my_pFA}\\
     \pAUE &:= \E\left[\ind\{\Ka\neq 0\}\cdot  \frac{1}{\Ka}\sum_{\ell:(\ell, w_\ell)\in \W}\ind\big\{\widehat{w_\ell}    \notin\{w_\ell, \emptyset\}\big\}\right] .\label{eq:def_my_pAUE}
\end{align} 
Here, the expectation is taken over $\Ka$, $ \hKa$, and the uniform distribution over messages for each active user.
 \paragraph*{Notation}
 The indicator function of an event $\mc{A}$ is denoted by $\ind\{\mc{A}\}$. For positive integers $a$ and $b$, $[a]:=\{1,...,a\}$ and $[a:b]=\{a, a+1, \dots, b\}$. We use boldface letters for vectors and matrices and  plain font for scalars. We use $A_{ij}$ to denote the $(i, j)$-th entry of matrix $\bA$. 
 Let $x^+=\max\{x, 0\}$. \edit{We use i.i.d.\@ as shorthand for independently and identically distributed.}
 We write $\mc{N}_d(\bmu, \bSigma)$ for a $d$-dimensional Gaussian with mean $\bmu$ and covariance $\bSigma$. \edit{We use  $\text{Bin}(p, n)$ to denote a Binomial distribution with  $n$ trials and success probability $p$.}  We write $\text{Gamma}(k, \theta) $ for the Gamma distribution with shape $k$ and scale $\theta$. We denote the Gamma function  by $\Gamma(k) = \int_0^\infty t^{k-1} e^{-t}\de t$, and the   lower and upper incomplete Gamma functions by $\gamma(k, w) = \int_0^w t^{k-1} e^{-t}\de t$ and $\Gamma(k, w) = \int_w^\infty t^{k-1} e^{-t}\de t$.

 \section{Finite-\edit{L}ength Random Coding Achievability Bounds}\label{sec:ngo_ach}
In this section, we derive  finite-length achievability bounds on $\pMD, \pFA$ and $\pAUE$ for 
a given $E_b/N_0$, $L$ and $n$, using a random coding scheme with maximum-likelihood decoding.    Without loss of generality,  we take $\sigma^2= N_0/2 = 1$, so the constraint on $E_b/N_0$ is enforced via $E_b$. 
The number of active users \edit{$\Ka$ is} assumed to be drawn from a  distribution $p_{\Ka}$.

\subsection{Random Coding and Maximum-\edit{L}ikelihood Decoding}\label{sec:ngo_ach_scheme}
Recall that each active user transmits $k$ bits using a length $n$ codeword.   For $\ell\in[L]$, 
we construct  the codebook $\mc{C}^{(\ell)}$ with   codewords 
 $\bc_j^{(\ell)} =\tilde{\bc}_{j}^{(\ell)} \ind\{\|\tilde{\bc}_{j}^{(\ell)}\|_2^2\le E_b k\}$, where   
 $\tilde{\bc}_j^{(\ell)}\stackrel{\text{i.i.d.}}{\sim}\normal_n(\bzero, E_b'k/n\bI)$ and $ 0<E_b' < E_b$. Note that the codewords across  codebooks are distinct almost surely.  The indicator function ensures the codewords $\bc_{j}^{(\ell)}$ satisfy the energy-per-bit constraint  $\|\bc_j^{(\ell)}\|_2^2/k\le E_b $.    Recall that the transmitted codeword of user $\ell \in [L]$ is denoted by  $\bc_{w_\ell}^{(\ell)}$.  The  decoder  
 first computes the  maximum-likelihood estimate $\Kap$ of  the number of active users $\Ka$, via
\begin{equation}\label{eq:determine_Ka'}
\Kap=\argmax_{K\in[\kl:\ku]} \, p(\by|\Ka=K)
 \end{equation}   
where  $\kl, \ku$  are deterministic integers between 0 and $L$, chosen such that $\prob(\Ka \notin [\kl: \ku])$ is small (e.g., $10^{-6}$).  This step is similar to the scheme in \cite{ngo2022unsourced}.
Using $\Kap$, the  decoder then produces an estimate  $\{\bc_{\hw_\ell}^{(\ell)}\, , \,\ell\in [L]\}$ of the transmitted set of  codewords $\{\bc_{w_\ell}^{(\ell)}\, , \, \ell\in [L]\}$ by solving  a combinatorial least squares problem. 
 Let  $w_\ell' \in\{\emptyset, 1, 2, \dots, M\}$ for $\ell\in[L]$ and define $\W':=\{(\ell, w'_\ell): w_\ell'\neq \emptyset\}$. The decoder computes:
\begin{equation}\label{eq:min_distance_decoder}
    \{\widehat{w_1}, \widehat{w_2}, \dots, \widehat{w_L}\} = \argmin_{\substack{\{w_1', w_2', \dots, w_L'\} : \\\uKap\le \abs{\W'}\le \oKap}} \|c(\W')-\by\|_2^2\, , \quad \text{where}\quad  c(\W'):=\sum_{\ell:(\ell,w_\ell')\in \W'} \bc_{w_\ell'}^{(\ell)}\, .
\end{equation}
Here, $\uKap=\max\{\kl,\Kap- \rl\}$ and $
  \oKap=\min\{\ku, \Kap+ \ru\}$, where \edit{$\rl, \ru\ge 0$} are prespecified  \edit{lower and upper} decoding \edit{radii  like in \cite{ngo2022unsourced}}.
In words, the decoder 
searches over every set  $\W'$ with size between $\uKap $ and $ \oKap$, which contains at most one codeword from each codebook,  to minimize  the distance between the sum of the codewords in $\W'$  and the  channel output $\by$. \edit{Setting $\rl=\ru$ provides equal control over $\pMD$ and $\pFA$, whereas setting $\rl<\ru$ controls  $\pMD$ more strictly than $\pFA$, and vice versa (see Fig.~\ref{fig:ngo_r_asymmetric}).} 
Choosing \edit{$\rl=\ru=L$}  gives the biggest search range with $\uKap=\kl$ and $\oKap=\ku$. 
The complexity of this decoder grows exponentially with $n$, making it computationally infeasible. 

\edit{ 
Applying this decoder,   $\pMD, \pFA$ and $\pAUE$ averaged over the  ensemble of joint-codebooks  $(\mc{C}^{(\ell)})_{\ell\in [L]}$ described above can \emph{individually} satisfy certain upper  bounds  $\pMD\le \epsMD, \pFA\le \epsFA$ and $\pAUE\le \epsAUE$, respectively. However, time-sharing among  two deterministic joint-codebooks from this ensemble is required to guarantee that all three bounds are achieved  \emph{simultaneously}, see \cite[Theorem 8]{Yavas2021random} and \cite[Remark 2]{ngo2022unsourced}.  Time-sharing  is a common  randomization technique over deterministic coding schemes to achieve  performance that no single deterministic scheme can attain  
\cite{polyanskiy2011feedback, Yavas2021random}. With time-sharing, } an error analysis of the random coding scheme with maximum-likelihood decoding  leads to the following result. 

\begin{thm}[Finite-length achievability bounds]\label{thm:ngo_ach} 
Let $ 0<E'_b< E_b$ and $P=E_bk/n$, $P'=E_b'k/n$. The decoding \edit{radii $\rl, \ru\ge 0$},
 and   $\kl, \ku $ satisfying $0\le \kl\le \ku\le L$ are defined as above. 
  \edit{There exists a randomized coding scheme (with time-sharing)} using the decoder described in \eqref{eq:determine_Ka'}--\eqref{eq:min_distance_decoder} that achieves 
  error probabilities satisfying 
$\pMD\le \epsMD$, $\pFA\le \epsFA$ and $\pAUE\le \epsAUE$ \edit{simultaneously}.
%
The quantities $\epsMD, \epsFA, \epsAUE$ are defined as follows in terms of the distribution of the number of active users $p_{\Ka}$:
\begin{align}
    & \epsMD =  \tp \,  + \, \sum_{\ka=\kl}^\ku p_{\Ka}(\ka) \sum_{\kap=\kl}^\ku \sum_{t\in\T}\sum_{\hatt\in\hTt} 
    \min\{p(t, \hatt), \xi(\ka, \kap)\} \cdot \sum_{\psi=0}^{\psiu} \frac{\nu(\tmin,\psi)}{\nu(\tmin)}
    \DelMD  \,,\label{eq:thm_epsMD}\\
   &  \qquad \DelMD= \frac{(\ka-\okap)^+ + (t-\hatt)^++\psi}{\ka} \,,  \quad \okap=\min\{\ku, \kap+ \ru\},  \nonumber\\
    & \epsFA = \tp \,  + \, \sum_{\ka=\kl}^\ku p_{\Ka} (\ka)\sum_{\kap=\kl}^\ku \sum_{t\in\T}\sum_{\hatt\in\hTt}
    \min\{p(t, \hatt), \xi(\ka, \kap)\} \cdot \sum_{\psi=0}^{\psiu}
     \frac{\nu(\tmin,\psi)}{\nu(\tmin)}\DelFA
      \,,\label{eq:thm_epsFA}\\
   & \qquad  \DelFA = \frac{(\ukap-\ka)^+ +(\hatt-t)^+ +\psi}{\ka-t -(\ka-\okap)^+  +\hatt+ (\ukap-\ka)^+}\, ,  \quad \ukap=\max\{\kl, \kap-\rl\},\nonumber \\
    & \epsAUE = \tp \, + \,  \sum_{\ka=\kl}^\ku p_{\Ka}(\ka)  \sum_{\kap=\kl}^\ku\sum_{t\in\T}\sum_{\hatt\in\hTt} \min\{
    p(t, \hatt), \xi(\ka, \kap)\} 
     \cdot \sum_{\psi=0}^{\psiu}   \frac{\nu(\tmin,\psi)}{\nu(\tmin)}\frac{\tmin-\psi}{\ka} \,.\label{eq:thm_epsAUE}
\end{align}
Here,
\begin{align}
         \tp & =  \prob\big(\Ka\notin [\kl: \ku]\big) + \E[\Ka]\frac{\Gamma(\frac{n}{2}, \frac{nP}{2P'})}{\Gamma(\frac{n}{2})}\, , \label{eq:tp_def}
\end{align}
and the sets $\T, \hTt$ are defined as:
\begin{align}
& \T = \Big[0: \min\{\ka, \okap\}\Big] \,, \quad 
 \hTt =\Big[\left\lbrace t+(\ka-\okap)^+ -(\ka-\ukap)^+ \right\rbrace^+ : \tu\Big] \,, \label{eq:T_thm}\\
  & \quad \text{ where }  \   \tu = \min\left\lbrace \okap-(\ukap-\ka)^+,\;  t+(\okap-\ka)^+ -(\ukap-\ka)^+ \right\rbrace \,.
     \label{eq:ut_thm}
\end{align}
The remaining quantities in the definitions of $ \epsMD,  \epsFA,  \epsAUE$ are defined as:
\begin{align}
    p(t, \hatt) &=\exp\left(-\frac{n}{2}E(t, \hatt)\right)\, , \quad  
     E(t, \hatt) = \max_{\rho,\rho_1\in[0,1]}\big[
     -\rho\rho_1   R_1(t, \hatt) -\rho_1 R_2(t)+E_0(\rho, \rho_1)\big] \label{eq:p_t,tHat_def}\, , \\
     &\quad\;\text{where}\quad  R_1(t, \hatt)=\frac{2}{n}\left(\tmin\ln M+ \ln\binom{\mc{R}}{\tmin}\right)\,,
     \qquad 
     \label{eq:R1_thm}\\
      & \quad \qquad\quad  \;\;\, \tmin = \min\{t, \hatt\}\,,  \quad \mc{R}  = L-\ka +\tmin - (\ukap-\ka)^+ -(\hatt-t)^+ \ge \tmin  \,,\label{eq:tmin_R_thm}\\
     & \quad \quad \qquad\;\; R_2(t) = \frac{2}{n} \ln \binom{\min\{\ka, \okap\}}{t}\,,\label{eq:R2_thm}\\
     &\quad\;\qquad\quad\;  E_0(\rho, \rho_1) =\max_{\lambda>-\frac{1}{P'\hat{t}}} \Big[\rho_1 a  (\rho ,\lambda)+\ln\big(1-\rho_1P_1 b(\rho ,\lambda) \big) \Big]\,,\label{eq:thm_E0}\\
     & \quad\;\qquad\quad\; a(\rho ,\lambda) =\rho \ln\big(1+P'\hatt \lambda\big) +\ln \big(1+P't\chi(\rho, \lambda)\big)\,, \label{eq:a_thm}\\
     &\quad\;\qquad\quad\;  b(\rho ,\lambda) = \rho\lambda -\frac{\chi(\rho, \lambda)}{1+P't\chi(\rho, \lambda)}\,, \qquad \chi(\rho, \lambda) = \frac{\rho \lambda}{ 1+P'\hatt \lambda}\,,\label{eq:b_mu_thm}\\
     & \quad\;\qquad\quad\;  P_1 =\left( (\ka-\okap)^+ + (\ukap-\ka)^+\right)P'+1 \,,\label{eq:P1_thm} \\
    \xi(\ka, \kap) &= \min_{\kappa\in[\kl:\ku]\setminus\kap}\left\lbrace\ind\{\kappa < \kap\}\frac{\Gamma(\frac{n}{2}, \zeta)}{\Gamma(\frac{n}{2})} + \ind\{\kappa > \kap\}\frac{\gamma(\frac{n}{2}, \zeta)}{\Gamma(\frac{n}{2})}\right\rbrace \,, \label{eq:xi_thm}\\
  &\quad\; \text{where}\quad   \zeta  =\frac{n}{2(1+\ka P')}\ln \left(\frac{1+ \kappa P'}{1+\kap P'}\right) \left[\frac{1}{1+\kap P'} - \frac{1}{1+\kappa P'}\right]^{-1} \,,  \label{eq:zeta_thm}\\
     %
  &\nu(\tmin,\psi)=
  \binom{\mc{R}-\tmin}{\psi}M^\psi \cdot \binom{\tmin}{\tmin-\psi}(M-1)^{\tmin-\psi} ,\quad \psi \in [0: \psiu]\,,\\
    &\nu(\tmin) = \sum_{\psi=0}^{\psiu} \nu(\tmin,\psi) \,,  \quad\;\psiu =\min\left\lbrace\tmin,\mc{R}-\tmin\right\rbrace \, .\label{eq:psiu_thm} 
     \end{align}
\end{thm}
\begin{proof}
The proof is similar to the proof of the  bounds for unsourced random access in \cite[Theorem 1]{ngo2022unsourced}, involving union bounds over error events via a change of measure and applications of Chernoff bound and Gallager's $\rho$-trick. \edit{The proof is detailed in Section~\ref{sec:ngo_full_proof}, where} we present the preliminaries for the proof in Section \ref{sec:preliminaries}, followed by  the proof  for a simple special case in Section \ref{sec:special_case}, and the general case   in Section \ref{sec:general_case}.

\edit{The key additional technical step in our setting is the careful separation of error events into the three categories defined in \eqref{eq:def_my_pMD},\eqref{eq:def_my_pFA}, and \eqref{eq:def_my_pAUE}, rather than only considering $\MD$s and $\FA$s as in \cite{ngo2022unsourced}. This separation of error events is illustrated in Figs.~\ref{fig:special_case_venn} and \ref{fig:general_case_venn} and mathematically formalized in  Lemmas  \ref{lem:cond_p_aMD} and \ref{lem:cond_expectation_general_case}. Since users have distinct codebooks in our setting, the number of codewords in each category is different from \cite{ngo2022unsourced} (see e.g., Remark \ref{rmk:ours_vs_ngo}), leading to different expressions for $\epsMD, \epsFA$ and $\epsAUE$.}
\end{proof}

 \paragraph*{Remarks}
    \begin{enumerate}
        \item \label{rmk:poly_vs_ours}
    Theorem \ref{thm:ngo_ach} reduces to the achievability bound in \cite[Section IV]{polyanskiy2017perspective} by Polyanskiy when all users are active and the  number of  users  is known, i.e.,  $\Ka=\Kap=L$ with probability 1 and the decoding \edit{radii are $\rl=\ru=0$}. In this case, all errors are $\AUE$s. 
    \item For the tightest possible achievability bounds, $P'$ should be optimized over the interval $(0, P)$ using methods such as the golden section search \cite{kiefer1953sequential, Avriel1966minimax}. To reduce the computational cost of the bounds, 
    we may choose a fixed value such as $P' = 0.8P$ 
    instead.   Fig.~\ref{fig:ngo_diff_P1factor} \edit{and surrounding text} in Section~\ref{sec:ngo_numerical}  demonstrates  how the choice of $P'$ might affect the  bounds.
    \item Given $\rho, \rho_1, t,\hatt$, the optimal $\lambda$ in \eqref{eq:thm_E0} admits  a closed-form solution, same as  the one given in  Remark vii) following  Theorem 1 in \cite{ngo2022unsourced}. 
    \item 
    An alternative achievability  bound may be derived using information densities, analogous to \cite{polyanskiy2017perspective, ngo2022unsourced}. However, such bounds are very expensive to evaluate numerically and are expected to be largely dominated by the error-exponent-type  bounds in Theorem \ref{thm:ngo_ach}.
    \end{enumerate}

\subsection{Error Floors}\label{sec:error_floors}
As before, we take $N_0=2\sigma^2=2$ and control the signal-to-noise ratio $E_b/N_0=nP/(2k)$ through $P$.  The bounds for $\epsFA, \epsMD, \epsAUE$ in \eqref{eq:thm_epsFA}, \eqref{eq:thm_epsMD}, \eqref{eq:thm_epsAUE} are each the sum of $\tp$ and another term. 
      Corollary \ref{cor:error_floors} below  highlights  a key observation: the second term in  $\epsMD$ and $\epsFA$ does not vanish  when $P\to \infty$, unless we choose   large enough decoding \edit{radii  $\rl,\ru$}  that grow with the problem size, e.g.,  \edit{$\rl=\ru=L$}. In contrast, the second term $\epsAUE$ vanishes as $P\to \infty$. We give a heuristic discussion of these error floors before stating the result.

    We note that $\tilde{p}$ in \eqref{eq:tp_def} reduces to  $\prob\big(\Ka\notin [\kl:\ku]\big)$ as $P\to \infty$, and that the optimal value of $P'$ that  minimizes the error probabilities   should  grow with $P$, i.e., $P'\to \infty$. 
    As $P'\to \infty$,   our maximum-likelihood decoder in \eqref{eq:min_distance_decoder} is guaranteed to decode the transmitted set of codewords correctly, provided  the number of active users  $\Ka$ falls within the decoder's search range $[\uKap: \oKap]$, where we recall $\uKap=\max\{\kl,\Kap-\rl\}$ and $
  \oKap=\min\{\ku, \Kap+\ru\}$. 
The term $\tp$ in the error bounds  accounts for the bad event where $\left\lbrace\Ka\notin[\kl: \ku]\right\rbrace$. Conditioned on  the good event that  $\left\lbrace\Ka\in[\kl: \ku]\right\rbrace$, we can ensure that  $\Ka\in[\uKap: \oKap]$  by selecting  large enough decoding \edit{radii}, e.g., \edit{$\rl=\ru=L$}, such that $\uKap=\kl$ and $ \oKap=\ku$. In this case, the error probability bounds in Theorem \ref{thm:ngo_ach} reduce to $\epsMD, \epsFA, \epsAUE = \tilde{p}$.

    However, choosing such  large decoding \edit{radii $\rl, \ru$}  make  the bounds   infeasible to compute for  even moderately large $L$.  With   fixed \edit{$\rl, \ru$} that do not scale with $L$,  the search range $ [\uKap: \oKap] $  may not contain     $\Ka$, and therefore as $P\to \infty$, 
     \begin{itemize}
         \item if $\Ka< \uKap$, the decoder  commits at least $ (\uKap-\Ka)$ $\FA$s; 
         \item if $ \Ka> \oKap$, the decoder commits at least $ (\Ka-\oKap)$ $\MD$s.
     \end{itemize}  
     This implies that the bounds $\epsMD$ and $\epsFA$ in Theorem \ref{thm:ngo_ach} exhibit error floors. In contrast, $\epsAUE$ vanishes as $P\to\infty$. \edit{This limiting behaviour of $\epsMD$ and $\epsFA$ matches that found in \cite[Corollary 1]{ngo2022unsourced}, and we recall that $\epsAUE$ does not feature in the unsourced set-up in \cite{ngo2022unsourced}.} These results are summarized below. 
 \begin{cor}[Error floors]\label{cor:error_floors}
 Let $\bar{p}=\P(\Ka\notin [\kl: \ku])$, then  the error bounds $\epsMD, \epsFA$ and $\epsAUE$ in \eqref{eq:thm_epsMD}--\eqref{eq:thm_epsAUE} satisfy 
 \begin{align}
     \lim_{P\to\infty}\epsMD&\ge \barepsMD= \sum_{\ka=\kl}^{\ku}p_{\Ka}(\ka)\sum_{\kap=\kl}^{\ku} \frac{(\ka-\okap)^+}{\ka} \xi(\ka, \kap)+ \bar{p}\ , \label{eq:error_floor_MD}\\
      \lim_{P\to\infty}\epsFA&\ge \barepsFA= \sum_{\ka=\kl}^{\ku}p_{\Ka}(\ka)\sum_{\kap=\kl}^{\ku} \frac{(\ukap-\ka)^+}{\ka - (\ka-\okap)^+ + (\ukap - \ka)^+} \xi(\ka, \kap)+ \bar{p}\ , \label{eq:error_floor_FA}\\
      \lim_{P\to\infty}\epsAUE&\ge \barepsAUE=\bar{p}\ . \label{eq:error_floor_AUE}
 \end{align}
  Note that $\barepsMD = \barepsFA=\bar{p}$ in \eqref{eq:error_floor_MD} and \eqref{eq:error_floor_FA} when \edit{$\rl, \ru$} are chosen to be sufficiently large  such that $\uKap=\kl$ and $ \oKap=\ku$.
 \end{cor}
 \begin{proof} The proof is similar to the proof of  \cite[Corollary 1]{ngo2022unsourced} and is provided in  Section \ref{sec:proof_error_floors}. 
 \end{proof}
 
\subsection{Numerical Results}\label{sec:ngo_numerical}
    \begin{figure}[t!]
        \centering
\subfloat[$\rl=\ru=0$, 1 or 2.]{\includegraphics[width=0.33\linewidth]{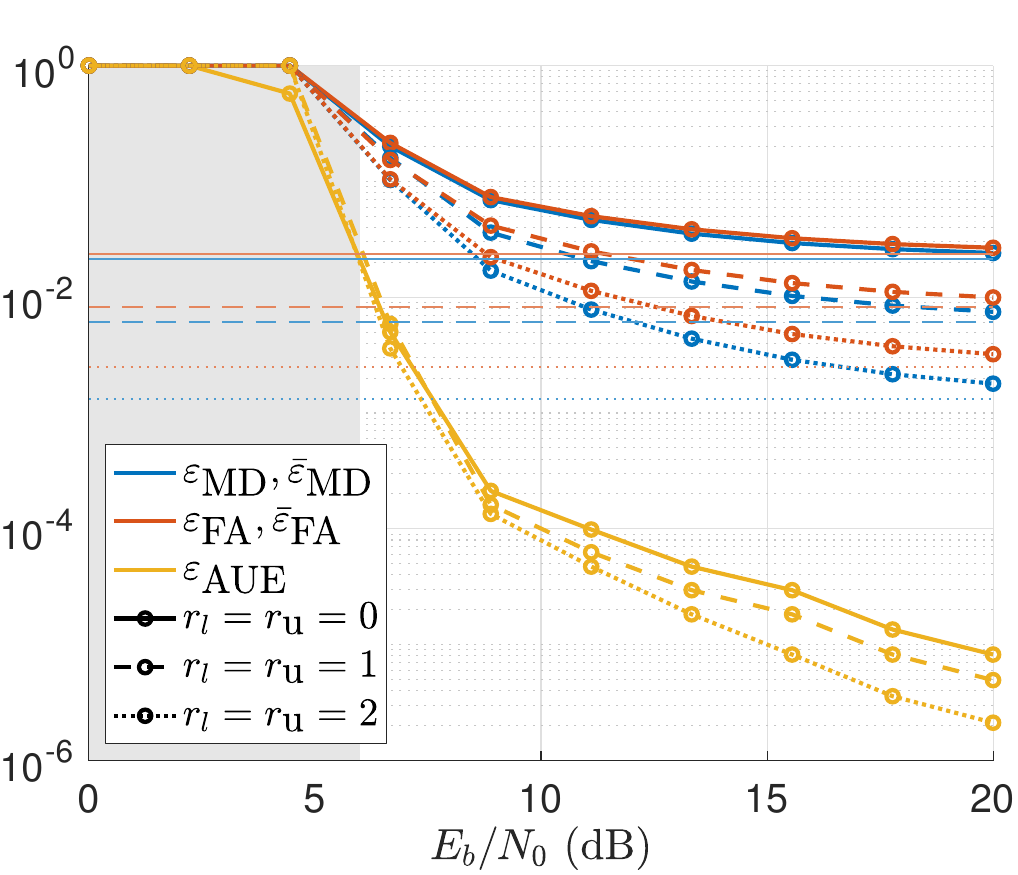} \label{fig:ngo_r0_1_2}}
        \subfloat[$\rl=\ru=0$ or $L$.]{\includegraphics[width=0.33\linewidth]{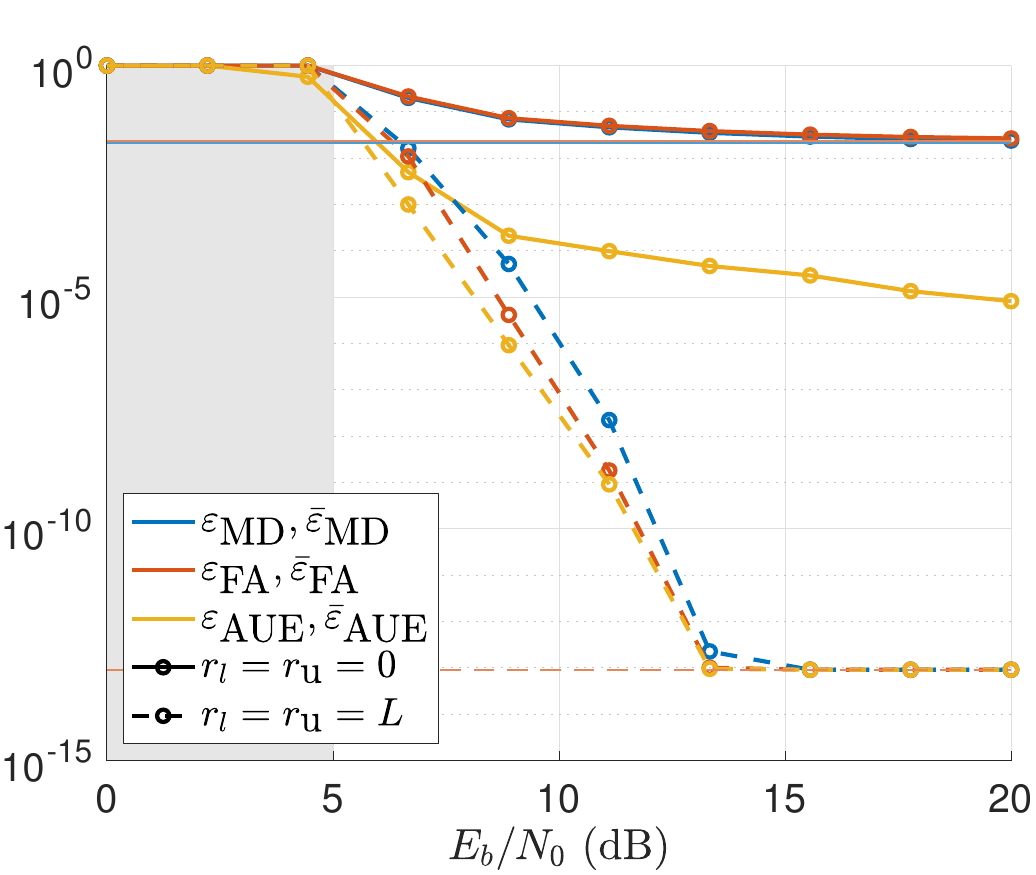} \label{fig:ngo_r_inf}}
        \subfloat[\edit{$\rl=0$, $\ru=2$, or $\rl=\ru=L$}.]{\includegraphics[width=0.33\linewidth]{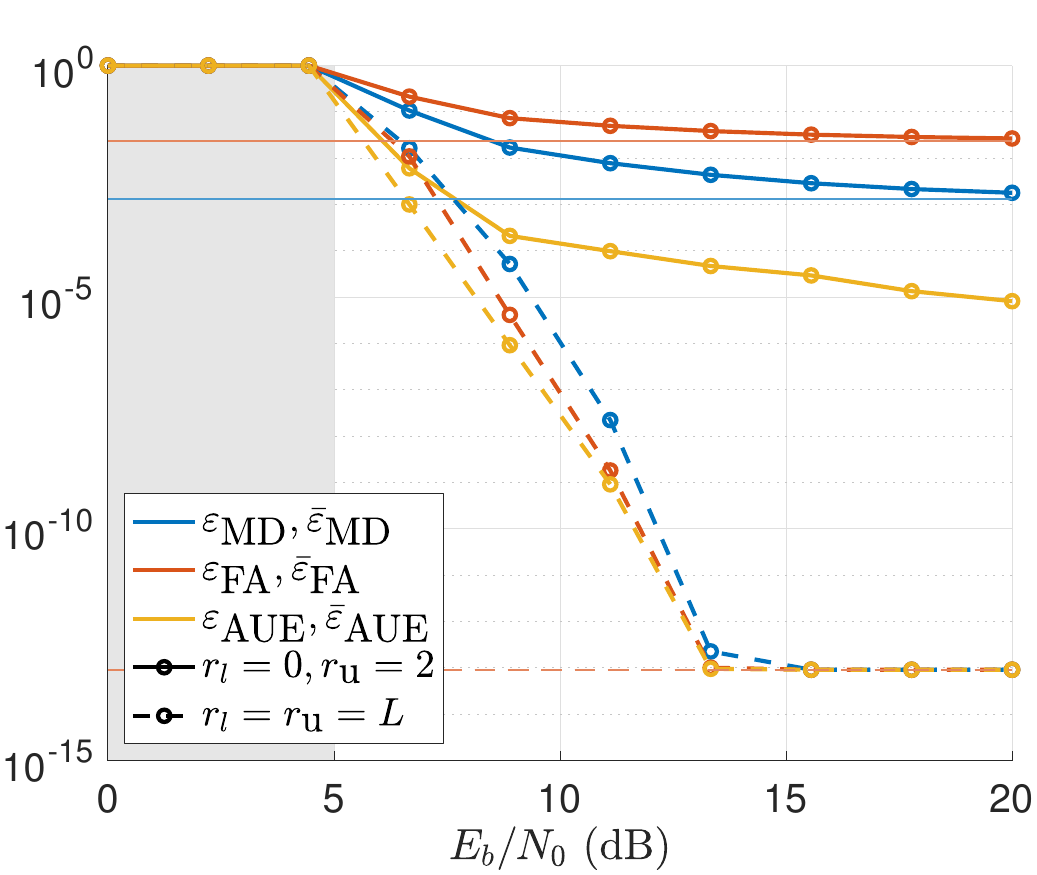} \label{fig:ngo_r_asymmetric}}
        \vspace*{-0.2cm}
        \caption{Three types of errors $\epsMD, \epsFA$ and $\epsAUE$ in Theorem \ref{thm:ngo_ach} (thicker curves) and the error floors $\barepsMD, \barepsFA$ and $\barepsAUE$ in Corollary \ref{cor:error_floors} (thinner horizontal lines)  plotted against $E_b/N_0$, with different choices of decoding \edit{radii $\rl, \ru$}. Different colours correspond to different types of errors; different line styles correspond to different  \edit{$\rl, \ru$}. Shaded regions: larger \edit{$\rl, \ru$} cause higher errors due to noise overfitting.   $n=2000,L=50, k=8 $, $\alpha=0.5$, $p_{\Ka}=\text{Bin}(\alpha, L),$ $P'$ is optimized over $ (0,P)$, and $(\kl,\ku)$ are   chosen so that $\P(\Ka\notin [\kl: \ku]) \le 10^{-13}$.   }
        \label{fig:ngo_ach_vary_r}
    \end{figure}
   
    We numerically evaluate the error  bounds $\epsMD, \epsFA$ and $\epsAUE$ in Theorem \ref{thm:ngo_ach}, and the error floors $\barepsMD, \barepsFA$ and $\barepsAUE$ of these error bounds in Corollary \ref{cor:error_floors} as $E_b/N_0\to\infty$. \edit{While Theorem \ref{thm:ngo_ach} and Corollary \ref{cor:error_floors} hold for any arbitrary probability distribution $p_{\Ka}$, 
    we assume  $p_{\Ka}$ to be the Binomial distribution 
    $\text{Bin}(\alpha, L)$ in our numerical experiments throughout this section.}
     We use  $\max\{\pMD, \pFA\} + \pAUE$ as   \edit{an overall error metric}, 
    since $\MD$s and $\FA$s largely exhibit a one-to-one correspondence. (This is easily seen in 
    the special case where  $\Ka=\hKa$, since each active user mistakenly declared silent corresponds to a silent user mistakenly declared active, creating an exact  one-to-one correspondence between the $\MD$s and $\FA$s.)
      Our  implementation of the bounds in Theorem \ref{thm:ngo_ach} and Corollary \ref{cor:error_floors} is adapted from the MATLAB code provided in \cite{ngo2022unsourced}
     and is  available at \cite{liu_ach_github}. 

    Fig.~\ref{fig:ngo_ach_vary_r} plots the error bounds $\epsMD, \epsFA$ and $\epsAUE$ as thick curves and the error floors $\barepsMD, \barepsFA$ and $\barepsAUE$ as thin horizontal lines.    We observe that $\epsMD, \epsFA$ and $\epsAUE$ decrease with $E_b/N_0$, and $\epsMD$ and $\epsFA $ are   lower bounded by and converge to the error floors $\barepsMD$ and $\barepsFA$ as $E_b/N_0$ increases.  The error floor  of $\epsAUE$ is $\barepsAUE=10^{-13}$, which is omitted from Fig.~\ref{fig:ngo_r0_1_2}. 
    The different line styles  correspond to different choices of  decoding \edit{radii $\rl, \ru$}. 
    Recall that  larger \edit{$\rl$ and  $\ru$} make the maximum-likelihood  decoder in \eqref{eq:min_distance_decoder} search over a larger set of values for $\hKa$. 
%
    We observe that while   larger \edit{$\rl$ and $\ru$} yield fewer errors at higher $E_b/N_0$, it can result in  higher errors at lower $E_b/N_0$ due to \emph{noise overfitting} (see shaded regions in the subfigures), \edit{where the decoder's larger search space increases the likelihood of mistakenly fitting to the channel noise realization rather than to the signal.} This phenomenon was also highlighted in \cite{ngo2022unsourced}.  
    
    Fig.~\ref{fig:ngo_r_inf} compares the errors when \edit{$\rl=\ru=0$ or $L$}, where \edit{$\rl=\ru=L$} makes the decoder  search for $\hW$ through all sets of codewords with size $\Ka\in [\kl: \ku]$. Note that with \edit{$\rl=\ru=L$}, $\epsMD, \epsFA$ and $\epsAUE$ all converge to the error floor $10^{-13}$ for large $E_b/N_0$, as expected due to  Corollary~\ref{cor:error_floors}. \edit{Fig.~\ref{fig:ngo_r_asymmetric} presents the asymmetric case where $\rl=0$ and $\ru=2$. This  biases the decoder toward  declaring  users active rather than  silent, resulting in  lower $\pMD$ than $\pFA$. Analogously, one can control $\pFA$ more strictly than $\pMD$ by choosing $\rl>\ru$.} \edit{In Fig.~\ref{fig:mua_crit_EbN0_alpha=0.7_diff_k}  in Section~\ref{subsec:num_results_CDMA}, we  plot the bound as   the minimum  $E_b/N_0$ required to achieve a target total error of $\max\{\pMD, \pFA\}+\pAUE\le 0.01$ for  a range of   active user densities $\mu_\a$.}

    The error floors in Corollary \ref{cor:error_floors} are  useful in informing our choice of parameters  such as \edit{$\rl, \ru$}, $\kl, \ku$ and $P'$. This is   explored  in Figs.~\ref{fig:ngo_diff_tail}--\ref{fig:ngo_diff_P1factor}, which   present  the contour plots of the total error floor $\max\{\barepsMD, \barepsFA\}+\barepsAUE$ for varying active user density $\mu_\a=\alpha L/n$ ($x$-axis), and different choices of decoding \edit{radii $\rl=\ru=r$} normalized by the standard deviation of $\Ka$ ($y$-axis). \edit{The contour plots give insight into how the performance of the maximum-likelihood decoder depends on $\rl, \ru, \kl, \ku$ and $P'$. 
    We discuss this below, noting that the shape of the plots may also depend on the (possibly suboptimal) bounding techniques in our analysis.}     
    
    In \edit{Fig.~\ref{fig:ngo_diff_tail}},
    $ \kl$ and $\ku$ are chosen to be the largest and smallest integers, respectively, so that $\P(\Ka\notin [\kl: \ku])\le \bar{p}$. 
    Fig.~\ref{fig:ngo_diff_tail} illustrates the effect of different choices  of $\bar{p}$. For  smaller values of $\mu_\a$, e.g., $\mu_\a=0.112$ (marked by the left gray vertical line in each subfigure),  a smaller $\bar{p}$ thus a larger search range can achieve a smaller error floor. Indeed, observe that  while choosing $\bar{p}=10^{-3}$ yields  an error floor as low as $0.02$, reducing  $\bar{p}$ to $10^{-7}$ further lowers the error floor to $0.01$. 
    In contrast, larger values of $\mu_\a$ correspond to a more challenging decoding task,  where the decoder observes increased interference among active users and attempts to recover a higher user payload. In such cases, choosing a larger $\bar{p}$ prevents the maximum-likelihood decoder from overfitting to noise. Indeed,  at  $\mu_\a=0.319$ for instance (marked by the right gray vertical line in each subfigure),  choosing $\bar{p}=10^{-3}$ instead of  $\bar{p}=10^{-7}$ leads to a significantly lower error floor.

 \begin{figure}[t!]
        \centering
        \subfloat[$\bar{p}= 10^{-3}$]{\includegraphics[width=0.51\linewidth]{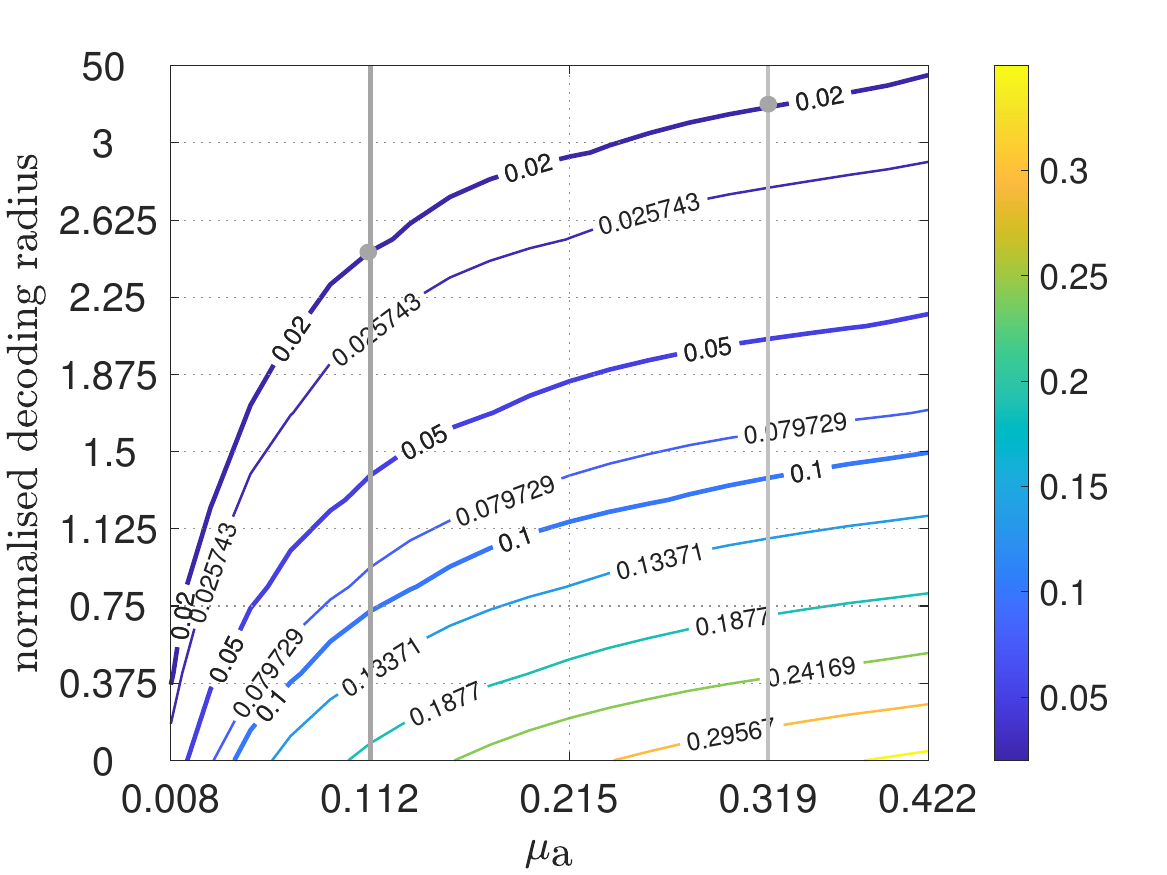}\label{fig:ngo_diff_tail_1e-3}}
        \subfloat[$\bar{p}= 10^{-7}$]{\includegraphics[width=0.49\linewidth]{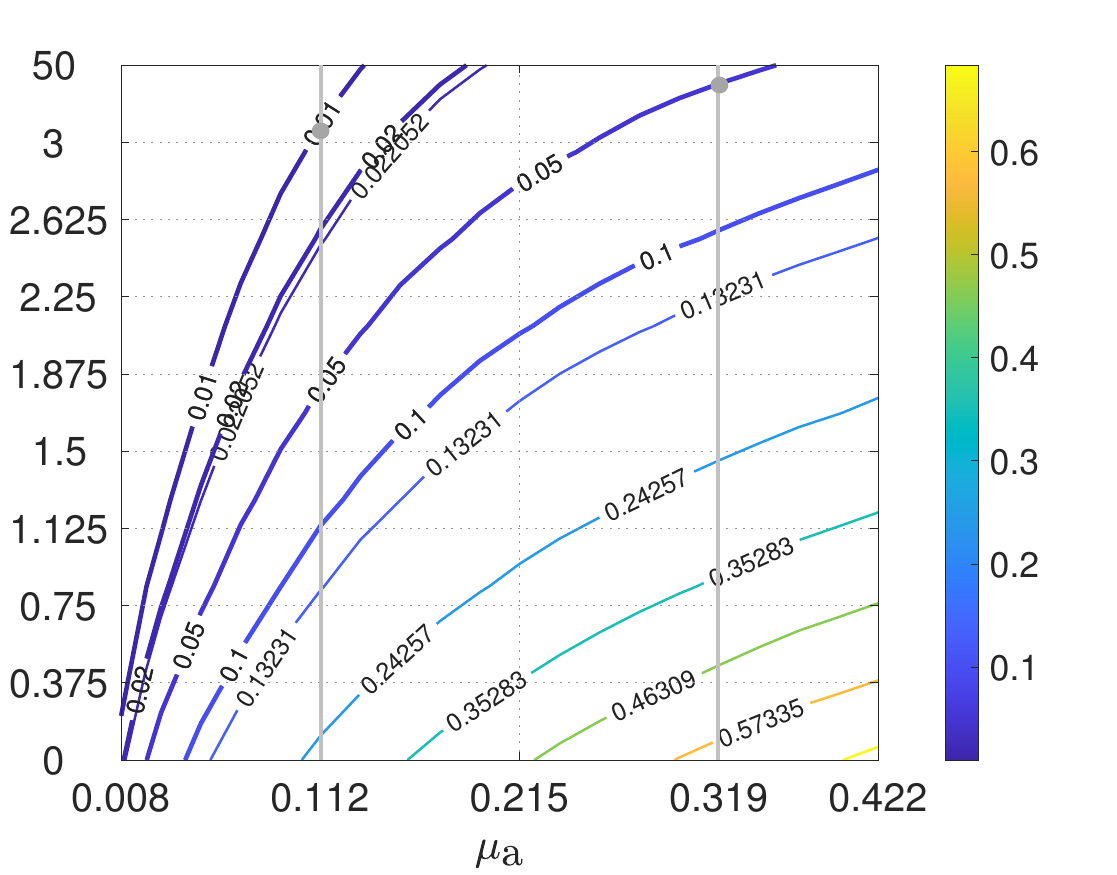}\label{fig:ngo_diff_tail_1e-7}}
        \vspace*{-0.2cm}
        \caption{Contour plots of $\max\{\barepsMD, \barepsFA\}+\barepsAUE$ for different active user density $\mu_\a$ ($x$-axis) and   normalized decoding radius $r/\std(\Ka)$   ($y$-axis) \edit{where $\rl=\ru=r$}.   $(\kl,\ku)$ are chosen to be the largest and smallest integers  so that $\P(\Ka\notin [\kl: \ku]) \le \bar{p}$, where $\bar{p}$ is indicated below the subfigures.  $L=600, k=6, \alpha=0.7$, $p_{\Ka}=\text{Bin}(\alpha, L)$, $P'=0.8P$, $r/\std(\Ka)=50$ corresponds to $r=L$.}
        \label{fig:ngo_diff_tail}
    \end{figure}
    \begin{figure}[t!]
        \centering
        \subfloat[$P'=0.8P$]{\includegraphics[width=0.52\linewidth]{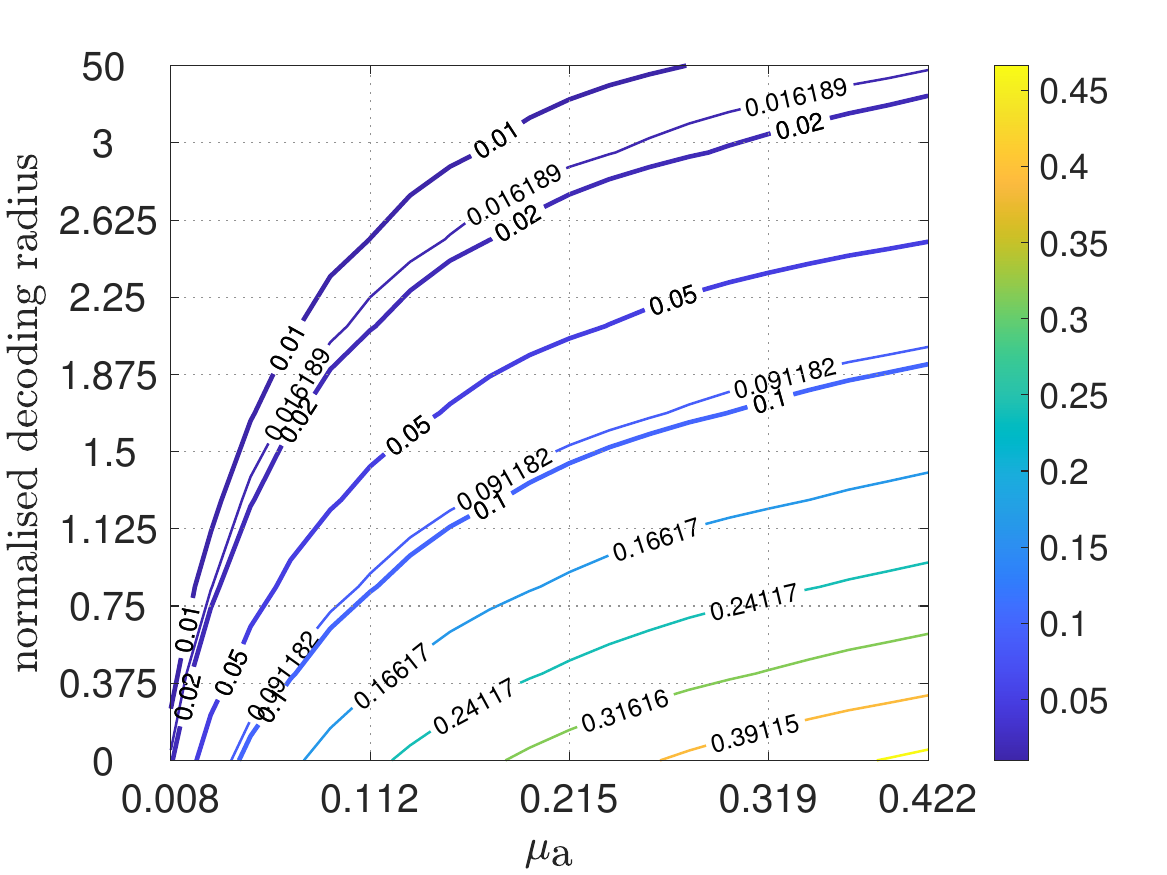}}
        \subfloat[$P'=0.9P$]{\includegraphics[width=0.47\linewidth]{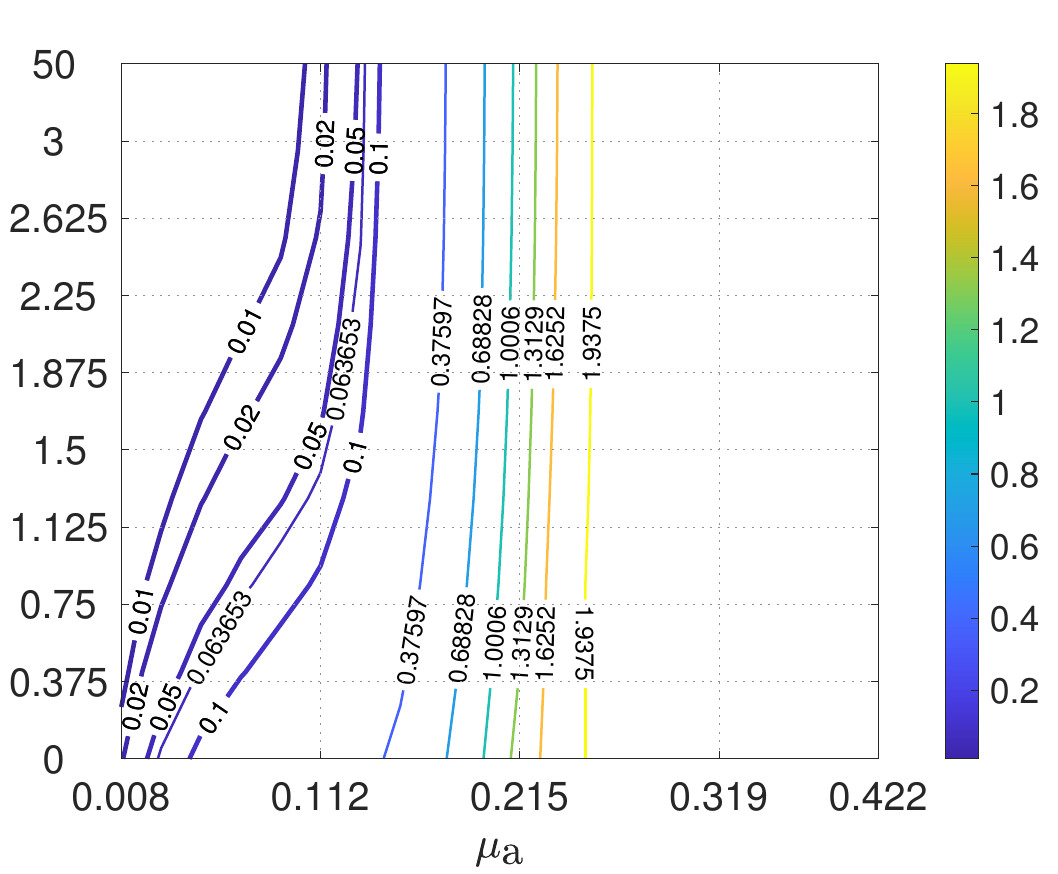}}
        \vspace*{-0.2cm}
        \caption{Contour plots of $\max\{\barepsMD, \barepsFA\}+\barepsAUE$ for different active user density $\mu_\a$ ($x$-axis) and normalized decoding radius $r/\std(\Ka)$ ($y$-axis) \edit{where $\rl=\ru=r$}.   Subfigures use different $P'$, indicated below the subfigures.  $L=600, k=6, \alpha=0.7$, $p_{\Ka}=\text{Bin}(\alpha, L)$, $r/\std(\Ka)=50$ corresponds to $r=L$, and $(\kl,\ku)$ are chosen   so that $\P(\Ka\notin [\kl: \ku]) \le 10^{-4}$.  }
        \label{fig:ngo_diff_P1factor}
    \end{figure}
    
    Fig.~\ref{fig:ngo_diff_P1factor}   compares different choices of $P'\in(0,P)$  with $\bar{p}$ fixed to be $10^{-4}$. We observe that the error floor can be significantly lowered by  choosing $P'$ not too close to $P$, especially  for larger values of $\mu_\a$. Nevertheless,  our experiments indicate that the optimal value of $P'$ varies on a case-by-case basis; it tends to be closer to $P$ for larger $L$ and $n$. Hence, whenever computational resources permit, such as in Fig.~\ref{fig:ngo_ach_vary_r}, we optimize $P'$ over $(0,P)$ using the golden section search method.

\paragraph{Computational cost}
    A  shortcoming of the  finite-length bounds  in Theorem \ref{thm:ngo_ach}  is their high  computational complexity, which is $\mc{O}(L^5)$ due to the five summations in each of the equations  \eqref{eq:thm_epsMD}--\eqref{eq:thm_epsAUE}. This  limits their applicability to relatively small values of $ L$ and $n$. In the next section (Section \ref{sec:asymp_ach}), we derive an  asymptotic bound which is significantly cheaper to evaluate,  with  computational complexity independent of $L$ or $n$.

\section{Asymptotic Random Coding Achievability Bounds}\label{sec:asymp_ach}
In this section,  we  provide  two  asymptotic bounds on $\pMD, \pFA$ and $\pAUE$ in the limit of $n,L\to\infty$ with $\alpha = \E[\Ka]/L$ and $\mu_\a = \E[\Ka]/n$  held constant. The first bound is tighter but hard to compute for large payloads, whereas the second is looser but can be computed efficiently for large payloads.

    \subsection{Random   Coding}\label{sec:sparc_scheme_for_asymp_ach}
    We consider a   random coding scheme similar to that in Section \ref{sec:ngo_ach_scheme} under AMP decoding instead of maximum-likelihood decoding. As before,  the codebook  of each user $\ell\in[L]$ is populated  with $M=2^k$ random codewords that satisfy the energy-per-bit constraint. When active, each user selects a codeword from their codebook uniformly at random for transmission.    For simplicity, we consider the prior where each user $\ell$ is independently active with probability $ \alpha$, i.e., $p_{\Ka}=\text{Bin}(\alpha, L)$, but the bounds can be adapted to other priors as well.  

    The random coding scheme  can be formulated  in terms a linear model. Let $\bA_\ell\in \reals^{n\times M}$ be the matrix whose columns store the codewords  of user $\ell$, scaled by ${1}/{\sqrt{E_bk}}$. 
    This scaling ensures  the columns of $\bA_\ell$ have unit squared $\ell_2$-norm in expectation, which is required by the standard AMP framework. Therefore  the codeword transmitted by user $\ell$ can be written as 
    \begin{align}
        \bc_\ell = \bA_\ell \bx_\ell\,, 
    \end{align}
 where $\bx_\ell
 \stackrel{\text{i.i.d.}}{\sim} p_{\bbx_{\sec}}$  with 
    \begin{align}\label{eq:px_sec}
        p_{\bbx_{\sec}} = (1-\alpha) \delta_{\bzero} + \alpha p_{\bbx_{\sec, \a}}\,, 
         \quad p_{\bbx_{\sec, \a}} = \frac{1}{M}\sum_{j=1}^M\delta_{\sqrt{E}\be_j}\,,\quad E=E_b k.
    \end{align}
Here, $\delta_{\bzero}$ is the unit mass at 
$\bzero$, and $\be_j \in\reals^M$ denotes the  canonical basis vector with one in the $j$th coordinate   and zeros everywhere else. This implies that 
$\bx_\ell$ 
    is  all-zero when the  user is silent, or one-sparse when the  user is active, with the nonzero coordinate uniformly distributed in $\{1, \ldots, M \}$. The vector $\bx_\ell$ can be interpreted as the message vector of user $\ell$. 

       \begin{figure}[t!]
    \centering
    \hspace{2cm}
  \includegraphics[width=0.75\textwidth]{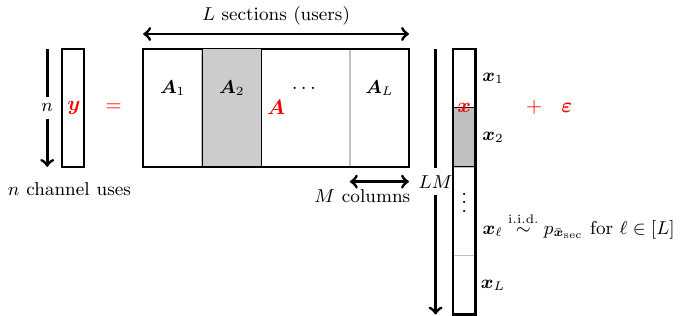}
    \caption{Coding scheme for the proof of  asymptotic achievability bounds. The joint message vector $\bx= [\bx_1^\top, \dots, \bx_L^\top]^\top $ has  $L$ sections corresponding to the $L$ users. Each section is drawn i.i.d.\@ from $p_{\bbx_{\sec}}$. The joint codebook matrix $\bA=[\bA_1, \dots, \bA_L]$ has $L$ sections, with the $\ell$-th section storing the codewords of the $\ell$-th user as columns of the matrix.}
    \label{fig:Y=AX+W_SPARC}
\end{figure}%

  Let $\bA = [\bA_1,  \dots, \bA_L] \in \reals^{n\times LM}$ and $\bx = [\bx_1^\top, \dots, \bx_L^\top]^\top\in \reals^{LM}$ be the concatenation of the random matrices and the message vectors of all users, respectively. Then the channel output from the GMAC in \eqref{eq:standard_gmac} takes the form
    \begin{align}\label{eq:Y=AX+W_SPARC}
        \by = \bA \bx + \bveps\in \reals^n,
    \end{align}
    where $\bveps\sim \mc{N}_n(\bzero, \sigma^2\bI)$ is the channel noise vector. See Fig.~\ref{fig:Y=AX+W_SPARC} for an illustration. 

The decoding task is to recover the concatenated message vector $\bx$ from $\by$ and $\bA$. This problem is similar to the decoding of Sparse Regression Codes (SPARCs) \cite{joseph2014sparc,barbier2017approximate, venkataramanan19monograph, rush2021capacity}, the key difference being that in SPARCs,  the prior $p_{\bbx_{\sec}}$ does not contain a mass at $\bzero$. Based on the similarity with SPARCs, we use a scheme with a  spatially coupled matrix $\bA$ and Approximate Message Passing (AMP) decoding to obtain asymptotic bounds on the three probabilities of error. Achievable regions for the many-user GMAC (without random access) were obtained using similar ideas in  \cite{hsieh2022near} and \cite{kowshik2022improved}.

\paragraph{Spatially coupled design} 
\begin{figure}[t!]
    \centering
    \includegraphics[width=0.55\textwidth]{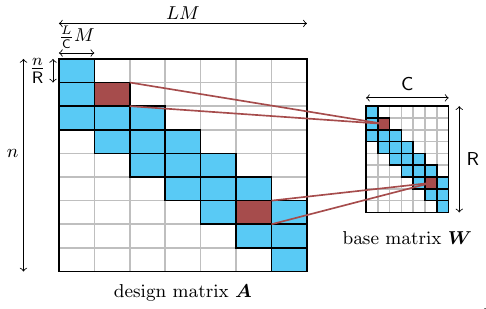}
    \vspace*{-0.2cm}
    \caption{A spatially coupled design matrix  $\bA$ constructed using a base matrix $\bW$ according to \eqref{eq:SC_design_Aij}. The base matrix shown here is an $(\omega, \Lambda)$ base matrix (defined in    Definition \ref{def:ome_lamb_rho}) with parameters $\omega=3, \Lambda=7$. The white parts of $\bA$ and $\bW$ correspond to zeros.}
    \label{fig:SC_sparc}
\end{figure}
Inspired by \cite{rush2021capacity,hsieh2022near,kowshik2022improved}, we choose $\bA\in\reals^{n\times LM}$ to be a spatially coupled  random Gaussian matrix  to ensure the tightest achievability bounds under AMP decoding. Fig.~\ref{fig:SC_sparc} illustrates an example of such a matrix.
 The matrix $\bA$ is divided into $\R$  \emph{row blocks} and $\C$ \emph{column blocks}, where $\R$ is chosen to divide $n$ and $\C$ to divide $L$. This implies that every column block contains $\frac{L}{\C}$ sections of $\bA$, corresponding to  $\frac{L}{\C}$ users. Let $\sfr(\cdot): [n] \to [\R]$ and $\sfc(\cdot): [LM]\to [\C]$ be the operators that map a particular row or column index into its corresponding row block or column block index. Then for $i\in [n]$ and $j\in [LM]$, the entries of $\bA$ are drawn according to
\begin{align}
    A_{ij}\stackrel{\text{i.i.d.}}{\sim} \normal\left(0, \frac{1}{n/\R} W_{\sfr(i), \sfc(j)}\right), \quad \text{for} \; i\in[n],\  j\in [LM].\label{eq:SC_design_Aij}
\end{align}
Here, $\bW \in \reals^{\R\times \C}$ is a  \emph{base matrix},  whose columns satisfy $\sum_{\sfr=1}^\R W_{\sfr,\sfc}=1$ for  $\sfc\in [\C]$. This implies that the columns of $\bA$ have unit squared $\ell_2$-norm in expectation. By concentration of measure, it follows that  $\|\bc_\ell\|_2^2/k =\|\bA_\ell\bx_\ell\|_2^2/k  \to E_b$ as $n\to\infty$, indicating that  the energy-per-bit constraint is asymptotically satisfied. The base matrix $\bW$ we use to obtain our bounds is defined in Definition \ref{def:ome_lamb_rho} below through two parameters $\omega$ and $\Lambda$.   
The standard i.i.d.\@ Gaussian design where  $A_{ij}\stackrel{\text{i.i.d.}}{\sim}\mc{N}(0, 1/n)$ is  a special case of the spatially coupled design, obtained by using a base matrix with a single entry $W=1$ ($\Lr=\Lc=1$).

\begin{definition}
    \label{def:ome_lamb_rho}
    An $(\omega , \Lambda)$ base matrix $\bW\in \reals^{\R\times \C}$ is described by two parameters: the coupling width $\omega\geq1$ and the coupling length $\Lambda\geq 2\omega-1$. The matrix has $\Lr=\Lambda+\omega-1$ rows and $\Lc=\Lambda$ columns, with each column having $\omega$ identical nonzero entries in the band-diagonal and zeros everywhere else.  For $\sfr \in [\Lr]$ and $ \sfc\in[\Lc]$, the $(\sfr,\sfc)$th entry of the base matrix is given by
    \begin{equation}
    \label{eq:W_rc}
    W_{\sfr,\sfc} =
    \begin{cases}
            \ \frac{1}{\omega} \quad &\text{if} \ \sfc \leq \sfr \leq \sfc+\omega-1,\\
        \ 0 \quad &\text{otherwise}.
    \end{cases}
    \end{equation}
    \end{definition}

Each nonzero block of the spatially coupled  matrix $\bA$ can be viewed as an (uncoupled) i.i.d.\@  matrix with $\frac{L}{\C}$ users and $\frac{n}{\R}$  channel uses. Hence we define the \emph{inner} user density to be 
\begin{align}
    \muin = \frac{L/\C}{n/\R} = \frac{\R}{\C}\mu = \left(1 + \frac{\omega-1}{\Lambda}\right)\mu.
    \label{eq:mu_indef}
\end{align}
Note that since $\omega\ge 1$, we have $\muin > \mu$. The difference between $\muin$ and $\mu$  becomes negligible when $\Lambda\gg \omega$. This relation will be used later  in the asymptotic performance analysis  of AMP decoding with the spatially coupled design matrix $\bA$.

\edit{The benefit of using a spatially coupled design compared to a standard i.i.d.\@ matrix is due to its band-diagonal structure, as shown in Fig.~\ref{fig:SC_sparc}. Each row in $\bA$  corresponds to a channel transmission, while each group of $M$ columns corresponds to the codebook of a user. (Columns $1$ to $M$ correspond to user 1, columns $(M+1)$ to $2M$ correspond to user 2, and so on.)  For illustration, suppose that each of the blocks in $\bA$ in Fig.~\ref{fig:SC_sparc} corresponds to 5 users and 5 transmissions. As there is a single non-zero block in the top row block of $\bA$, the first 5 transmissions (corresponding to the first row block) will only involve interference from the first 5 users and therefore the first 5 messages are easier to decode. As the next row-block contains two non-zero blocks, the next 5 transmissions will involve 10 users, but if the first 5 users' messages have been decoded, the next 5 messages will also be easier to decode. This generates a \textit{snowball} effect, where as the iterations progress, the adjacent users' messages present less interference. Gradually the full set of messages can be decoded, resulting in a gain in performance compared to the standard i.i.d.\@ Gaussian scheme. An illustration of the decoding progression can be found in \cite[Fig.~2.2]{hsieh2021thesis}.} 

\subsection{AMP Decoding and State Evolution}\label{sec:AMP_SE_asymp_ach}
The channel decoding task is to estimate the joint message vector $\bx$ given  the channel output $\by$, the design matrix $\bA$ and the channel noise variance $\sigma^2$.  We use an AMP algorithm  for decoding, similar to those in  \cite{hsieh2022near,kowshik2022improved}. Starting with the initial estimate $\bx^0=\bzero$ of the joint message vector $\bx$,    for $t\ge 0$, the AMP decoder recursively computes:
    \begin{align} 
        &\bz^t  = \by - \bA \bx^t+\tbv^t \odot\bz^{t-1} , \label{eq:AMP_sparc_z}\\
        &\bx^{t+1}  = \eta_t\left(\bs^t\right),\quad \text{where} \quad \bs^t  =\bx^t + \left(\tbQ^t \odot \bA\right)^\top \bz^t .\label{eq:AMP_sparc_x}
    \end{align}
    Here $\eta_t: \reals^{LM} \to \reals^{LM}$ is a Lipschitz denoising function that produces an updated estimate $\bx^{t+1}$ of $\bx$ from  the  effective observation $\bs^t$. 
    \edit{The operator $\odot$ is the Hadamard (entrywise) product, and quantities with negative  iteration index are set to all-zero vectors which means $\bz^0 = \by - \bA\bx^0 = \by$}. 
   Before giving the exact form of the vector $\tbv^t\in\reals^n $, the matrix $\tbQ^t\in\reals^{n\times LM}$ and the denoiser $\eta_t$ in \eqref{eq:AMP_sparc_x}, we  explain the asymptotic distributional properties of the AMP algorithm via a deterministic recursion called \emph{state evolution}.

\paragraph{State evolution (SE)}
Recall that the joint  message vector $\bx$ consists of  $L$ length-$M$ sections $\bx_1, \dots, \bx_L$ drawn i.i.d.\@ from $p_{\bbx_{\sec}}$. Each section corresponds to a specific column block $\sfc\in[\C]$, so there are $L/\C$ sections per column block.  Let $\bs^t_1, \dots, \bs^t_L$ denote the $L$ sections of $\bs^t$. 
For each section $\ell$ located in block $\sfc$, the debiasing term $\tbv^t\odot \bz^{t-1}$ in \eqref{eq:AMP_sparc_x} ensures that  for $n, L\to\infty$ with $\mu=L/n\in(0, \infty)$, and $t\ge 1$, we have 
\begin{align}\label{eq:SC_sparc_eff_obs}
    \left(\bs^t_\ell - \bx_\ell\right) \sim \normal_{M}(\bzero, \tau_{\sfc}^t\bI)\ ,
\end{align}
where the noise variance $\tau_\sfc^t$ 
can be characterized through the following SE recursion.  Initialized with $\psi_\sfc^0=E_bk$ for $\sfc\in [\C]$, the SE iteratively computes for $t\ge 0$, $\sfr\in[\R]$ and $\sfc\in[\C]$:
\begin{align}
    &\phi_\sfr^t = \sigma^2 + \muin \sum_{\sfc=1}^\C W_{\sfr,\sfc}\psi_{\sfc}^t\ ,\label{eq:SC_SE_sparc_phi}\\
    &\tau_\sfc^t = \left[\sum_{\sfr=1}^\R \frac{W_{\sfr,\sfc}}{\phi_\sfr^t}\right]^{-1}\ ,\label{eq:SC_SE_sparc_tau}\\
    &\psi_\sfc^{t+1}  = \E \big\|\bbx_{\sec} - \eta_{t,\sfc}\left(\bbx_{\sec} + \bg_\sfc^t\right) \big\|_2^2\ ,  \quad \text{where}\quad \bg_{\sfc}^t\sim \normal_{M}(\bzero, \tau_\sfc^t\bI)\ , \label{eq:SC_SE_sparc_psi}
\end{align}
where $\eta_{t, \sfc}: \reals^{M}\to \reals^M$ denotes the denoising function  associated with  block $\sfc$ (see \eqref{eq:denoiser_sparc}). 
Using the parameters in the SE equations, we can define 
the vector $\tbv^t\in\reals^n $ and the matrix $\tbQ^t\in\reals^{n\times LM}$ in \eqref{eq:AMP_sparc_z}--\eqref{eq:AMP_sparc_x} as follows:
\begin{align}
    &
    \tv_i^t = \frac{\muin \sum_{\sfc=1}^{\C} W_{\sfr(i), \sfc}\psi_{\sfc}^t}{\phi_{\sfr(i)^{t-1}}},\quad 
    \tQ_{ij}^t =  \frac{\tau_{\sfc(j)}^t}{\phi_{\sfr(i)}^t}\label{eq:AMP_sparc_Q_v} \quad \text{for}\;  i\in [n],\; j\in [LM] \ .
\end{align}
The parameter $\psi_\sfc^{t+1}$  can be interpreted as the MSE for a section  located in block $\sfc$. \edit{Since each block in the spatially coupled matrix $\bA$ has different variance, the Hadamard product with $\bQ$ in \eqref{eq:AMP_sparc_x} takes this into account and adjusts the matrix $\bA$ accordingly.} 

\paragraph{Choice of denoiser $\eta_t$ in AMP}
The distributional characterization in \eqref{eq:SC_sparc_eff_obs} motivates applying the denoiser $\eta_t$ in \eqref{eq:AMP_sparc_x} section-wise to its input, i.e., for $t\ge 1$, we have 
    \begin{align}
        \eta_t(\bs^t) = \begin{bmatrix}
            \eta_{t,1}\left(\bs^t_1\right)\\
            \vdots\\
            \eta_{t,1}\left(\bs^t_{L/\C}\right)\\
            \vdots\\
            \eta_{t,\C}\left(\bs^t_{(\C-1)L/\C+1}\right)\\
            \vdots\\
            \eta_{t,\C}\left(\bs^t_{L}\right)
        \end{bmatrix}
        \setlength{\arraycolsep}{0pt} 
        \begin{array}{ c }
          \left.\kern-\nulldelimiterspace
          \vphantom{\begin{array}{ c }
              \eta_1  \\ 
            \vdots \\ 
              \eta_1 \\   
            \end{array}}
            \right\}\text{$\frac{L}{\C}$ sections with $\sfc=1$}\\\\
            \vphantom{\vdots} 
            \left.\kern-\nulldelimiterspace
            \vphantom{\begin{array}{ c }
              \eta_\C \\
              \vdots\\
                \eta_\C\\
          \end{array}}
          \right\}\text{$\frac{L}{\C}$ sections with $\sfc=\C$.}
        \end{array}\label{eq:denoiser_sparc}
    \end{align}
Here, for $\sfc\in[\C]$, $\eta_{t, \sfc}: \reals^{M}\to \reals^M$ denotes the denoiser for block $\sfc$, and is applied to each of the $\frac{L}{\C}$ sections located in block $\sfc$. We choose $\eta_{t,\sfc}$ to be either the Bayes-optimal estimator or the marginal-MMSE estimator  for each section of the joint message vector  $\bx$ in block $\sfc$. \edit{These two estimators are described next.}

Let $\mc{L}_\sfc\coloneqq \left\lbrace(\sfc-1){L}/{\C}+1, \dots, \sfc{L}/{\C}\right\rbrace$ for $\sfc\in [\C]$. Then  \eqref{eq:SC_sparc_eff_obs}  says that    section $\ell\in \mc{L}_{\sfc}$ of the effective observation vector in iteration $t\ge 1$, $\bs_\ell^t$, is asymptotically distributed as the  signal $\bx_\ell$ embedded in i.i.d.\@ Gaussian noise with  variance $\tau_{\sfc}^t$. 
Hence, the Bayes-optimal   denoiser  for block $\sfc$, denoted by $\eta_{t, \sfc}^\Bayes: \reals^M\to\reals^M$ and   applied to  $\bs_\ell^t$, is the MMSE estimator  taking the form:
\begin{align}
  \bx_\ell^{t+1} &=  \eta^\Bayes_{t, \sfc}(\bs^t_\ell) = \E\left[\bbx_{\sec}\big | \ \bbx_{\sec}+\bg_{\sfc}^t = \bs_\ell^t\right]\nonumber\\
    & =\sum_{j=1}^M\sqrt{E}\,\be_j \cdot \frac{\frac{\alpha}{M}\exp\left({(s_{\ell}^t)_j\sqrt{E}}/{\tau_\sfc^t} - E/(2\tau_\sfc^t)\right) }{ \frac{\alpha}{M}\sum_{j'=1}^M\exp\left((s_{\ell}^t)_{j'}\sqrt{E}/\tau_{\sfc}^t - E/(2\tau_\sfc^t) \right) +(1-\alpha)}\,,\label{eq:bayes_opt_eta_t_ach}
\end{align}
where we recall  $E=E_bk$ and  $\be_j\in\reals^M$ is the $j$-th canonical basis vector.  

A suboptimal alternative to the Bayes-optimal denoiser is the marginal-MMSE denoiser, denoted by $\eta_{t,\sfc}^\marginal: \reals^M\to \reals^M$. This denoiser acts entrywise on its input and computes the entrywise conditional expectation using the marginal distribution of $\bbx_{\sec}$ and $\bg_\sfc^t$:
\begin{align}
    \bx_\ell^{t+1}= \eta^\marginal_{t, \sfc}(\bs_\ell^t) = 
    \begin{bmatrix}
        \E\left[\bar{x}_{\sec,1} | \ \bar{x}_{\sec,1}  + g_{\sfc,1}^t = (s_\ell^t)_1\right]\\
        \vdots\\
        \E\left[\bar{x}_{\sec,M} | \ \bar{x}_{\sec,M}  + g_{\sfc,M}^t = (s_\ell^t)_M\right]    
    \end{bmatrix},\label{eq:marg_eta_t_ach}
\end{align}
with 
\[\E\left[\bar{x}_{\sec,j} | \ \bar{x}_{\sec,j}  + g_{\sfc,j}^t = (s_\ell^t)_j\right] =\sqrt{E}\cdot \frac{\frac{\alpha}{M}\exp\left((s_\ell^t)_j\sqrt{E}/\tau_{\sfc}^t - E/(2\tau_\sfc^t)\right)}{\frac{\alpha}{M} \exp\left((s^t_\ell)_j\sqrt{E}/\tau_{\sfc}^t - E/(2\tau_\sfc^t)\right) + \left(1-\frac{\alpha}{M}\right)}\, , \quad j\in[M].\]

In the  next subsection (Section \ref{sec:sparc_asymp_error_analysis}), we  will  characterize the asymptotic  performance of the AMP decoder with $\eta_{t,\sfc}^\Bayes$ and with $\eta_{t,\sfc}^\marginal$. While  $\eta_{t,\sfc}^\marginal$ leads to worse performance than $\eta_{t,\sfc}^\Bayes$, its asymptotic performance  can be computed much more efficiently than of $\eta_{t,\sfc}^\Bayes$.

\paragraph{Hard-decision step in AMP}
In each iteration $t\ge 1$, given the effective observation $\bs^t$, the AMP decoder can produce a hard-decision  estimate $\hbx^{t+1}$ of the joint message  vector $\bx$.    Let $h_{t, \sfc}: \reals^M\to \reals^M$ be the hard-decision function for block $\sfc\in[\C]$, which applies section-wise to the $\frac{L}{\C}$ sections in block $\sfc$. To minimize the probability of detection error, we choose $h_{t, \sfc} $ to be the MAP estimator, i.e., for  section $\ell\in \mc{L}_{\sfc}$,  we have 
\begin{align}\label{eq:hard_decision_sparc}
 \hbx_\ell^{t+1}=   h_{t, \sfc}(\bs_\ell^t) = \argmax_{\bx' \in \mc{X}} \P\left(\bbx_{\sec} = \bx' \mid  \bbx_{\sec}+ \bg_\sfc^t = \bs_\ell^t\right),
\end{align}
where $\mc{X} = \{\sqrt{E}\, \be_j, \, \forall j\in[B]\}\cup \bzero$  is the support of $p_{\bbx_{\sec}}$. We can then calculate the error probabilities $\pMD, \pFA$ and $\pAUE$ by comparing the final estimate $\hbx_\ell^t$ against the ground truth $\bx_\ell$,  for  $\ell\in[L]$. Note that the entrywise variance $\tau_{\sfc}^t$ of $\bg_{\sfc}^t$ depends on the choice of the sequence of denoising functions $\eta_{0, \sfc}, \dots, \eta_{t,\sfc}$, as stated in \eqref{eq:SC_SE_sparc_phi}--\eqref{eq:SC_SE_sparc_psi}. This implies that $h_{t, \sfc}$ uses a different value for $\tau_{\sfc}^t$, depending on whether we use the Bayes-optimal denoiser or the marginal-MMSE denoiser.


\subsection{Asymptotic Error Analysis for AMP Decoding}\label{sec:sparc_asymp_error_analysis}
  We characterize the asymptotic error probabilities $\pMD, \pFA$ and $\pAUE$ of the random  coding scheme in Section \ref{sec:sparc_scheme_for_asymp_ach} 
 with spatially coupled design matrix  under AMP decoding, in the limit as $n, L\to\infty$ with $L/n\to \mu\in(0, \infty)$. We consider two AMP decoders with $\eta_t^\Bayes$ or  $\eta_t^\marginal$ as denoiser.
To quantify their error probabilities, we need to characterize the fixed point of the AMP decoders in the limit  $t\to \infty$. We do this using the stationary points of  suitable   potential functions, following the approach in \cite{barbier2017approximate,yedla2014simple_scalar}. We will bound the 
  error probabilities for each decoder  in terms of the minimizer of a potential function. The potential function corresponding to $\eta_{t}^\Bayes$ is  similar to the one used in  \cite{hsieh2022near}, and the potential function for  $\eta_t^\marginal$ is similar to that in \cite{kowshik2022improved}. 

\paragraph*{Potential function for AMP with  $\eta_{t}^\Bayes$}
Consider the  following single-section channel, with output $\bs_\tau \in \reals^M$ given by
\begin{align}\label{eq:sparc_single_section_channel}
    \bs_\tau = \bbx_{\sec}+ \sqrt{\tau}\bz \ \in\reals^M,
\end{align}
where $\bbx_{\sec}\sim p_{\bbx_{\sec}}$ and $ \bz\sim \normal_{M}(\bzero, \bI)$ independent of $\bbx_{\sec}$.
Then the potential function for the AMP decoder with  $\eta_{t}^\Bayes$ is defined  as \cite{hsieh2022near}:
\begin{align}\label{eq:pot_fn_sectionwise}
    \mc{F}_\Bayes\left(\mu,\sigma^2,\psi\right) = I(\bbx_{\sec}; \bs_\tau) + \frac{1}{2 \mu}\left[\ln\left(\frac{\tau}{\sigma^2}\right)-\frac{\mu\psi}{\tau}\right],
\end{align}
where $\mu=L/n$ is the user density, $\sigma^2$ is the channel noise of the GMAC, $ \psi\in[0, E]$ and $\tau=\sigma^2+\mu\psi$. The mutual information $I(\bbx_{\sec}; \bs_{\tau})$ is in base-$e$ and can be computed using \eqref{eq:sparc_single_section_channel} to be:
\begin{align}
    I(\bbx_{\sec}; \bs_\tau) 
    &= \alpha\ln\left(\frac{M}{\alpha}\right) -(1-\alpha)\ln(1-\alpha)   \nonumber\\
    &\quad  + \alpha \left(\ln \frac{\alpha}{M} + \frac{E}{2\tau}\right) - \alpha\E_{\bz}\left[\ln\left(\frac{\alpha}{M}\sum_{j= 2}^M e^{\sqrt{\frac{E}{\tau}}z_j - \frac{E}{2\tau}} + \frac{\alpha}{M}e^{\sqrt{\frac{E}{\tau}}z_1 + \frac{E}{2\tau}}+ 1-\alpha\right)\right]\nonumber\\
    &  \quad -(1-\alpha)\E_{\bz}\left[ \ln \left(\frac{\alpha}{M(1-\alpha)} \sum_{j=1}^M e^{\sqrt{\frac{E}{\tau}} z_j - \frac{E}{2\tau}}+ 1\right)\right].\label{eq:MI_sectionwise}
\end{align}
Define the largest minimizer  of $\mc{F}_\Bayes$ w.r.t.\@ $\psi$ as:
\begin{align}\label{eq:min_FBayes}
\mc{M}_\Bayes\left(\mu, \sigma^2\right) = \max\Bigg\{\argmin_{\psi\in[0,E]}\mc{F}_{\Bayes}\left(\mu,\sigma^2,\psi\right)\Bigg\}.
\end{align}
This will be used later in this subsection to characterize the asymptotic error probabilities of the AMP decoder that uses  $\eta_{t}^\Bayes$.

\paragraph*{Potential function for AMP with $\eta_{t}^\marginal$}
Consider the following scalar channel, with output $s_\tau \in \reals$ given by 
\begin{align}\label{eq:sparc_scalar_channel}
    s_\tau = \bar{x}_{\sec} + \sqrt{\tau}z \ \in\reals,
\end{align}
where $\bar{x}_{\sec}$ is distributed as the entrywise marginal of $\bbx_{\sec}$, i.e., 
\[ 
 \bar{x}_{\sec}=\begin{cases}
        \sqrt{E} & \text{with probability }\frac{\alpha}{M}\\
        0 & \text{otherwise}.
    \end{cases} 
    \]
The noise $z\sim \normal(0,1)$ is independent of $\bar{x}_{\sec}$. Then the potential function $\mc{F}_\marginal$ for the AMP decoder with $\eta_{t}^\marginal$ is defined as \cite{kowshik2022improved}:
\begin{align}\label{eq:pot_fn_entrywise}
    \mc{F}_\marginal\left(\mu,\sigma^2,\psi\right) = I(\bar{x}_{\sec}; s_\tau) +\frac{1}{2 \mu M}\left[\ln\left(\frac{\tau}{\sigma^2}\right)-\frac{\mu\psi}{\tau}\right],
\end{align}
 where $\mu, \sigma^2$, $\psi$ and  $\tau$ are as defined in \eqref{eq:pot_fn_sectionwise}, and $I(\bar{x}_{\sec}; s_\tau)$ can be computed as
 \begin{align}
    I(\bar{x}_{\sec}; s_\tau) &=  \frac{\alpha}{M} \ln\left(\frac{M}{\alpha}\right) -\left(1-\frac{\alpha}{M}\right)\ln\left(1-\frac{\alpha}{M}\right) 
    - \frac{\alpha}{M}\E_z\left[\ln\left(1 + \left(\frac{M}{\alpha}-1\right)e^{\sqrt{\frac{E}{\tau}}z -\frac{E}{2\tau}}\right)\right]
    \nonumber\\
    &\quad -\left(1-\frac{\alpha}{M}\right)\left\lbrace\E_z\left[\ln\left(e^{\sqrt{\frac{E}{\tau}}z-\frac{E}{2\tau}} +\frac{M}{\alpha}-1\right)\right] - \ln\left(\frac{M}{\alpha}-1\right)\right\rbrace.\label{eq:MI_entrywise}
 \end{align}
Similarly to  \eqref{eq:min_FBayes}, we define the largest minimizer  of $\mc{F}_\marginal$ w.r.t.\@ $\psi$ as:
\begin{align}
\mc{M}_\marginal\left(\mu, \sigma^2\right) = \max\Bigg\{\argmin_{\psi\in[0,E]}\mc{F}_{\marginal}\left(\mu,\sigma^2,\psi\right)\Bigg\}.
\end{align}
Using the definitions above, we can state the following result, which will be used to characterize the fixed point of each AMP decoder.

\begin{lem}[Fixed points of state evolution] \label{lem:sparc_fixed_point}
Consider the state evolution recursion in \eqref{eq:SC_SE_sparc_phi}--\eqref{eq:SC_SE_sparc_psi}, defined via  an $(\omega, \Lambda)$ base matrix $\bW$ (Definition \ref{def:ome_lamb_rho}), and $\eta_t$ given by either $\eta_{t}^\Bayes$ or $\eta_{t}^\marginal$. Then, for  $\sfc \in [\C]$, the SE parameter $\tau_\sfc^t$  is non-increasing in $t$ and converges to a fixed point  $\tau_\sfc^\infty$, which satisfies the following.

For any $\epsilon > 0$, there exists $\omega_0 < \infty$ and  $\Lambda_0< \infty$ such that for all $\omega> \omega_0$ and  $\Lambda> \Lambda_0$, the fixed points $\{\tau^\infty_\sfc \}_{\sfc \in [\C]}$ satisfy
    \begin{align}\label{eq:lem_tau_bar_sparc}
      \max_{\sfc\in[\C]}  \tau_{\sfc}^\infty  \le \tauconv
      :=\sigma^2 + \muin\left(\mc{M}\left(\muin, \sigma^2\right) + \epsilon\right),
    \end{align}
    where $\muin=\vartheta \mu$, $\vartheta:=\frac{\R}{\C} = 1+ \frac{\omega-1}{\Lambda}$, and   $\mc{M}(\muin,\sigma^2)=\mc{M}_\Bayes(\muin,\sigma^2)$ when $\eta_{t,\sfc}=\eta_{t,\sfc}^\Bayes$, or  $\mc{M}(\muin,\sigma^2)=\mc{M}_\marginal(\muin,\sigma^2)$ when $\eta_{t,\sfc}=\eta_{t,\sfc}^\marginal$. Moreover,
\begin{align}\label{eq:tau_c_large_W}
\lim_{\omega\to\infty}\lim_{\Lambda\to\infty}\tauconv\to \sigma^2 + \mu\left(\mc{M}\left(\mu, \sigma^2\right)+\epsilon\right).
\end{align}
\end{lem}

\begin{figure}[t!]
    \centering
   \includegraphics[width=0.58\textwidth]{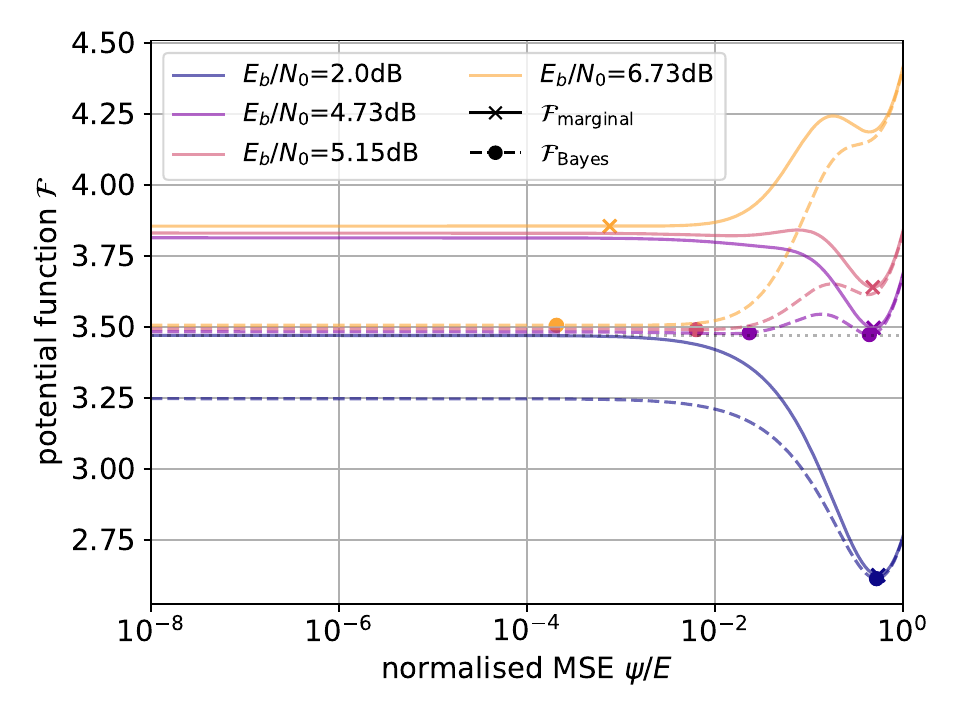}
    \vspace*{-0.4cm}
    \caption{Potential functions $\mc{F}_\Bayes$ and $\mc{F}_\marginal$ vs.\@ normalized MSE $\psi/E$. Different colours correspond to different values of $E_b/N_0$. Here $k=6, \alpha=0.7$, $ \mu = 0.28$.}\label{fig:RA_pot_fn_v_psi_k_6_alpha_0_7}
\end{figure}


We discuss the implications of Lemma  \ref{lem:sparc_fixed_point} before giving the proof. 
Consider  the AMP decoder in \eqref{eq:AMP_sparc_z}--\eqref{eq:AMP_sparc_x} for solving the linear estimation problem in \eqref{eq:Y=AX+W_SPARC} with a spatially coupled Gaussian  design. Recall from \eqref{eq:SC_sparc_eff_obs} that the  asymptotic noise variance in the effective observation for each iteration $t \ge 0$ is quantified by the SE parameters $\{\tau_\sfc^t\}_{\sfc\in[\C]}$, where $\sfc\in[\C]$ indexes the column blocks.  Also recall that these SE parameters depend on the denoiser used by the AMP decoder, i.e., $\eta_{t}^\Bayes$ or $\eta_t^\marginal$. Given a user density $\mu$, Lemma \ref{lem:sparc_fixed_point} says that at the fixed point of the AMP decoder, the effective noise variance  in all column blocks   is asymptotically upper bounded using the largest minimizer of the corresponding potential function $\mc{F}_\Bayes$ or $\mc{F}_\marginal$ for    user density $\muin =\ \vartheta\mu$, where we note that $\muin\to \mu$ in the limit of large base matrices (i.e., $\Lambda\to \infty$ then  $\omega \to \infty $). 
 
Using \eqref{eq:SC_SE_sparc_psi}, the  upper bound on  the asymptotic effective noise variance $\tau_\sfc^\infty$ in Lemma \ref{lem:sparc_fixed_point}  can be converted into an upper bound on  $ \psi_\sfc^\infty$, the asymptotic MSE for   block $\sfc$ achieved  at the fixed point of the AMP decoder.
%
We illustrate this result in Fig.~\ref{fig:RA_pot_fn_v_psi_k_6_alpha_0_7}, where we plot the potential functions $\mc{F}_\Bayes\equiv \mc{F}_\Bayes(\mu, \sigma^2, \psi)$ and $\mc{F}_\marginal\equiv \mc{F}_\marginal(\mu, \sigma^2, \psi)$ against $\psi/E\in [0,1]$, for $\mu=0.28$ and a few different values of  $E_b/N_0$ (varied through $\sigma^2$). Different colours represent different values of $E_b/N_0$.   The dashed curves  correspond to  $\mc{F}_\Bayes$ and  the solid curves  correspond to  $\mc{F}_\marginal$.  The global  minimizers of $\mc{F}_\Bayes$ and $\mc{F}_\marginal$ are marked by circles and crosses, respectively. These minimizers upper bound  $\psi_\sfc^\infty/ E$, the  asymptotic normalized MSE  achievable by the AMP decoder. 

Observe that for a lower SNR, e.g., $E_b/N_0=2$dB, $\mc{F}_\Bayes$ and $\mc{F}_{\marginal}$ each has a unique minimizer, and these minimizers coincide at a high normalized MSE (NMSE) of around $0.5.$ 
When the SNR  increases to $4.73$dB, while $\mc{F}_\marginal$ still has a unique minimizer at  $ \text{NMSE}\approx 0.5$, $\mc{F}_\Bayes$ exhibits two minimizers, one at $\text{NMSE}\approx0.5$ and the other at $\text{NMSE}\approx 0.02$. Nevertheless,  the largest minimizers of $\mc{F}_\marginal$ and $\mc{F}_\Bayes$ still coincide at $0.5$, indicating that the NMSE achievable by the AMP decoder, using either $\eta_t^\marginal$ or $\eta_t^\Bayes$ as its denoiser, is upper bounded by $0.5$.

As the SNR  further increases to $E_b/N_0=5.15$dB or  $6.73$dB,  the  minimizers of the potential functions shift towards smaller NMSE values, reaching as low as  $10^{-4}$. This implies that the AMP decoder  achieves smaller NMSE  for higher SNRs.
Moreover, $\mc{F}_\Bayes$ and $\mc{F}_{\marginal}$ begin to exhibit distinct  minimizers, with the minimizer of $\mc{F}_\Bayes$  consistently being smaller than that of $\mc{F}_{\marginal}$. This indicates that using $\eta_t^{\Bayes}$ instead of $\eta_t^{\marginal}$ as the denoiser  allows the AMP decoder to converge to a smaller NMSE, as one would expect.

\begin{proof}[Proof of Lemma \ref{lem:sparc_fixed_point}]
    This result is similar to \cite[Theorem 2]{hsieh2022near} and \cite[Lemma IV.1]{kowshik2022improved}, and can be proved using the same techniques. 
    Omitting the details, we sketch the key steps of the proof.  Defining  $\gamma_{\sfr}^t := \sum_{\sfc=1}^\C W_{\sfr\sfc} \psi_\sfc^t$ for $\sfr\in [\R]$,  the  SE recursion \eqref{eq:SC_SE_sparc_phi}--\eqref{eq:SC_SE_sparc_psi} associated with the spatially-coupled system  can be rewritten into a single-line recursion:
    \begin{align}
        \gamma_{\sfr}^{t+1} = \sum_{\sfc=1}^\C W_{\sfr\sfc} \,\mmse\left(\sum_{\sfr'=1}^\R  \,\frac{W_{\sfr'\sfc}}{\sigma^2 + \muin \gamma_{\sfr'}^t}\right)\quad \text{for}\; \sfr\in [\R],\label{eq:SC_SE_one_line_recursion}
    \end{align}
    where the function $\mmse(\cdot)$ is defined differently depending on the denoiser used in the AMP decoder.
    For the Bayes-optimal denoiser, we have $\mmse(1/\tau):=\E\|\bbx_{\sec}- \E\left[\bbx_{\sec}\mid \bbx_{\sec} + \sqrt{\tau}\bz\right]\|_2^2$ with $\bz\sim \normal_M(\bzero, \bI)$, and for the marginal-MMSE denoiser, we have  $\mmse(1/\tau):=\E(\bar{x}_{\sec}- \E\left[\bar{x}_{\sec}\mid \bar{x}_{\sec} + \sqrt{\tau}z\right])^2$ with $ z\sim \normal(0,1)$.

    The  SE recursion for the corresponding \emph{uncoupled} system  can be obtained by substituting $W = 1$ in \eqref{eq:SC_SE_sparc_phi}--\eqref{eq:SC_SE_sparc_psi} and dropping the block indices, resulting in the following single-line recursion:
    \begin{align}
       \gamma^{t+1}= \psi^{t+1} = \mmse\left(\frac{1}{\sigma^2 + \muin \psi^t}\right),\label{eq:uncoupled_recursion}
    \end{align}
    where the function $\mmse(\cdot)$ for each denoiser is defined as  in \eqref{eq:SC_SE_one_line_recursion}. The recursions \eqref{eq:SC_SE_one_line_recursion} and \eqref{eq:uncoupled_recursion} correspond exactly to those in \cite[Eqs.\@ (27)--(28)]{yedla2014simple_scalar}. Hence, by \cite[Theorem 1]{yedla2014simple_scalar}, the fixed points of the coupled recursion in \eqref{eq:SC_SE_one_line_recursion} exist and can be bounded from above using the largest minimizer  of the potential function corresponding to the uncoupled recursion in \eqref{eq:uncoupled_recursion}. This potential function can be computed using the formula in \cite[Eq.\@ (4)]{yedla2014simple_scalar}, and
    takes the form of $\mc{F}_\Bayes$ in \eqref{eq:pot_fn_sectionwise} or $\mc{F}_\marginal$  in \eqref{eq:pot_fn_entrywise},  depending on whether $\eta_t^\Bayes$ or $\eta_t^\marginal$ is the denoiser. Finally the bounds on the fixed points $\{ \gamma^\infty_\sfr \}_{\sfr \in [\R]}$ can be translated to bounds on $\{ \tau_\sfc^\infty \}_{\sfc \in [\C]}$ using the same arguments as in the proof of \cite[Theorem 2]{hsieh2022near}.
\end{proof}

Using the characterization of the asymptotic effective noise variance in Lemma \ref{lem:sparc_fixed_point}, we can bound the asymptotic error probabilities $\pMD, \pFA$ and $\pAUE$  achieved by the AMP decoder.

\begin{thm}[Asymptotic achievability bounds]\label{thm:sparc_asymp_errors}
    Consider the coding scheme described in Section \ref{sec:sparc_scheme_for_asymp_ach}, 
    with a spatially coupled design  $\bA$  constructed using an $(\omega,\Lambda)$ base matrix, and   the AMP decoder in \eqref{eq:AMP_sparc_z}--\eqref{eq:AMP_sparc_x} with   $\eta_{t, \sfc} = \eta_{t, \sfc}^\Bayes$ or $\eta_{t, \sfc} = \eta_{t, \sfc}^\marginal$.
 Let the base matrix parameters $\omega> \omega_0$ and $\Lambda> \Lambda_0$, where 
 $\omega_0, \Lambda_0$ are given by Lemma \ref{lem:sparc_fixed_point}. Let $\tauconv$ be as defined  in \eqref{eq:lem_tau_bar_sparc},   which takes different forms depending on  whether the AMP decoder uses $\eta_{t}^\Bayes$ or $\eta_{t}^\marginal$ as the denoiser. 
    %
Recall that $\hbx^{t+1}$ be the hard-decision estimate of $\bx$ that the AMP decoder returns in iteration $t$ as defined in \eqref{eq:hard_decision_sparc}.
    Define also  $\xi(\tau):= \ln\left(\frac{M}{\alpha} (1-\alpha)\right)/\sqrt{E/\tau}$, and
    \begin{align}
        \epsMD(\tau)&:=\Phi\left(\xi(\tau) - \frac{1}{2}\sqrt{\frac{E}{\tau}}\right) \Phi\left(\xi(\tau) + \frac{1}{2}\sqrt{\frac{E}{\tau}}\right)^{M-1}, \label{eq:sparc_epsMD_def}\\
        \epsFA(\tau)&:=\left[\frac{\alpha  \left(1- \Phi\left(\xi(\tau) - \frac{1}{2}\sqrt{\frac{E}{\tau}}\right) 
        \Phi\left(\xi(\tau) +\frac{1}{2}\sqrt{\frac{E}{\tau}}\right)^{M-1} \right)}{(1-\alpha) \left( 1- \Phi\left(\xi(\tau ) + \frac{1}{2}\sqrt{\frac{E}{\tau}}\right)^{M}\right) } +1\right]^{-1},\\
        \epsAUE(\tau)&:=1-\E_z\left[ \Phi\left(\max\left\lbrace \xi(\tau) + \frac{1}{2}\sqrt{\frac{E}{\tau}}, z+\sqrt{\frac{E}{\tau}}\right\rbrace\right)^{M-1}\right] , \quad z\sim\normal(0,1),\label{eq:sparc_epsAUE_def} 
    \end{align}
    where $\Phi(\cdot) $ denotes the cumulative density function of the standard Gaussian. 
     Fix any $\delta>0$, and let $T$ denote the first iteration for which $\max_{\sfc\in[\C]} \tau_{\sfc}^T\le \tauconv + \delta$. 
    Then the $\pMD, \pFA$ and $\pAUE$ achieved by the AMP decoder in iteration $T$ satisfy the following almost surely:
    \begin{align}
        \lim_{L\to\infty}\pMD 
        &= \lim_{L\to\infty} \E\left[\ind\{\Ka\neq0\}\cdot \frac{1}{\Ka} \sum_{\ell\in[L]: \bx_\ell \neq \bzero} \ind\{\hbx_\ell^{T+1}=\bzero\}\right] \nonumber\\
         &\stackrel{\text{a.s.}}{=}  \frac{1}{\C}\sum_{\sfc=1}^\C \P\left(h_{T, \sfc}(\bbx_{\sec,\a}+\bg_\sfc^{T})=\bzero\right)
          \le \epsMD(\tauconv+\delta), \label{eq:asymp_pMD_sparc} \\
    \lim_{L\to\infty}\pFA
    &= 
    \lim_{L\to\infty} \E\left[\ind\{\hKa\ne 0\}\cdot \frac{1}{\hKa} \sum_{\ell\in[L]: \hbx_\ell^{T+1}\neq \bzero} \ind\{\bx_\ell=\bzero\}\right]\nonumber \\
    & \stackrel{\text{a.s.}}{=}\left[\frac{\alpha\sum_{\sfc=1}^{\C}\P(h_{T,\sfc}(\bbx_{\sec,\a}+\bg_{\sfc}^T) \neq\bzero)}{(1-\alpha)\sum_{\sfc=1}^{\C}\P(h_{T,\sfc}(\bg^T_{\sfc})\neq \bzero)}+1\right]^{-1}
    \le
    \epsFA(\tauconv+\delta), \label{eq:asymp_pFA_sparc}\\
        \lim_{L\to\infty}\pAUE 
        &= \lim_{L\to\infty} \E\left[\ind\{\Ka\neq0\}\cdot \frac{1}{\Ka} \sum_{\ell\in[L]: \bx_\ell\ne \bzero} \ind\big\{\hbx_\ell^{T+1}\,\notin \{\bx_\ell, \bzero\}\big\}\right]\nonumber\\
        &\stackrel{\text{a.s.}}{=}  \frac{1}{\C}\sum_{\sfc=1}^\C\P\left( h_{T,\sfc}\left(\bbx_{\sec,\a}+ \bg_{\sfc}^T\right) \notin\{ \bbx_{\sec,\a}, \bzero\}\right) 
        \le \epsAUE(\tauconv+\delta) , \label{eq:asymp_pAUE_sparc}
    \end{align}
    where the limits are taken with $L/n \to\mu$ held constant.
\end{thm}
\begin{proof}
The first equality in each of  \eqref{eq:asymp_pMD_sparc}--\eqref{eq:asymp_pAUE_sparc} follows from the definition of $\pMD, \pFA$ and $\pAUE$ in \eqref{eq:def_my_pMD}--\eqref{eq:def_my_pAUE}.  The second equality in each of in \eqref{eq:def_my_pMD}--\eqref{eq:def_my_pAUE}  can be shown using the state evolution convergence result  and a sandwiching argument, analogous to the proof of \cite[Theorem 1]{hsieh2022near}.  As an example, we sketch the proof of  the second equality in \eqref{eq:asymp_pMD_sparc}.
Applying results from  \cite{bayati2011dynamics, rush2017capacity, rush2021capacity}, the joint empirical distribution of the AMP iterates converges as follows, to a law specified by the state evolution parameters. For any Lipschitz test function $\varphi: \reals^{M} \times \reals^{M} \to \reals$ and $t\ge 1$, we have
\begin{align}
    \label{eq:Lip_conv_sc}
    \lim_{L\to \infty} \frac{1}{L/\C}\sum_{\ell\in \mc{L}_\sfc}
    \varphi(\bx_\ell, \bs_\ell^t) \stackrel{a.s.}{=} \E\left[\varphi (\bar{\bx}_{\sec}, \bar{\bx}_{\sec} + \bg^t_\sfc)\right], \quad \bg_\sfc^t\sim \normal_M(\bzero, \tau_{\sfc}^t\bI),
\end{align}
\edit{where we recall that $\mc{L}_\sfc= \{(\sfc-1)L/\C+1, \dots, \sfc L/\C\}$}.
The claim in \eqref{eq:asymp_pMD_sparc}  requires a test function $\varphi$ that is defined via indicator functions, which are not Lipschitz. We handle this by sandwich
the required function $\varphi$ between two Lipschitz functions $\varphi_{\epsilon, -}(\bx_\ell, \bs_\ell^t)$, $\varphi_{\epsilon, +}(\bx_\ell, \bs_\ell^t)$  that both converge to $\varphi$ as $\epsilon\to 0$). This allows us to apply \eqref{eq:Lip_conv_sc} to  $ \varphi_{\epsilon,-}$ and $\varphi_{\epsilon,+}$, and translate them to $\varphi$ by taking $\epsilon\to 0$) and applying the dominated convergence theorem.  We refer the reader to the proofs of  \cite[Theorem 1]{hsieh2022near} or    \cite[Theorem 1]{liu2024LDPC} for details of the sandwich argument. Hence, for $t\ge 1$ we almost surely have
\begin{align}
      & \lim_{L\to \infty}  \ind\{\Ka\neq0\}\cdot \frac{1}{\Ka}\sum_{\ell\in[L]: \bx_\ell \neq \bzero}\ind\{\hat{\bx}_\ell^{t+1} = \bzero\} \nonumber\\
     & =\left[\lim_{L\to \infty}\ind\{\Ka\neq0\} \frac{L}{\Ka}\right]\cdot 
     \left[\lim_{L\to \infty}\frac{1}{L}\sum_{\ell=1}^L  \ind\{\bx_\ell \neq \bzero\; \text{and}\; \hat{\bx}_\ell^{t+1} = \bzero\}\right]\nonumber\\
     & \stackrel{\text{(a)}}{=} \frac{1}{\alpha} \cdot \frac{1}{\C}\sum_{\sfc=1}^{\C}\lim_{L\to \infty}\frac{1}{L/\C}\sum_{\ell\in \mc{L}_\sfc}\ind\left\lbrace\bx_\ell\neq \bzero\; \text{and}\; h_{t,\sfc}(\bx_\ell + \bg_\sfc^t)=\bzero \right\rbrace\nonumber\\
     & \stackrel{\text{(b)}}{=}   \frac{1}{\C}\sum_{\sfc=1}^{\C}\P\left( h_{t, \sfc}(\bbx_{\sec,\a} + \bg_\sfc^t) =\bzero \right), \label{eq:lim_expectation_MD}
   \end{align}
 where (a) holds by the strong law of large numbers, and (b) by \eqref{eq:Lip_conv_sc} combined with the sandwich argument.    
 
 In the rest of the proof,  we prove the last  inequality in each of  \eqref{eq:asymp_pMD_sparc}--\eqref{eq:asymp_pAUE_sparc}. Recall from \eqref{eq:hard_decision_sparc}--\eqref{eq:sparc_single_section_channel} that for  $\sfc\in [\C]$ and $\ell\in \mc{L}_\sfc$, the hard-decision step   in iteration $t\ge 1$ takes the form:
\begin{align}\label{eq:hard_decision_sparc_1}
\hbx_\ell^{t+1}&=   h_{t, \sfc}(\bs_\ell^t) = \argmax_{\bx' \in \mc{X}} \P\left(\bbx_{\sec} = \bx'\, |\, \ \bbx_{\sec}+ \bg^t_\sfc = \bs_\ell^t\right)\nonumber\\
& =  \argmax_{\bx'\in \mc{X}}\left\lbrace \ln p_{\bbx_{\sec}}(\bx')+ \frac{(\bs_\ell^t)^\top \bx'}{\tau_{\sfc}^t} - \frac{(\bx')^\top \bx'}{2\tau_\sfc^t}\right\rbrace\nonumber\\
& = \argmax_{\bx'\in \mc{X}}
\left\lbrace
    \ln \frac{\alpha}{M} + \frac{(\bs_\ell^t)^\top\bx'}{\tau_{\sfc}^t} - \frac{E}{2\tau_{\sfc}^t} \;  \text{ for}\;  \bx'\ne \bzero\ , \;\ 
        \ln(1-\alpha) \; \text{ for}\;  \bx'=\bzero
\right\rbrace,
\end{align}
where we  recall  $\mc{X} = \{\sqrt{E} \be_j, \, \forall j\in[M]\}\cup \bzero$. Using \eqref{eq:hard_decision_sparc_1} and assuming without loss of generality that $\bx_{\ell} = \sqrt{E}\be_1$, we first derive the expressions for  $\P(h_{t,\sfc}(\bbx_{\sec,\a} + \bg_{\sfc}^t) = \bzero)$, $\P(h_{t,\sfc}(\bg_{\sfc}^t) \ne \bzero )  $ and $\P(h_{t,\sfc}(\bbx_{\sec,\a}+\bg_{\sfc}^t) \ne \bbx_{\sec,\a} ) $,  which form the left side of the inequalities in \eqref{eq:asymp_pMD_sparc}--\eqref{eq:asymp_pAUE_sparc}.  

Letting $z_1, \ldots,  z_M \sim_{\text{i.i.d.}} \normal(0,1)$ and   $\xi(\tau) = \ln\left(\frac{M}{\alpha} (1-\alpha)\right)/\sqrt{E/\tau}$, the hard-decision rule in   \eqref{eq:hard_decision_sparc_1} implies that 
\begin{align}\label{eq:h_active_zero}
    &\P\left(h_{t, \sfc}(\bbx_{\sec,\a}+\bg_\sfc^t)=\bzero\right)\nonumber\\
     &= \P\left(\ln(1-\alpha) > \max\left\lbrace\ln\frac{\alpha}{M} + z_1\sqrt{\frac{E}{\tau_{\sfc}^t}} +\frac{E}{2\tau_{\sfc}^t}\ , \ \ln\frac{\alpha}{M} + z_j\sqrt{\frac{E}{\tau_{\sfc}^t}} - \frac{E}{2\tau_{\sfc}^t} \;\, \text{for all } 2\le j\le M\right\rbrace\right)\nonumber\\
    & = \P\left(z_1< \xi(\tau_{\sfc}^t) - \frac{1}{2}\sqrt{\frac{E}{\tau_{\sfc}^t}}\right) \prod_{j=2}^M\P\left(z_j  < \xi(\tau_{\sfc}^t) + \frac{1}{2}\sqrt{\frac{E}{\tau_{\sfc}^t}}\right)\nonumber\\
    & = \Phi\left(\xi(\tau_{\sfc}^t) - \frac{1}{2}\sqrt{\frac{E}{\tau_{\sfc}^t}}\right) 
    \Phi\left(\xi(\tau_{\sfc}^t) +\frac{1}{2}\sqrt{\frac{E}{\tau_{\sfc}^t}}\right)^{M-1}.
    %
\end{align}
        Similarly, we have
        \begin{align}
&  \P\left(h_{t,\sfc}\left(
    \bbx_{\sec,\a} + \bg_{\sfc}^t\right)  \notin \{\bbx_{\sec,\a}, \bzero\}\right)\nonumber\\
    & = \P\left(\ln\frac{\alpha}{M} + z_j\sqrt{\frac{E}{\tau_{\sfc}^t}} - \frac{E}{2\tau_{\sfc}^t} > \max\left\lbrace\ln (1-\alpha),\  
    \ln\frac{\alpha}{M} + z_1\sqrt{\frac{E}{\tau_{\sfc}^t}} + \frac{E}{2\tau_{\sfc}^t}\right\rbrace,\; \text{for some }2\le j\le M\right) \nonumber\\
    & = 1- \P\left(z_j < \max\left\lbrace \xi(\tau_{\sfc}^t ) + \frac{1}{2}\sqrt{\frac{E}{\tau_{\sfc}^t}}\ ,\ z_1+  \sqrt{\frac{E}{\tau_{\sfc}^t}}\right\rbrace, \; \text{for all }j\ge 2\right)\nonumber\\
    & = 1- \E_{z_1} \left[\Phi\left(\max\left\lbrace\xi(\tau_{\sfc}^t ) + \frac{1}{2}\sqrt{\frac{E}{\tau_{\sfc}^t}}\ ,\  z_1+\sqrt{\frac{E}{\tau_{\sfc}^t}}\right\rbrace\right)^{M-1}\right],\label{eq:h_active_neq_gt}
\end{align}
and
\begin{align}
    \P\left(h_{t,\sfc}\left(
         \bg_{\sfc}^t\right) \neq \bzero \right)
        & = 1- \P\left(h_{t,\sfc}\left(
         \bg_{\sfc}^t\right) =\bzero\right)\nonumber\\
         & = 1- \P\left(\ln (1-\alpha) > \ln \frac{\alpha}{M} + z_j\sqrt{\frac{E}{\tau_{\sfc}^t}} -\frac{E}{2\tau_{\sfc}^t}\,,\;\text{for all } j\in [M]\right)\nonumber\\
     &  = 1- \Phi\left(\xi(\tau_{\sfc}^t ) + \frac{1}{2}\sqrt{\frac{E}{\tau_{\sfc}^t}}\right)^{M}. \label{eq:P_hG_nonzero}
\end{align}
Finally, by substituting \eqref{eq:h_active_zero}--\eqref{eq:P_hG_nonzero} into the left side of each inequality in \eqref{eq:asymp_pMD_sparc}--\eqref{eq:asymp_pAUE_sparc},  and applying the result from Lemma \ref{lem:sparc_fixed_point} along with   $\max_{\sfc\in[\C]} \tau_{\sfc}^T\le \tauconv + \delta$ for some $\delta>0$, we obtain  the upper bounds on the asymptotic error probabilities. 
\end{proof}

When all the users are always active, i.e., $\alpha=1$, the schemes based on random  coding and AMP decoding, along with their asymptotic error analysis,  reduce to their counterparts in \cite{hsieh2022near, kowshik2022improved}. The  error probabilities $\pMD$ and $\pFA$ vanish, leaving $\pAUE$ as the only nonzero error probability, which is equivalent to the per-user probability of error,  $\PUPE := \E\big[\frac{1}{L}\sum_{\ell=1}^L\ind\{\hbx^{T+1}_\ell\ne \bx_\ell\} \big]$. Formally,  Corollary \ref{cor:asymp_error_alpha=1} below follows from  Theorem \ref{thm:sparc_asymp_errors} by taking  $\alpha=1$.
\begin{cor}[Asymptotic achievability bound for $\alpha=1$]\label{cor:asymp_error_alpha=1}
    Consider the setting of Theorem \ref{thm:sparc_asymp_errors} with $\alpha=1$. Let  $\tauconv$ be as defined in \eqref{eq:lem_tau_bar_sparc}, 
    and  $\hbx^{t+1}$ be the  hard-decision estimate of $\bx$ in iteration $t$, defined in \eqref{eq:hard_decision_sparc}.
    Fix any $\delta>0$, and let $T$ denote the first iteration for which $\max_{\sfc\in[\C]} \tau_{\sfc}^T\le \tauconv + \delta$. 
    Then in iteration $T$,  the per-user probability of error ($\PUPE$) achieved by the  AMP decoder satisfies the following almost surely:
    \begin{align}
        \lim_{L\to\infty}\PUPE&=\lim_{L\to\infty}\E\left[\frac{1}{L}\sum_{\ell=1}^L \ind\{\hbx_\ell^{T+1}\ne \bx_\ell\}\right]\le   
        1-\E_z\left[ \Phi\left( z+\sqrt{\frac{E}{\bar{\tau}_\vartheta + \delta}} \right)^{M-1}\right], \label{eq:asymp_pAUE_sparc_alpha=1}
       \end{align}
       where $z\sim \normal(0,1) $, and  the limit is taken with $L/n\to \mu$ held constant.
\end{cor}

Recall  that $\tauconv$  takes different values depending on the denoiser used by the AMP decoder.  
With $\eta_t^{\Bayes}$ as the denoiser,  \eqref{eq:asymp_pAUE_sparc_alpha=1}  matches the bound by Hsieh et al.\@ \cite{hsieh2022near}; see equations  (22) and (37) in \cite{hsieh2022near}. When the AMP decoder employs $\eta_t^{\marginal}$ as its denoiser, \edit{Eq.~\eqref{eq:asymp_pAUE_sparc_alpha=1}} offers a tighter upper bound on $\PUPE$ than the bound by Kowshik  in \cite[Thm.\@ IV.3]{kowshik2022improved}. This is because while our bound assumes the optimal section-wise MAP  rule \eqref{eq:hard_decision_sparc} in the hard-decision step, the latter uses the suboptimal entrywise MAP rule, coupled with the following union bound:
 \begin{align}
     \prob\left(\hat{\bx}_\ell^{t+1} \neq \bx_\ell\right) = \prob\left((\hat{x}_{\ell}^{t+1})_j \neq (x_\ell)_{j}\ \text{for any }j\in [M]\right) \le \sum_{j=1}^M\prob\left((\hat{x}_\ell^{t+1})_j \neq (x_\ell)_{j}\right).\label{kowshik_union}
\end{align}  
Moreover,  both the bound in  \eqref{eq:asymp_pAUE_sparc_alpha=1} and that in \cite[Thm.\@ IV.3]{kowshik2022improved} involve only the distribution function of  the standard normal or its expectation, and are thus equally efficient to compute for larger $k$. 

\subsection{Numerical Results}\label{sec:sparc_ach_numerical}


We numerically evaluate  the asymptotic error bounds on $\pMD, \pFA$ and $\pAUE$ in Theorem~\ref{thm:sparc_asymp_errors} as the system size $n,L\to\infty$ with $L/n \to \mu$ held constant. We also consider the limit of large base matrices, with  $\Lambda\to\infty$ and then $\omega\to\infty$ (hence $ \vartheta = 1+ \frac{\omega-1}{\Lambda}\to 1$ and $\muin=\vartheta\mu \to \mu$). 
Our Python code is available at \cite{liu_amp_github}, and implementation details are given in Appendix~\ref{appx:implementation}.

\begin{figure}[t!]
    \centering
    \includegraphics[width=0.55\textwidth]{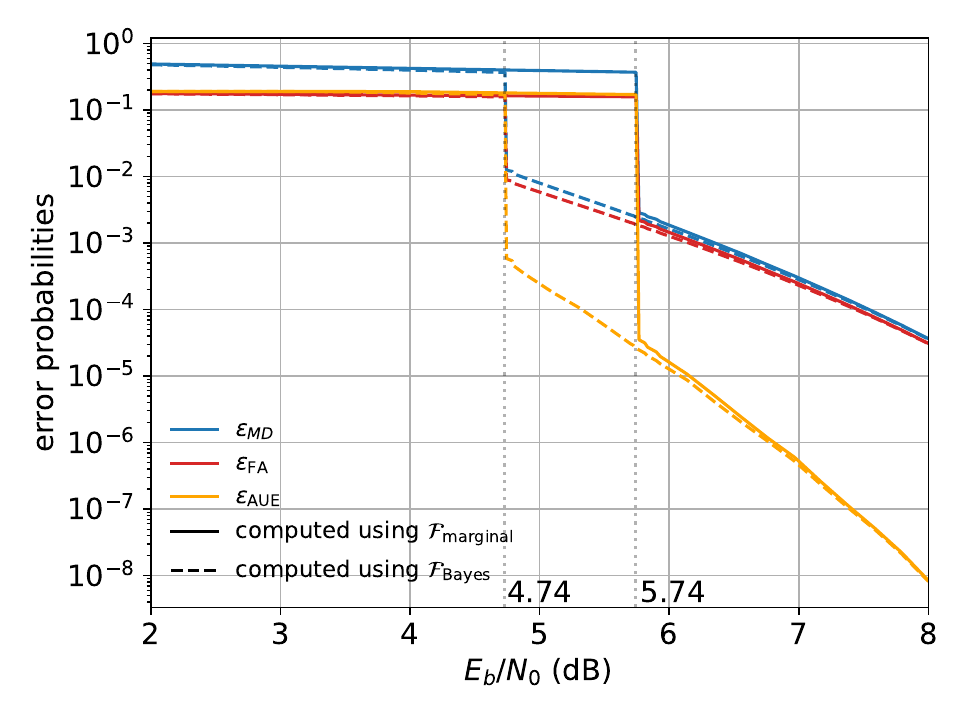}
    \vspace*{-0.4cm}
    \caption{Asymptotic error bounds $\epsMD$, $\epsFA$ and $\epsAUE$  vs.\@  $E_b/N_0$ with $\Lambda\to \infty$ then $\omega\to \infty$,  
    and $ \delta\to 0$ in \eqref{eq:asymp_pMD_sparc}--\eqref{eq:asymp_pAUE_sparc}. $k=6, \alpha=0.7$, $ \mu = 0.28$ (same setting as Fig.~\ref{fig:RA_pot_fn_v_psi_k_6_alpha_0_7}).}
    \label{fig:RA_errors_k_6_alpha_0_7}
\end{figure}
\begin{figure}[t!]
    \centering
    \subfloat[$k=4$.]{\includegraphics[width=0.51\textwidth]{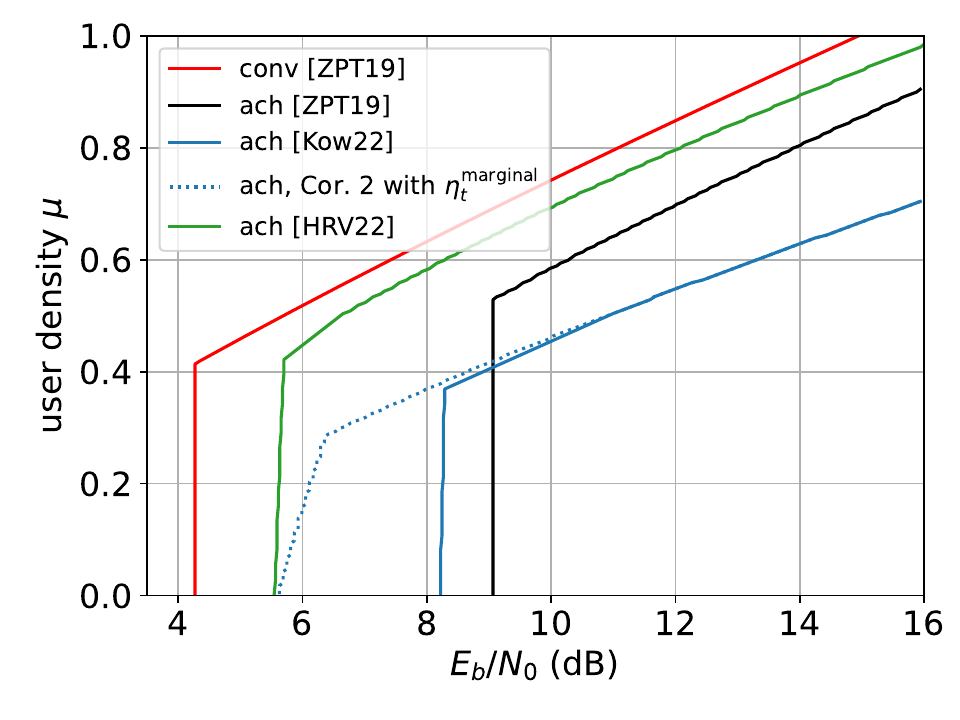}\label{fig:ach_comparison_k4}}
    \hspace*{-0.3cm}
    \subfloat[$k=60$. The achievability bound of Hsieh et al.\@ \cite{hsieh2022near} is omitted as it cannot be computed in an efficient or numerically stable way.]{\includegraphics[width=0.51\textwidth]{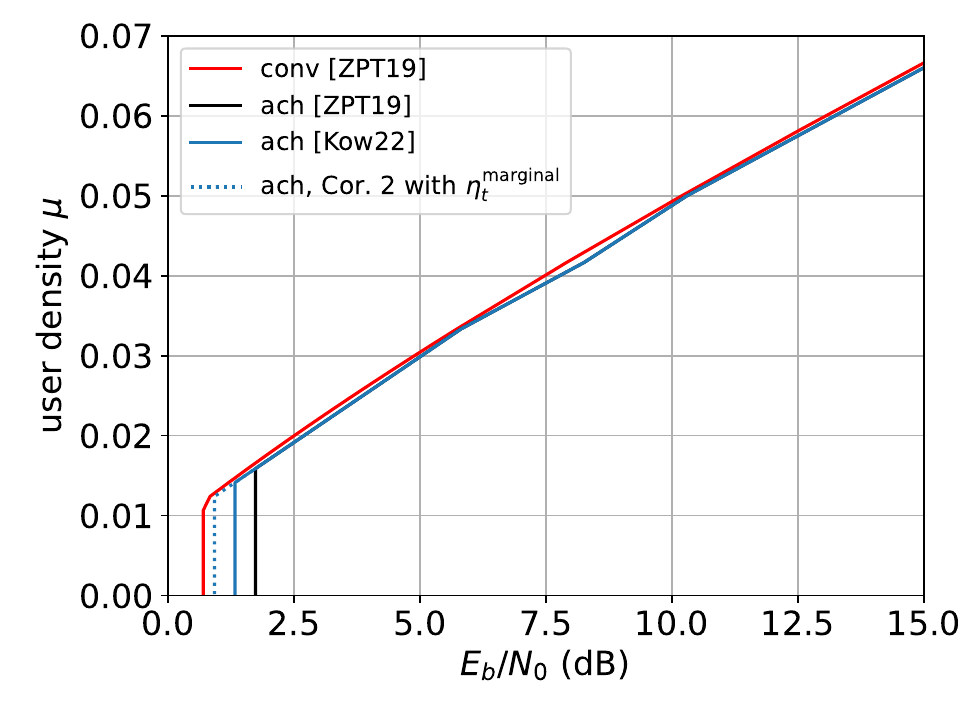}\label{fig:ach_comparison_k60}}
    \caption{Comparison of asymptotic achievability bounds for $\PUPE$ (when $\alpha=1$). Maximum $\mu$  achievable as a function of $E_b/N_0$ for a target   $ \PUPE=10^{-3}$. Subfigures correspond to different user payloads $k$.}\label{fig:alpha=1_kowshik_v_kuan_v_ours}
\end{figure}

Fig.~\ref{fig:RA_errors_k_6_alpha_0_7} plots the error bounds  against  $E_b/N_0$ for the same setting as in Fig.~\ref{fig:RA_pot_fn_v_psi_k_6_alpha_0_7}.
 The dashed curves are obtained using  $\mc{F}_{\Bayes}$ and the solid curves using   $\mc{F}_{\marginal}$. Observe that the error bounds decrease with increasing $E_b/N_0$, with a sharp drop at $4.74$dB or $5.74$dB, depending on whether the bounds are computed using   $\mc{F}_{\Bayes}$ or $\mc{F}_{\marginal}$. This is consistent with   Fig.~\ref{fig:RA_pot_fn_v_psi_k_6_alpha_0_7}, where the minimizer of $\mc{F}_\Bayes$ is significantly  smaller than that of $\mc{F}_{\marginal}$ between $ 4.74$ dB and $ 5.74$dB. In contrast, despite the gap between  the minimizers of $\mc{F}_\Bayes$ and $\mc{F}_{\marginal}$ for $E_b/N_0>5.74$dB in Fig.~\ref{fig:RA_pot_fn_v_psi_k_6_alpha_0_7}, the magnitude of the gap is very small (e.g., $10^{-4}$). 
As a result, the  error bounds associated with  $\mc{F}_\Bayes$ or $\mc{F}_{\marginal}$ nearly coincide with each other  for $E_b/N_0>5.74$dB in Fig.~\ref{fig:RA_errors_k_6_alpha_0_7}.

Moreover, we  observe that  $\epsAUE$ is always lower than $\epsMD$ and $\epsFA$ in  Fig.~\ref{fig:RA_errors_k_6_alpha_0_7}, which is consistent with  the behaviour of the  finite-length bounds in Theorem~\ref{thm:ngo_ach} (see  Fig.~\ref{fig:ngo_ach_vary_r}). Indeed, both the   finite-length  and the asymptotic bounds are  based on random coding, the former using an i.i.d.\@ design matrix and the latter a spatially coupled design. This partially explains why the two  bounds exhibit the same trend. \edit{Furthermore, Fig.~\ref{fig:mua_crit_EbN0_alpha=0.7_diff_k} in Section~\ref{subsec:num_results_CDMA} presents the  active user density $\mu_\a=\alpha\mu$ achievable as a function of $E_b/N_0$ for a target total error  $\max\{\pMD, \pFA\}+\pAUE\le 0.01$, obtained by evaluating  the bounds in Theorem \ref{thm:sparc_asymp_errors}. The results are shown for $\alpha=0.7$,  and  $k=6$ or $k=60$.}

In the absence of random access ($\alpha=1$), Fig.~\ref{fig:alpha=1_kowshik_v_kuan_v_ours} compares our asymptotic achievability bound \eqref{eq:asymp_pAUE_sparc_alpha=1} for $\PUPE$ using $\eta_t=\eta_t^\marginal$  with the   bounds  by Hsieh et al.\@ \cite{hsieh2022near} and Kowshik \cite{kowshik2022improved}. We plot the maximum $\mu$ achievable as a function of $E_b/N_0$ for a target $\PUPE=10^{-3}$. In both  subfigures, the solid red and black curves correspond to the converse and achievability bounds from \cite{zadik2019improved}.

Fig.~\ref{fig:ach_comparison_k4} considers a  small user payload of $k=4$ bits.   The solid green  curve corresponds to the achievability bound by Hsieh et al.\@ \cite{hsieh2022near}, 
the best known one in this regime. 
We observe that our bound (dotted blue)  significantly improves on the bound in  \cite{kowshik2022improved} (solid blue) at lower user densities, and almost matches the near-optimal bound  in \cite{hsieh2022near} 
 (green) at very low user densities. We recall that the improvement arises from our use of  the optimal section-wise MAP rule in the hard-decision step, compared to entry-wise MAP  in 
 \cite{kowshik2022improved} (see \eqref{kowshik_union}).

%
Fig.~\ref{fig:ach_comparison_k60} considers a larger user payload $k=60$, where the bound by Hsieh et al.\@ \cite{hsieh2022near} cannot be computed efficiently without numerical issues. Here  our bound (dotted blue) gives a larger achievable region than the other computable bounds (\cite{zadik2019improved} in solid black and  \cite{kowshik2022improved} in solid blue), which is very close to the converse bound from \cite{zadik2019improved} (solid red).

\section{Efficient CDMA-\edit{T}ype Coding Scheme}
The two  achievability bounds in Theorems \ref{thm:ngo_ach} and \ref{thm:sparc_asymp_errors} above are based on random coding schemes which cannot be implemented for user payloads larger than a few bits, e.g., $k>8$. This is because the  size  $M=2^k$  of the users' codebooks grows exponentially in $k$, posing challenges  in   both  memory  and decoding  complexity. 
In this section, we propose an efficient CDMA-type coding scheme with decoding complexity  that is linear in the payload $k$, and provide performance guarantees that can be compared with the achievability bounds in Theorems \ref{thm:ngo_ach} and \ref{thm:sparc_asymp_errors}.

\subsection{CDMA-\edit{T}ype Coding Scheme}

    
    \edit{As in Section~\ref{sec:asymp_ach}, we consider the simple prior where each user $\ell\in[L]$  is independently active with probability $ \alpha$.  Let   $\bX_\ell \in \reals^{k}$ denote  the message of user $\ell$. We assume  
$\bX_\ell\stackrel{\text{i.i.d.}}{\sim} p_{\bbX}$ for $\ell \in [L]$, where $\bbX \in \reals^{k}$ has the distribution 
\begin{align}\label{eq:p_X}
 p_{\bbX} =   (1-\alpha)\delta_{\bzero} + \alpha \, p_{\bbX_\a}\, .
\end{align}
Here,  $p_{\bbX_\a}$ is   a uniform distribution on the set of BPSK sequences $ \mc{X}_{\a}=\{\pm\sqrt{E_b}\}^k$, where we recall that $E_b$ is the energy per bit.}
Recall that $k$ is  fixed, and does not grow with $n$ and $L$. To construct the codeword of user $\ell$, we  first compute the outer product of a  signature sequence $\ba_\ell\in \reals^{\frac{n}{k}}$  with  $\bX_\ell$ to obtain the matrix $\bC_\ell=\ba_\ell\bX_\ell^\top\in\reals^{\frac{n}{k}\times k}$.  The codeword $\bc_\ell$ of user $\ell$ is then obtained as $\bc_\ell = \vectorize(\bC_\ell) \in\reals^n$. 

For simplicity, let  $\tn=\frac{n}{k}$, and define the matrices $\bA= [\ba_1, \dots, \ba_L]\in \reals^{\tn\times L}$ and $\bX=[\bX_1, \dots, \bX_L]^\top \in \reals^{L\times k}$. 
Then the GMAC channel output in \eqref{eq:standard_gmac} can be rewritten in matrix form as
\begin{equation}
\bY = \sum_{\ell=1}^L\bC_\ell+\bEps = \bA \bX +\bEps \, \in \reals^{\tn\times k},
\label{eq:model_B=1}
\end{equation}
where $\Eps_{ij}\stackrel{\text{i.i.d.}}{\sim}\normal(0, \sigma^2)$ for $i\in \tilde{n}$ and $ j\in [k]$  is  the channel noise.
See Fig.~\ref{fig:Y=AX+W} for an  illustration. 
\edit{This coding scheme can be straightforwardly adapted to  more complicated priors $p_{\bbX}$, including those that model higher-order modulations  such as QPSK instead of BPSK.}

\paragraph*{Spatially coupled  design}
We use a spatially coupled design for $\bA\in \reals^{\tn \times L}$, the matrix of signature sequences. It is defined similarly to the one  in Section \ref{sec:sparc_scheme_for_asymp_ach},  with independent zero-mean Gaussian entries whose variances are specified by a {base matrix} $\bW\in \reals^{\R \times \C}$. 
More precisely, 
\begin{equation}
    \label{eq:construct_sc_A}
A_{i\ell}\stackrel{\text{i.i.d.}}{\sim} \mc{N}\bigg(0,\frac{1}{\tn/\Lr} 
W_{\sfr(i), \sfc(\ell)} \bigg),  \quad \text{for }  i \in [\tn], \ \ell\in [L].
\end{equation}
Similarly to Section \ref{sec:sparc_scheme_for_asymp_ach},  the operators $\sfr(\cdot):[\tn]\rightarrow[\Lr]$ and $\sfc(\cdot):[L]\rightarrow[\Lc]$ here  map a particular row or column index of $\bA$ to its corresponding {row block} or {column block} index.  This construction  ensures  the  columns of $\bA$ have unit squared $\ell_2$-norm asymptotically, and therefore the energy-per-bit constraint is satisfied asymptotically since $ \|\bC_\ell\|_F^2/k =\|\ba_\ell\bX_\ell^{\top}\|_F^2 /k = \|\ba_\ell\|_2^2E_b\to E_b$ as $\tn\to \infty$.  
Our theoretical analysis applies to a generic base matrix $\bW$. For the numerical results, we use the $(\omega, \Lambda)$ base matrix in Definition \ref{def:ome_lamb_rho}.

 \begin{figure}[t!]
\centering
\includegraphics[width=0.7\linewidth]{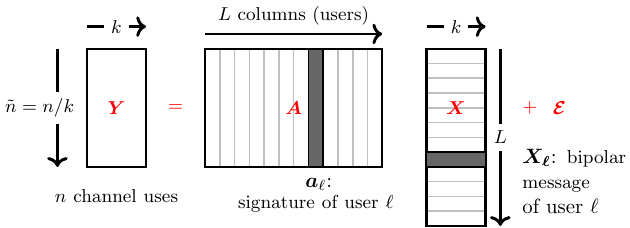}
\vspace{-0.2cm}
\caption{CDMA-type  coding scheme.}
\label{fig:Y=AX+W}
 \end{figure}

\subsection{AMP Decoding and State Evolution}\label{sec:AMP-SE}
This section  describes an AMP decoder to recover $\bX$ from $\bY$ given the spatially coupled matrix $\bA$, and provides the  asymptotic guarantees on its error performance. A key difference from the AMP decoder in Section \ref{sec:AMP_SE_asymp_ach} is that  the signal $\bX$ in \eqref{eq:model_B=1} is a matrix with a row-wise i.i.d. prior, rather than a vector with a section-wise prior.

Starting with an initializer $\bX^0=\left[\E\left[\bbX\right], \, ... \, , \, \E\left[\bbX\right]\right]^\top=\bzero_{L\times k}$,  the  decoder computes the following  in iteration $t\ge 0$: 
\begin{align}
&\bZ^t = \bY -\bA\bX^t +\btZ^t, 
\label{eq:alg_sc_Zt}\\
& \bX^{t+1} =\eta_t\left( \bS^t\right),\quad \text{where}\quad \bS^t= \bX^t+\bV^t.
\label{eq:alg_sc_Xt}    
\end{align}
Here  $\bX^{t+1}$ is an updated estimate of $\bX$ produced using a denoising function $\eta_t: \mathbb{R}^{L\times k}\to \mathbb{R}^{L\times k}$, and $\bZ^t$ can be viewed as a modified residual. The denoiser is assumed to be Lipschitz and separable, acting row-wise on matrix inputs:
\begin{align}
        \eta_t(\bS) = \begin{bmatrix}
            \eta_{t,1}\left(\bS_1\right)^\top\\
            \vdots\\
            \eta_{t,1}\left(\bS_{L/\C}\right)^\top\\
            \vdots\\
            \eta_{t,\C}\left(\bS_{(\C-1)L/\C+1}\right)^\top\\
            \vdots\\
            \eta_{t,\C}\left(\bS_{L}\right)^\top
        \end{bmatrix}
        \setlength{\arraycolsep}{0pt} 
        \begin{array}{ c }
          \left.\kern-\nulldelimiterspace
          \vphantom{\begin{array}{ c }
              \eta_1  \\ 
            \vdots \\ 
              \eta_1 \\   
            \end{array}}
            \right\}\text{$\frac{L}{\C}$ users with $\sfc=1$}\\\\
            \vphantom{\vdots} 
            \left.\kern-\nulldelimiterspace
            \vphantom{\begin{array}{ c }
              \eta_\C \\
              \vdots\\
                \eta_\C\\
          \end{array}}
          \right\}\text{$\frac{L}{\C}$ users with $\sfc=\C$.}
        \end{array}
\label{eq:eta_t_separable}
    \end{align}
Here $\eta_{t,\sfc}: \mathbb{R}^{k}\to \mathbb{R}^{k}$ corresponds to the denoiser for users in block $\sfc\in [\C]$. 

For $t\ge 0$,  $\tilde{\bZ}^t$ and $ \bV^t$ in   \eqref{eq:alg_sc_Zt}--\eqref{eq:alg_sc_Xt} are defined through a matrix 
$\boldsymbol{Q}^t\in\reals^{k\R\times k\C}$, which consists of $ \R \times \C$   submatrices, each of size $k\times k$.  For $\sfr \in [\R], \sfc \in [\C]$, the submatrix $\bQ_{\sfr,\sfc}^t\in \reals^{k\times k}$ is defined as:
\begin{equation}
\bQ_{\sfr,\sfc}^t=\left[\bPhi_{\sfr}^t\right]^{-1}\bT_\sfc^t\,,
   \label{eq:Q_t_def}
\end{equation}
where $\bPhi_\sfr^t, \bT_\sfc^t\in\reals^{k\times k}$ are  deterministic matrices defined later  in \eqref{eq:SC_SE_Phi_t}--\eqref{eq:G_c_t}. 
The $i$th row of matrix $\btZ^t\in\reals^{\tn \times k}$ then takes the form:
\begin{equation}
    \btZ_i^t = {k\muin} \sum_{\sfc =1}^{\Lc} W_{\sfr(i), \sfc} \cdot  \frac{1}{L/{\C}}\sum_{\ell\in \mc{L}_\sfc} \eta_{t-1, \sfc}'\left( \bS_\ell^{t-1} \right)  \, \big[\bQ^{t-1}_{\sfr(i), \sfc} \big]^\top \,\bZ_i^{t-1} \, , \quad  \text{for }  i\in [\tn],
    \label{eq:Z_tilde}
\end{equation}
where  $\muin =  \frac{\R}{\C}\mu$ and $\mc{L}_\sfc= \{(\sfc-1)L/\C+1, \dots, \sfc L/\C\}$.  Here
$\eta_{t,\sfc}'(\bs) = \frac{\de \eta_{t,\sfc}(\bs)}{\de \bs}\in\reals^{k\times k}$ is the derivative (Jacobian) of $\eta_{t,\sfc}$, and  quantities with negative  iteration index are set to all-zero matrices which means at iteration $t=0$, 
\begin{equation}
    \bZ^0=\bY - \bA\bX^0.
    \label{eq:Z_init}
\end{equation}
The $\ell$th row of $\bV^t\in \reals^{L\times k}$ takes the form:
\begin{equation}
    \bV_\ell^t = \sum_{i=1}^{\tn} A_{i,\ell}\Big[\bQ_{\sfr(i),\sfc(\ell)}^t \Big]^\top \bZ_i^t \, , \quad  \text{for }\ell\in[L].
     \label{eq:SC_AMP_V_t}
\end{equation}

\begin{remark}
\normalfont
     We have chosen our notation  to highlight the similarities between the AMP decoder in \eqref{eq:alg_sc_Zt}--\eqref{eq:alg_sc_Xt}  and the one in \eqref{eq:AMP_sparc_z}--\eqref{eq:AMP_sparc_x} in Section \ref{sec:AMP_SE_asymp_ach}. i) The effective observation $\bS^t$ in \eqref{eq:alg_sc_Xt}  defined via \eqref{eq:SC_AMP_V_t} is similar to the expression for $\bs^t$ in \eqref{eq:AMP_sparc_x}. ii)  The debiasing term  $\btZ^t$ in \eqref{eq:alg_sc_Zt}  defined in  \eqref{eq:Z_tilde} is also analogous to the debiasing term $\tbv^t\odot \bz^{t-1}$ in \eqref{eq:AMP_sparc_z}, by noticing that 
 one may rewrite the expression for $\tv^t_i$ in \eqref{eq:AMP_sparc_Q_v} as
\begin{align}
    \tv_i^t=M\muin \sum_{\sfc=1}^{\C} W_{\sfr(i),\sfc}\cdot \frac{1}{LM/\C}\sum_{\ell\in \mc{L}_\sfc}\sum_{j=1}^M
    \left[\frac{ \de \left[\eta_{t-1, \sfc}(s_\ell^{t-1})\right]_j}{\de \left(s_{\ell}^{t-1}\right)_j}\right]^\top Q_{\sfr(i),\sfc}^{t-1} ,\quad  \text{for}\; i\in [n], \label{eq:tilde_vi_expanded}
    \end{align}
    where $Q^t_{\sfr, \sfc} = {\tau_{\sfc}^t}/{\phi_{\sfr}^t}$.
\end{remark}

\paragraph{State evolution (SE)}\label{sec:SC-SE} As we  show  in Theorem \ref{thm:gen_sc_gamp_SE} below, for $t\ge 1,$ in the limit as $L, n\to \infty$ with $L/n\to \mu $,  
the  empirical distribution  of the rows of  $\bZ^t$ in the  block $\sfr\in [\R]$ converges  to a Gaussian  $\mc{N}_k(\bzero, \bPhi_{\sfr}^t) $. Similarly, the empirical distribution of the  rows of $(\bS^t-\bX)$ in block $\sfc\in [\C] $ converges to a Gaussian $\mc{N}_k(\bzero, \bT_{\sfc}^t)$. 
To initialize the state evolution, we make the following assumption \textbf{(A0)} on the AMP initialization $\bX^0 \in \reals^{L \times k}$.

\textbf{(A0)} Denoting by $\bX_\sfc,\,\bX^0_\sfc \in \reals^{L/\C \times k}$ the $\sfc$th block of rows of $\bX,\,\bX^0 \in  \reals^{L\times k}$, 
we assume that there exists a symmetric non-negative definite matrix  $\bXi_\sfc \in \reals^{k \times k}$ for each $\sfc \in [\C]$ such that  we almost surely have
\begin{equation}
\lim_{L \to \infty}  \,    \frac{1}{L/\C} 
\left(\bX_\sfc^0 - \bX_\sfc\right)^\top  \left(\bX_\sfc^0 - \bX_\sfc\right) = \bXi_\sfc.
\label{eq:init_limits_mat}
\end{equation}
The AMP initialization $\bX_\ell^0= \E\left[\bbX\right]=\bzero_{k}$ for $\ell\in [L]$ corresponds to $\bXi_\sfc = \text{Cov}\left[\bbX\right]=\alpha E_b\bI_{k\times k}$. 
The state evolution iteration is initialized with $\bPsi_\sfc^0$ for $\sfc \in [\Lc]$ given by:
\begin{equation}
    \bPsi_\sfc^0 = \bXi_\sfc \label{eq:Psic0}
\end{equation}

The covariance matrices $ \bPhi_{\sfr}^t, \bT_{\sfc}^t\in \reals^{k\times k}$ are then iteratively defined for $t\ge0$ via the following state evolution recursion (SE), for $\sfr\in[\Lr]$ and $\sfc\in[\Lc]$:
\begin{align}
&\bPhi_{\sfr}^t  = \sigma^2 \bI + k\muin\sum_{c=1}^\Lc W_{\sfr,\sfc}\bPsi_\sfc^t\,,\label{eq:SC_SE_Phi_t}\\
&\bPsi_\sfc^{t+1} =  \E \left\lbrace\left[ \eta_{t, \sfc}\left(\bbX +\bG_{\sfc}^{t}\right)-\bbX\right]\left[ \eta_{t,\sfc}\left(\bbX +\bG_{\sfc}^{t}\right)-\bbX\right]^\top\right\rbrace 
 \label{eq:SC_SE_Psi_t},\\
&\text{where } \bG^t_{\sfc} \sim\mathcal{N}_k(\bzero, \bT_{\sfc}^t)\,\perp  \, 
\bbX \sim p_{\bbX}, \quad 
 \bT_{\sfc}^t ={\left[\sum_{\sfr=1}^\Lr W_{\sfr,\sfc}{[\bPhi_\sfr^t]}^{ -1}\right]^{-1}}.
\label{eq:G_c_t}
\end{align}
We can interpret $\bPsi_\sfc^t$ as the   covariance of the asymptotic estimation error between the estimate $\bX^t$ and the true signal matrix $\bX$, for the $\sfc$-th block of users.

\paragraph{Choice of  denoiser $\eta_t$ in AMP}  Since the empirical distribution of the  rows of $\bS^t$ in  block $\sfc$ converges to  the law of  $\bar{\bX}+\bG_{\sfc}^t$ for $t\ge 1$, the Bayes-optimal denoiser $\eta^{\Bayes}_{t,\sfc}: \reals^{k}\to \reals^{k}$ takes the form $\eta_{t,\sfc}^{\Bayes}(\bs)=\mathbb{E} \left[\bbX | \bbX + \bG_{\sfc}^t=\bs \right]$. 
Although this denoiser minimizes the MSE of the estimate, its computational cost  is exponential in $k$, making it impractical for even moderately large payloads. 
 The marginal-MMSE denoiser   $\eta_{t,\sfc}^\marginal: \reals^{k}\to \reals^{k}$ with $ \big[\eta_{t,\sfc}^\marginal(\bs)\big]_j =\E[\bar{X}_j\mid \bar{X}_j+(G^t_\sfc)_j = s_j]$ for $j\in[k]$ enjoys a computational cost linear in $k$, but  performs poorly  as it doesn't account for the correlation between the entries of $\bar{\bX} \in \reals^k$. This correlation is due to the point mass at $\bzero$ in the prior \eqref{eq:p_X}, corresponding to the silent users.

  In light of the limitations of $\eta_{t,\sfc}^\Bayes$ and $\eta_{t,\sfc}^\marginal$, we propose the following \emph{thresholding} denoiser, denoted by $\eta_{t, \sfc}^{\thres}(\bs): \reals^{k}\to \reals^{k}$.   \edit{This denoiser first performs a  Bayesian hypothesis test using $Y=\|\bs\|_2^2/k\in \reals$ as the test statistic, testing  $H_0: \bs\stackrel{d}{=}\bG_\sfc^t$ versus $H_1: \bs\stackrel{d}{=}\bar{\bX}_\a+\bG_\sfc^t$, with prior probabilities $\prob(H_1)= 1 - \prob(H_0) =\alpha$. If $H_0$ is chosen, it   returns the all-zero vector; if  $H_1$ is chosen, it produces an entrywise MMSE nonzero signal estimate.}    Mathematically, we have for $t\ge 1$, $\sfc\in[\C]$ and $j\in [k]$,
\edit{
\begin{align}
   \left[\eta_{t, \sfc}^{\thres}(\bs)\right]_j =  \begin{cases}
      0 & \text{if $H_0$ is chosen}\,,\\
      \E\big[(\bar{X}_{\a})_j\ |\ (\bar{X}_{\a})_j+(G_{\sfc}^t)_j = s_j \big]
      & \text{otherwise}\,,
    \end{cases}   \label{eq:thres_denoiser}    
\end{align}}
where, since $p\big((\bar{X}_\a)_j = \sqrt{E_b}\big) = p\big((\bar{X}_\a)_j=-\sqrt{E_b}\big)=\frac{1}{2}$,
\begin{align*}
\E\big[(\bar{X}_{\a})_j\ |\ (\bar{X}_{\a})_j+(G_{\sfc}^t)_j = s_j \big]
= \sqrt{E_b}\tanh \big( \sqrt{E_b}s_j/ (T_{\sfc}^t)_{j, j}\big).
%
\end{align*}
This denoiser $\eta_{t, \sfc}^\thres$ has the same $\mc{O}(k)$ computational cost  as $\eta_{t, \sfc}^\marginal$, but  can noticeably  enhance decoding  performance because it uses the fact that approximately $(1-\alpha)$ fraction of the rows of the signal matrix $\bX$  are all-zero. 
The Jacobian of this denoising function can also be easily computed.

\edit{\emph{Designing the  hypothesis test:}  We use $\chi_k^2(\mu)$ to denote a chi-squared random variable with  $k$ degrees of freedom and non-centrality parameter $\mu$. If  $\mu_1, \dots, \mu_k$ are the means of  $k$ independent Gaussians with unit variance whose squares are summed to form the chi-squared variable, then $\mu=\sum_{j=1}^k \mu_j^2$.
We approximate the entries of $\bG_\sfc^t$ as  i.i.d.\@ with variance $\bar{T}^t_\sfc\coloneqq\frac{1}{k}\sum_{j=1}^k (T^t_{\sfc})_{j,j}$. Based on this, the test statistic  $Y$ follows a  chi-squared distribution whose  non-centrality parameter depends on whether  $\bs$ contains a non-zero signal. Furthermore, by the central limit theorem, as $k\to \infty$, we have   $\left(\chi_k^2(\mu)-(k+\mu)\right)/ \sqrt{2(k+2\mu)}\stackrel{d}{\to} \mc{N}(0,1)$. Therefore, assuming large $k$, we use the following approximations for the test statistic:
\begin{align}
\begin{split}
       & \text{under }H_0: Y\sim (\bar{T}_\sfc^t/k) \,  \chi_k^2(0) \; \approx   \mc{N}\left(\bar{T}_\sfc^t\,,\, 2 (\bar{T}_\sfc^t)^2/k\right)\, , \\
    %
    & \text{under }H_1: Y\sim (\bar{T}_\sfc^t/k) \, \chi_k^2(kE_b/\bar{T}_\sfc^t) \; \approx  \mc{N}\left( \bar{T}_\sfc^t+ E_b\,,\, 2 \bar{T}_\sfc^t (\bar{T}_\sfc^t + 2E_b)/k\right) \, . 
    \label{eq:H0_H1_chi_squared} 
\end{split}
 \end{align}
 Applying \eqref{eq:H0_H1_chi_squared}, the MAP decision rule chooses $H_0$ if 
 \begin{align}
     Y^2 - \bar{T}_\sfc^t Y + \frac{1}{2E_b} \left(-\bar{T}_\sfc^t E_b^2 -  \bar{T}_\sfc^t(\bar{T}_\sfc^t + 2E_b) \cdot  \frac{4\bar{T}_\sfc^t}{k}\ln\left[\frac{1-\alpha}{\alpha} \sqrt{\frac{\bar{T}_\sfc^t + 2E_b}{\bar{T}_\sfc^t}}\right] \right) < 0\,,\label{eq:chi_squared_thres}
 \end{align}
 and chooses $H_1$ otherwise. This implies that  $\bs$ is always declared to contain  a non-zero signal when $\alpha\in [\alpha^*, 1]$ with $ \alpha^*\coloneqq \left[\sqrt{\bar{T}_\sfc^t/(\bar{T}_\sfc^t+2E_b)} \exp\left(-E_bk/(8\bar{T}_\sfc^t)\right)+1\right]^{-1} $.  When $\alpha\in [0,\alpha^*)$,  the decision follows the threshold test in \eqref{eq:chi_squared_thres}.} Moreover, depending on the application, one might want to  control  $\pMD$ more strictly than $\pFA$, or vice-versa. \edit{In such cases, a principled approach is to apply the threshold test in  \eqref{eq:chi_squared_thres}  by substituting  a  higher value $\hat{\alpha}$ for $\alpha$  to control $\pMD$ more strictly, or a lower value $\hat{\alpha}$ to control $\pFA$ more strictly. The effectiveness of this approach is illustrated in Fig.~\ref{fig:mismatch_alpha} in Section~\ref{subsec:num_results_CDMA}.} 



\paragraph*{Hard-decision step in AMP}
In each iteration $t\ge 1$, the decoder can produce a hard-decision   estimate for each row of the signal matrix. The MAP estimate is analogous to the one   in \eqref{eq:hard_decision_sparc}. 
Additionally, we also consider a hard-decision estimator complementary to the thresholding denoiser $\eta_{t,\sfc}^\thres$ defined in \eqref{eq:thres_denoiser},  which provides an entrywise estimate as follows.  For $\sfc\in [\C]$ and $j\in [k]$,
\begin{align}\label{eq:thres_hard_decision}
   \left[h_{t, \sfc}(\bs)\right]_j =  \begin{cases}
      0 \qquad \text{if $H_0$ is chosen}\,,\\
      \argmax_{x\in \{\pm\sqrt{E_b}\}}\P\big((\bar{X}_{\a})_j=x \ |\ (\bar{X}_{\a})_j+(G_{\sfc}^t)_j = s_j \big)=\text{sign}(s_j)\sqrt{E_b}
      \qquad   \text{otherwise}\,.
    \end{cases}       
\end{align}
This estimator has an $\mc{O}(k)$ computational complexity.

 \subsection{Asymptotic Error Analysis for AMP Decoding}\label{sec:error_analysis_AMP_cdma}

We first present a general state evolution result, Theorem \ref{thm:gen_sc_gamp_SE}, that characterizes the asymptotic performance of AMP for the model \eqref{eq:model_B=1} with a spatially coupled design matrix $\bA$. Theorem \ref{thm:gen_sc_gamp_SE} applies to generic row-wise  priors on the signal $\bX$ and to any Lipschitz denoisers for the AMP algorithm. 
It states that   the  row-wise 
 empirical distributions of the AMP iterates in \eqref{eq:alg_sc_Zt}--\eqref{eq:alg_sc_Xt} converge to the laws of the state evolution random variables defined in \eqref{eq:SC_SE_Phi_t}--\eqref{eq:G_c_t}. Based on this result, we then obtain Theorem \ref{thm:MetricsSE}, which characterizes the asymptotic $\pMD, \pFA$ and $\pAUE$ of the AMP decoder for the CDMA scheme.

 To present Theorem \ref{thm:gen_sc_gamp_SE}, we require the following assumptions on the  general  linear model \edit{$\bY = \bA \bX +\bEps$, where $\bA \in \reals^{\tn\times L}$,$\bX \in \reals^{L \times k}$,  $\bEps \in \reals^{\tn\times k}$.}

\noindent\textbf{(A1)} \edit{As $L,\tn \to \infty$,  $L/\tn\to k\mu$, for some constant $\mu>0$}. The number of columns in the signal matrix, $k$,  is a fixed constant that does not grow  with $n$.

\noindent\textbf{(A2)} Both the signal matrix $\bX$ and noise matrix $\bEps$ are independent of $\bA$. 

\noindent\textbf{(A3)} As $n \to \infty$, the row-wise empirical distributions of the signal and the noise matrices converge in Wasserstein distance to well-defined limits. More precisely, write $\nu_L(\bX)$ and $\nu_{\tn}(\bEps)$ for the row-wise empirical distributions of $\bX$ and $\bEps$, respectively. Then for some $m \in [2, \infty)$,  there exist $k$-dimensional  random variables $\bbX \sim p_{\bbX}$  and $\bbE\sim p_{\bbE}$ with $\int_{\reals^k}\|\bx\|^m dp_{\bbX}(\bx), \int_{\reals^k}\|\bx\|^m dp_{\bbE}(\bx)<\infty$ such that  $d_m(\nu_L(\bX),p_{\bbX}) \rightarrow 0$ and $d_m(\nu_{\tn}(\bEps), p_{\bbE}) \rightarrow 0$ almost surely. Here $d_m(P,Q)$ denotes the $m$-Wasserstein distance between distributions $P$ and $Q$ defined on the same Euclidean probability space.

Theorem \ref{thm:gen_sc_gamp_SE} is stated in terms of \emph{pseudo-Lipschitz} test functions. A function $\xi:\reals^k\rightarrow\reals$ is pseudo-Lipschitz of order $m$ if  $|\xi(\bx)-\xi(\by)|\leq C \left(1+\|\bx\|_2^{m-1}+\|\by\|_2^{m-1}\right)\|\bx-\by\|_2$ for all $\bx,\by\in\reals^k$, for some constant $C>0$. Pseudo-Lipschitz functions of order $m=2$ include the squared difference $\xi(\bu, \bv)=\|\bu - \bv\|_2^2$, and the correlation $\xi(\bu,\bv)=\left\langle\bu, \bv\right\rangle$.

\begin{thm}[State evolution result] Consider the  linear model $\bY = \bA \bX +\bEps \, \in \reals^{\tn\times k}$ with the spatially coupled design  matrix $\bA$ in \eqref{eq:construct_sc_A},
and the AMP decoding algorithm in \eqref{eq:alg_sc_Zt}--\eqref{eq:alg_sc_Xt} with 
a  sequence of Lipschitz continuous denoisers $\{\eta_t\}_{t\ge 0}$. 
Assume   \textbf{(A0)}--\textbf{(A3)} hold. 
Let $\xi: \reals^{k}\times \reals^k \to \reals$ and $\zeta: \reals^{k}\times \reals^k \to \reals$ be  pseudo-Lipschitz test functions of order $m$, where $m$ is specified by \textbf{(A3)}. 
Then, for $t \ge 1$ and  $\sfr \in [\Lr]$ and $\sfc \in [\Lc]$, we almost surely have:
\begin{align}
  &\lim_{L \to \infty} \,  \frac{1}{L/\C} \sum_{\ell \in  \mc{L}_\sfc} \xi \left(\bX_\ell^{t+1} \, , \bX_\ell\right)  = \E\left[ \xi\left(\eta_{t,\sfc}\left(\bbX+{\bG}_{\sfc}^t\right), \bbX\right) \right],     \label{eq:Xt_result} \\
  &\lim_{\tn \to \infty} \,  \frac{1}{\tn/\R} \sum_{i \in  \mc{I}_\sfr} \zeta\left(\bZ_i^t \, , \bEps_i\right)  = \E\left[ \zeta \left( \tilde{\bG}_{\sfr}^t \, , \bbE\right) \right],     \label{eq:Zt_result}
\end{align}
\label{thm:gen_sc_gamp_SE}
where ${\bG}_{\sfc}^t \sim\normal_k(\bzero, \bT_{\sfc}^t)$ and $\tilde{\bG}_{\sfr}^t \sim \mathcal{N}_k\left(\0, \bPhi_\sfr^t\right)$   with $\bPhi_\sfr^t, \bT_{\sfc}^t$  given by \eqref{eq:SC_SE_Phi_t}--\eqref{eq:G_c_t},  and  \\  
$\mc{L}_\sfc:= \left\lbrace(\sfc-1){L}/{\C}+1, \dots, \sfc{L}/{\C}\right\rbrace$ for $\sfc\in [\C]$ and $ \mc{I}_\sfr:= \left\lbrace (\sfr-1)\tn/\R+1, \dots, \sfr\tn/\R\right\rbrace$ for $\sfr\in [\R]$.
\end{thm}
\begin{proof}
    The proof of Theorem \ref{thm:gen_sc_gamp_SE} is given in Section \ref{app:gen_sc_gamp_proof}.
\end{proof} 
Using Theorem \ref{thm:gen_sc_gamp_SE}, we can obtain  Theorem \ref{thm:MetricsSE} which characterizes 
the asymptotic $\pMD, \pFA$ and $\pAUE$ of the AMP decoder \eqref{eq:alg_sc_Zt}--\eqref{eq:alg_sc_Xt} with a generic Lipschitz denoiser.

\begin{thm}[Asymptotic error probabilities of CDMA scheme]\label{thm:MetricsSE}
Consider the CDMA-type scheme defined in \eqref{eq:model_B=1} with the spatially coupled design matrix $\bA$ in \eqref{eq:construct_sc_A}, with the rows of the signal matrix $\bX$ drawn   i.i.d.\@ according to the prior $p_{\bbX}$ in \eqref{eq:p_X}.
Consider the AMP decoder \eqref{eq:alg_sc_Zt}--\eqref{eq:alg_sc_Xt} with a sequence of Lipschitz continuous denoisers $\{\eta_t\}_{t\ge 0}$ taking the form of \eqref{eq:eta_t_separable}. 
 Let $\hat{\bX}^{t+1} = h_{t}(\bS^t )$ be the hard-decision estimate  in  iteration $t$.
 %
In the limit as $L, n \to \infty$ with $L/n \to \mu$, 
the $\pMD$, $\pFA$ and $\pAUE$ in iteration $t\ge 1$  satisfy the following almost surely:
\begin{align}
&  \lim_{L\to \infty}    \pMD =\lim_{L\to \infty} \E\Bigg[\frac{\ind\{\Ka\neq0\}}{\Ka}\sum_{ \ell\in[L]: \bX_\ell \neq \bzero}\ind\{\hat{\bX}_\ell^{t+1} = \bzero\}\Bigg] 
  = \frac{1}{\Lc}\sum_{\sfc=1}^\Lc\prob \left( h_{t,\sfc}\left(\bbX_\a+ \bG_{\sfc}^t\right) = \bzero \right) , \label{eq:MD_thm}\\ 
& \lim_{L\to \infty}     \pFA=\lim_{L\to \infty}\E\Bigg[ \frac{\ind\{\hKa\ne 0\}}{\widehat{\Ka}}\sum_{\substack{\ell\in[L]:\\\hat{\bX}_\ell^{t+1} \neq \bzero}}\ind\{\bX_\ell=\bzero\} \Bigg] 
 = \left[\frac{\alpha\sum_{\sfc=1}^{\C}\prob(h_{t,\sfc}(\bbX_{\a}+\bG_{\sfc}^t) \neq\bzero)}{(1-\alpha)\sum_{\sfc=1}^{\C}\prob(h_{t,\sfc}(\bG^t_{\sfc})\neq \bzero)}+1\right]^{-1}\label{eq:FA_thm}, \\
 &\lim_{L\to \infty}    \pAUE=\lim_{L\to \infty}\E\Bigg[ \frac{\ind\{\Ka\ne0\}}{\Ka}\sum_{\ell\in[L]: \bX_\ell \neq \bzero}\ind\big\{\hat{\bX}_\ell^{t+1} \notin \{\bX_\ell,  \bzero\}\big\} \Bigg]\nonumber\\
 &\quad\quad\qquad\;\;
 = \frac{1}{\Lc}\sum_{\sfc=1}^\Lc\prob \left( h_{t,\sfc}\left(\bbX_\a+ \bG_{\sfc}^t\right)  \notin \{\bbX_\a, \bzero\}\right).
    \label{eq:AUE_thm}
\end{align}
Here $\bbX_\a\sim p_{\bbX_{\a}}$  and $\bG_{\sfc}^t\sim \mc{N}_k(\bzero, \bT_{\sfc}^t)$ are independent, with their laws given in \eqref{eq:p_X} and \eqref{eq:G_c_t}.
\end{thm}
\begin{proof}
We  use Theorem \ref{thm:gen_sc_gamp_SE} to prove   Theorem \ref{thm:MetricsSE}. It is easy to check that assumptions \textbf{(A0)}--\textbf{(A3)} are satisfied. Setting  $\eta_{t,\sfc}(\bs)=\bs$ in \eqref{eq:Xt_result} gives
\begin{align}\label{eq:Xt_result_simplified}
    \lim_{L \to \infty} \,  \frac{1}{L/\C} \sum_{\ell \in  \mc{L}_\sfc} \xi\left(\bS_\ell^t \, , \bX_\ell\right)  \stackrel{a.s.}{=} \E\left[ \xi\left(\bbX+{\bG}_{\sfc}^t, \bbX\right) \right],
\end{align}
for any pseudo-Lipschitz function of order $2$.
The results \eqref{eq:MD_thm}--\eqref{eq:AUE_thm} then  follow from arguments similar to \eqref{eq:Lip_conv_sc}--\eqref{eq:lim_expectation_MD} in the proof of Theorem \ref{thm:sparc_asymp_errors}. 
\end{proof}
\begin{cor}
    The asymptotic error guarantees  in Theorem \ref{thm:MetricsSE} hold for the AMP decoder with any of the three denoisers: $\eta^{\Bayes}_{t,\sfc}, \eta^{\marginal}_{t,\sfc}$, and for $\eta^{\thres}_{t,\sfc}$.
    \end{cor}
    \begin{proof}All three denoisers can be verified to be Lipschitz continuous by direct differentiation. The result then follows from Theorem \ref{thm:MetricsSE}.
    \end{proof}

\edit{We remark that in the special case of an i.i.d.\@ Gaussian design matrix (i.e., $\R = \C =1$), the asymptotic characterization of the CDMA scheme in  Theorem \ref{thm:MetricsSE} can be directly obtained from the state evolution result of Bayati and Montanari \cite{bayati2011dynamics}.}

\subsection{Numerical Results} \label{subsec:num_results_CDMA}
The Python code for all simulations presented in this section can be found  at \cite{liu_amp_github}. 
Fig.~\ref{fig:marg_v_thres} plots $\pMD, \pFA$, and $\pAUE$ for the CDMA scheme,  with a   user payload of $k = 60$ bits. For such payloads, the marginal-MMSE denoiser $\eta_t^\marginal$ and  the thresholding denoiser $\eta_t^\thres$ are the only computationally viable options, 
coupled with the entrywise MAP hard-decision step or the thresholding hard-decision step defined in \eqref{eq:thres_hard_decision}. The two subfigures of Fig.~\ref{fig:marg_v_thres} compare the  performance of the two denoisers for varying $E_b/N_0$. We observe that the thresholding 
denoiser $\eta^\thres_t$ achieves significantly lower error probabilities. In both cases, the empirical error probabilities (crosses) closely match the asymptotic formulas from Theorem~\ref{thm:MetricsSE} (solid lines).


\edit{Fig.~\ref{fig:mismatch_alpha} shows  the limiting error probabilities for the CDMA scheme using the thresholding denoiser $\eta_t^\thres$, where the threshold test \eqref{eq:chi_squared_thres} uses a  user-chosen  value $\hat{\alpha}$ for  $\alpha$   instead of the true $\alpha=0.7$. The curves are obtained using state evolution and  the expressions in Theorem \ref{thm:MetricsSE}. We observe that the  denoiser balances $\pMD$ and $\pFA$ equally when $\hat{\alpha}=\alpha$ (Fig.~\ref{fig:mismatch_b}). However, when $\hat{\alpha}\ne \alpha$, the performance becomes asymmetric:  choosing  $\hat{\alpha}<\alpha$  biases the denoiser towards detecting silent users, favoring a smaller $\pFA$ (Fig.~\ref{fig:mismatch_a}), while 
choosing   $\hat{\alpha}>\alpha$ biases the denoiser towards detecting active users, leading to a smaller $\pMD$  (Fig.~\ref{fig:mismatch_c}).} 

Additionally,  we observe that  the CDMA  schemes  tend to incur higher $\pAUE$ than  $\pMD$ or $\pFA$, which  contrasts  the random coding schemes illustrated in Figs.~\ref{fig:ngo_ach_vary_r} and \ref{fig:RA_errors_k_6_alpha_0_7} where $\pAUE$ is typically   \emph{lower} than   $\pMD$ and $\pFA$.  This is because, in the  CDMA schemes, the AMP decoder recovers each  user's message as a  sequence of $k$ bipolar symbols or all zeros, and it  is more likely to make a symbol error than to mistake  all $k$ symbols for zeros, or vice versa. In contrast, for random coding, the maximum-likelihood  decoder in \eqref{eq:min_distance_decoder} or the AMP decoder in \eqref{eq:AMP_sparc_z}--\eqref{eq:AMP_sparc_x} is more likely to mistake a transmitted codeword for noise, or vice versa, than to mistake a transmitted codeword for a different codeword. 
\begin{figure}[t!]
    \centering
 \includegraphics[width=0.6\linewidth]{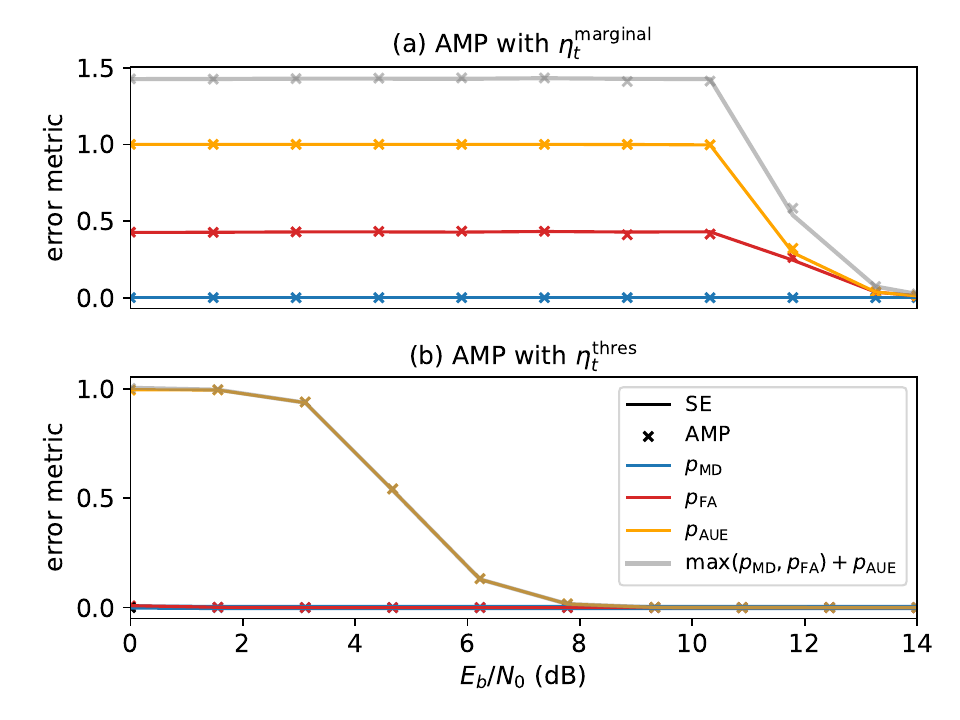}
\vspace{-0.5cm}
    \caption{Error probabilities $\pMD, \pFA, \pAUE$  and \edit{the total error metric $\max\{\pMD, \pFA\}+ \pAUE$} 
    from Theorem \ref{thm:MetricsSE}  for  the CDMA scheme with  $k=60, \alpha=0.7, \mu_\a=0.013$,  $L=8000$, and spatial coupling parameters $(\omega=1, \Lambda=1)$. \edit{Subfigures (a) and (b)} correspond to  $\eta_t^\marginal$ and $\eta_t^\thres$ as the denoiser, respectively. }
    \label{fig:marg_v_thres}
\end{figure}

\begin{figure}[t!]
\centering
    \subfloat[$\hat{\alpha}=0.6$ ]{\includegraphics[width=0.337\textwidth]{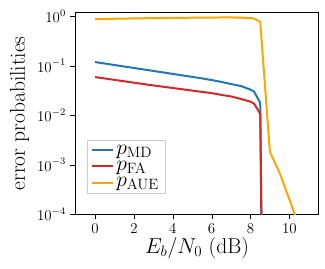}\label{fig:mismatch_a} }
    \hspace*{0.01cm}
    \subfloat[$\hat{\alpha}=\alpha=0.7$]{\includegraphics[width=0.32\textwidth]{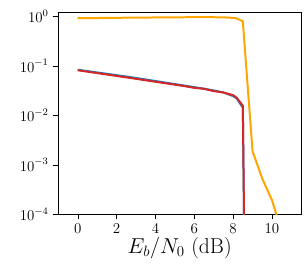}\label{fig:mismatch_b} }
    \hspace*{0.01cm}
    \subfloat[$\hat{\alpha}=0.8$]{\includegraphics[width=0.32\textwidth]{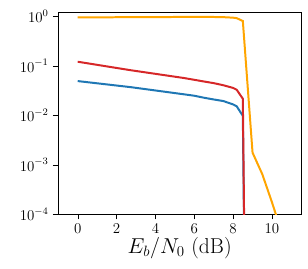}\label{fig:mismatch_c} }
    \caption{Limiting error probabilities  $\pMD, \pFA, \pAUE$  for the CDMA scheme with $\eta_t^{\thres}$, where $k=60,  \alpha=0.7, \mu_\a=2$, and spatial coupling parameters are $(\omega=11, \Lambda = 50)$.  The denoiser $\eta^\thres$ uses the threshold test \eqref{eq:chi_squared_thres} with a mismatched value  $\hat{\alpha}$ for $\alpha$.}
    \label{fig:mismatch_alpha}
\end{figure}

Fig.~\ref{fig:mua_crit_EbN0_alpha=0.7_diff_k} compares the limiting error probabilities of the efficient CDMA scheme in Theorem~\ref{thm:MetricsSE} with the two   achievability bounds from Theorems~\ref{thm:ngo_ach} and \ref{thm:sparc_asymp_errors}.   We plot the maximum active user density $\mu_\a$ achievable as a function of $E_b/N_0$ for a \edit{target total error $\max\{\pMD, \pFA\}+\pAUE\le0.01$.}
%
The asymptotic achievability bounds from Theorem~\ref{thm:sparc_asymp_errors} are computed using either the potential function $\mc{F}_\Bayes$ (dashed black) or $\mc{F}_\marginal$ (solid black) with $\Lambda\to \infty$ then $\omega\to \infty$.

Fig.~\ref{fig:RA_pot_fn_ach_v_sc_se_k6} considers a small user payload $k=6$, same setting as in Figs.~\ref{fig:RA_pot_fn_v_psi_k_6_alpha_0_7}--\ref{fig:RA_errors_k_6_alpha_0_7}. The CDMA schemes employ AMP decoding with either the Bayes-optimal denoiser $\eta_t^\Bayes$ (solid blue) or the thresholding denoiser $\eta_t^\thres$ (solid orange). Notably,   the achievable region of the efficient  thresholding denoiser is only slightly smaller than that of the Bayes-optimal denoiser, and the latter  is smaller than that of random coding (solid and dashed black).

The dotted black curve in Fig.~\ref{fig:RA_pot_fn_ach_v_sc_se_k6} is computed using the finite-length achievability bounds from Theorem~\ref{thm:ngo_ach} with \edit{$L=100$} \edit{and $(\kl, \ku)$ chosen so that $\P(\Ka\notin [\kl:\ku])\le10^{-6}$}. 
The decoding \edit{radii} are chosen to be \edit{$\rl=\ru=L$}, which  implies that the maximum-likelihood decoder in  \eqref{eq:min_distance_decoder} searches through all possible combinations of active users out of $L$, \edit{and  $P'$ is optimized over the full range of values $(0,P)$.} This gives the tightest achievability bound in this setting.  
We observe  that the overall shape of all the plots in Fig.~\ref{fig:RA_pot_fn_ach_v_sc_se_k6}  is quite similar. They  are steep (near-vertical) for very low active user densities ($\mu_\a< 0.17$), where the multi-user interference is nearly  cancelled, essentially achieving the performance of  a single-user channel. The plots exhibit a sharp  transition at around $\mu_\a\in (0.17,0.21)$. Beyond this point, the plots display a much shallower slope due to the noticeable multi-user interference.
This observation is consistent with the findings in \cite{polyanskiy2017perspective,zadik2019improved,kowshik2022improved, kowshik2021fundamental}.

\begin{figure}[t!]
    \centering
    \subfloat[$k=6$,  CDMA scheme uses $\omega=6,\Lambda=40$. ]{\includegraphics[width=0.49\textwidth]{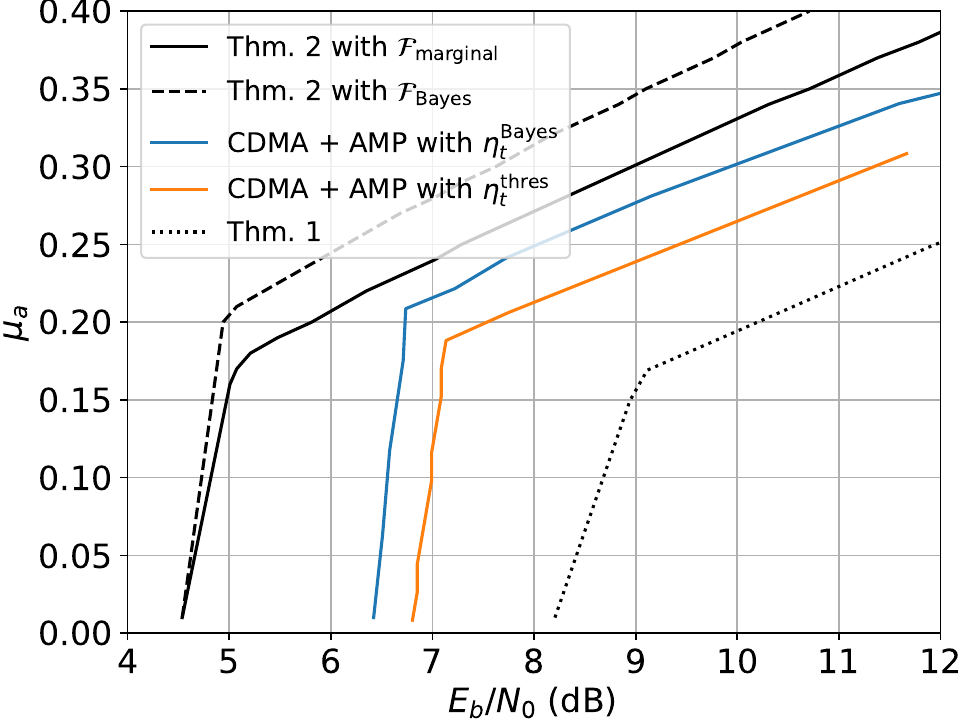}\label{fig:RA_pot_fn_ach_v_sc_se_k6} }
    \hspace*{0.01cm}
    \subfloat[$k=60$,  CDMA scheme uses $\omega=11,\Lambda=50$.]{\includegraphics[width=0.49\textwidth]{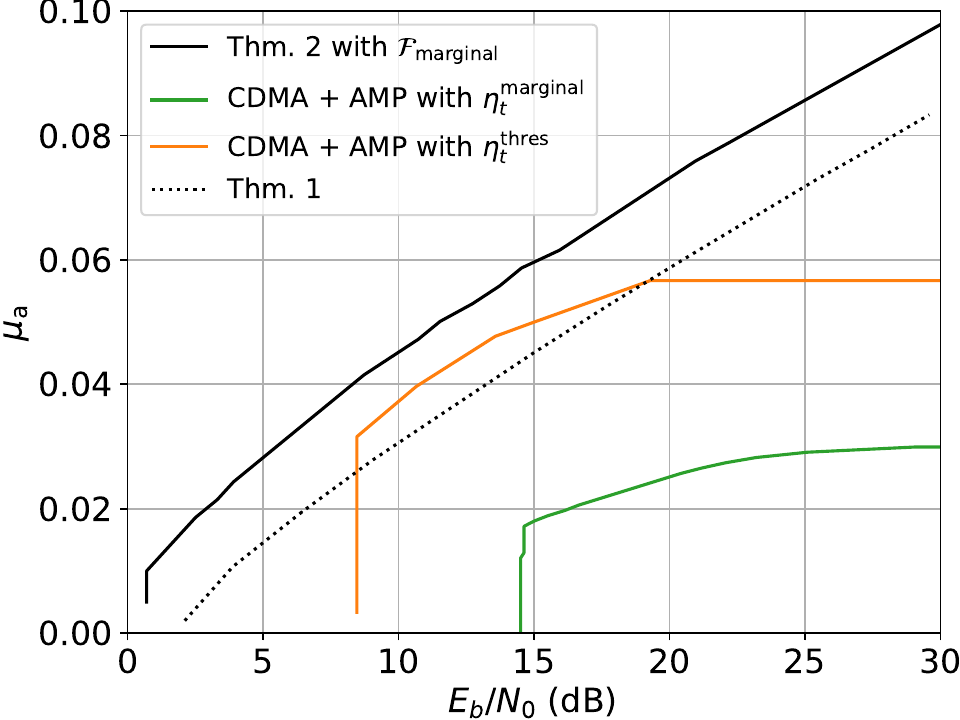}\label{fig:mua_crit_EbN0_alpha=0.7,k60} }
    \caption{Achievable $\mu_\a$  as a function of $E_b/N_0$ for a target \edit{total error $\max\{\pMD, \pFA\}+\pAUE\le 0.01$ with}
    $\alpha=0.7$, \edit{$p_{\Ka} = \text{Bin}(\alpha, L)$}, and $k=6$ or $k=60$. Comparison of the finite-length achievability bounds from Theorem \ref{thm:ngo_ach} (evaluated for \edit{$L=100$}), 
    the asymptotic achievability bounds from Theorem \ref{thm:sparc_asymp_errors} with $\Lambda\to \infty$ then $\omega\to \infty$, and the asymptotic performance characterization of the CDMA scheme  in Theorem \ref{thm:MetricsSE} with specific $(\omega, \Lambda)$. }
    \label{fig:mua_crit_EbN0_alpha=0.7_diff_k}
\end{figure}

    Fig.~\ref{fig:mua_crit_EbN0_alpha=0.7,k60} considers  a larger user payload $k=60$. The asymptotic achievability bounds  are computed  using $\mc{F}_\marginal$ because those  based on $\mc{F}_\Bayes$ are computationally infeasible for such a large $k$. 
     The CDMA schemes employ  the  thresholding denoiser $\eta_t^\thres$ or the marginal-MMSE denoiser $\eta_t^\marginal$ because the Bayes-optimal denoiser $\eta_t^\Bayes$  is  infeasible. We observe that $\eta_t^\thres$ significantly outperforms  $\eta_t^\marginal$ since it takes into account the correlation between the $k$ bipolar symbols transmitted by each user. \edit{Note that at $\mu_\a=0.013$, the transition thresholds    of the CDMA curves (i.e., 8dB and 14.5dB) in Fig.~\ref{fig:mua_crit_EbN0_alpha=0.7,k60} are consistent with the  $E_b/N_0$ values where the curves plateau in Fig.~\ref{fig:marg_v_thres}.} 
     
     \edit{Finally, the dotted black curve in Fig.~\ref{fig:mua_crit_EbN0_alpha=0.7,k60} plots the finite-length achievability bound from Theorem~\ref{thm:ngo_ach}  using $L=100, \rl=\ru=L,$  $P'$ optimized over $(0, P)$, and $(\kl, \ku)$ chosen so that $\P(\Ka\notin [\kl:\ku])\le10^{-6}$. The curve exhibits a  slope similar to the asymptotic achievability bound based on $\mc{F}_{\marginal}$ (solid black curve). Notably, for   $\mu_\a\le 0.026$ or $\mu_\a\ge0.057$, the finite-length bound achieves the target error metric over a larger range of $E_b/N_0$ values than the bounds for the efficient  CDMA schemes (orange and green curves). This confirms that  random coding with maximum-likelihood decoding  outperforms CDMA with suboptimal AMP decoding as expected.}

   \section{Proof of Theorem \ref{thm:ngo_ach}}\label{sec:ngo_full_proof}

\subsection{Preliminaries}\label{sec:preliminaries}

 For unsourced multiple-access with random user activity, where all users share the same codebook, Ngo et al.\@ recently obtained a finite-length achievability bound \cite[Theorem 1]{ngo2022unsourced}.
 We follow the structure of their proof and adapt their bounds to the standard (sourced) GMAC where users have distinct codebooks. 
 As in  \cite{polyanskiy2017perspective, ngo2022unsourced}, the proof of Theorem \ref{thm:ngo_ach} involves union bounds over error events via a change of measure, Chernoff bound and Gallager's $\rho$-trick, which we lay out as preliminaries below. 
A key additional technical step in our setting is to carefully separate error events into the three categories defined in \eqref{eq:def_my_pMD},\eqref{eq:def_my_pFA}, and \eqref{eq:def_my_pAUE}. 

\edit{The proof, except for the final step, assumes a given  Gaussian joint-codebook from the ensemble  described in Section \ref{sec:ngo_ach_scheme} and establishes  the error probability guarantees in \eqref{eq:thm_epsMD}--\eqref{eq:thm_epsAUE} individually. The final step then shows the existence of a randomized coding scheme that uses  time-sharing among  two such joint-codebooks to simultaneously achieve the bounds \eqref{eq:thm_epsMD}--\eqref{eq:thm_epsAUE}.} 

\begin{lem}[Change of measure, {\cite[Lemma 4]{ohnishi2021novel}}]\label{lem:change_measure}
Let $p$ and $q$ be two probability measures. Consider a random variable $X$ supported on a set $\mc{S}$ and a function $f:\mc{S}\to [0,1].$ It holds that $\E_p[f(X)]\le \E_q[f(X)]+d_{\TV}(p,q)$, where $d_{\TV}(p,q)$ denotes the total variation distance between $p$ and $q$, i.e., $d_{\TV}(p,q) = \frac{1}{2}\int_{\mc{S}}\abs{\frac{\mathrm{d} p}{\mathrm{d} q}-1}\mathrm{d} q$.
\end{lem}
\begin{lem}[Chernoff bound]\label{lem:chernoff}
For a random variable $X$ with moment-generating function $\E[e^{tX}]$ defined for all $\abs{t}\le b$, it holds for all $\lambda\in[0,b]$ that $\P(X\le x)\le e^{\lambda x}\E[e^{-\lambda X}]$.
\end{lem}
\begin{lem}[Gallager's $\rho$-trick] \label{lem:rho_trick}  For a series of events  $\mc{E}_1, \dots, \mc{E}_m$,  it holds that 
\begin{align}\label{eq:rho_trick}
    \P\left(\cup_{i=1}^m \mc{E}_i\right)\le \left(\sum_{i=1}^m\P(\mc{E}_i)\right)^\rho\quad \forall \rho\in[0,1].
\end{align}
Moreover, if    $\P(\mc{E}_i|V)\le e^{-E(V)}$ for $i\in[m]$, where $V$ is some random variable and $E(V)$ a  function of $V$, then \eqref{eq:rho_trick} implies that 
\begin{equation}
\P(\cup_{i=1}^m \mc{E}_i)\le m^\rho \,\E_V\left[e^{-\rho E(V)}\right].\label{eq:rho_trick_cond}
\end{equation}
\end{lem}
\begin{lem}[Gaussian identity]\label{lem:gaussian_identity}
 For $\bX\sim \normal_n(\bmu, \sigma^2\bI)$, 
 \begin{equation}
 \E\left[e^{-\gamma \|\bX\|_2^2}\right] = (1+2\sigma^2 \gamma)^{-\frac{n}{2}} \exp\left(-\frac{\gamma \|\bmu\|_2^2}{1+2\sigma^2\gamma}\right) ,\quad \forall \gamma > -\frac{1}{2\sigma^2}.   \nonumber
 \end{equation}
\end{lem}

\paragraph{Change of measure}
\label{sec:change_of_measure}
Recall the definitions of the error probabilities $\pMD, \pFA$ and $\pAUE$ from  \eqref{eq:def_my_pMD}--\eqref{eq:def_my_pAUE}. 
Since the quantities inside the expectations  in \eqref{eq:def_my_pMD}--\eqref{eq:def_my_pAUE} are  bounded between 0 and 1, using Lemma \ref{lem:change_measure}, we replace the measure over which the expectations are taken with a new measure, denoted by $q$, under which:  i) $ \Ka\in[\kl: \ku]$ almost surely,  and ii) the  codeword transmitted by user $\ell$ is simply $\bc_{w_\ell}^{(\ell)} = \tilde{\bc}_{w_\ell}^{(\ell)}$  instead of $\bc_{w_\ell}^{(\ell)}=\tilde{\bc}_{w_\ell}^{(\ell)} \ind\{\|\tilde{\bc}_{w_\ell}^{(\ell)}\|_2^2\le E_bk\}$ for every $(\ell,w_\ell)\in \W$. Recall that $\tilde{\bc}_{w_\ell}^{(\ell)}\sim \normal_n(\bzero, E_b'k/n\bI) $ and  $ P' = E_b'k/n < P=E_bk/n$. Thus, recalling $\Ka := |\W|$,  the  total variation distance  $d_{\TV}$ between the new and original measures is upper bounded by 
\begin{align}
    d_{\TV} &\le  \P(\Ka\notin [\kl: \ku]) + 
\E_{\Ka}\left[\P\left(\bigcup_{\ell: (\ell, w_\ell)\in \W}\|\bc_{w_\ell}^{(\ell)}\|_2^2 >E_bk\, \bigg| \,\Ka\right)\right]\nonumber\\
&\stackrel{\text{(a)}}{\le} \P(\Ka\notin [\kl:\ku]) + \E_{\Ka} \left[\sum_{\ell: (\ell, w_\ell)\in \W}\P\left(\|\bc_{w_\ell}^{(\ell)}\|_2^2>nP\right)\right]\nonumber\\
    &\stackrel{\text{(b)}}{=}\P(\Ka\notin [\kl: \ku]) + \E[\Ka]\frac{\Gamma(\frac{n}{2}, \frac{nP}{2P'})}{\Gamma(\frac{n}{2})}=: \tp\,, \label{eq:tilde_p}
\end{align}
where (a) applies a union bound and uses 
\[\P\left(\|\bc_{w_\ell}^{(\ell)}\|_2^2>E_bk\, \big|\,\Ka\right) = \P\left(\|\bc_{w_\ell}^{(\ell)}\|_2^2>E_bk\right)=\P\left(\|\bc_{w_\ell}^{(\ell)}\|_2^2>nP\right),\]
 and (b) holds because $\frac{1}{P'}\|\bc_{w_\ell}^{(\ell)}\|_2^2\sim \textrm{Gamma}(\frac{n}{2}, 2)$ and $\P\left(\textrm{Gamma}(n', \theta)>z_0\right) =\Gamma(n', \frac{z_0}{\theta})/\Gamma(n')$. 
 Therefore,  Lemma \ref{lem:change_measure} implies that  
\begin{align}
    \pMD&\le \E_q\left[\ind\{\Ka\neq 0\}\cdot  \frac{1}{\Ka}\sum_{\ell: (\ell,w_\ell)\in \W}\ind\{\widehat{w_\ell} = \emptyset\}\right] + \tp\,, \label{eq:change_measure_pMD}\\
    \pFA &\le\E_q\left[\ind\{\widehat{\Ka}\neq 0\}\cdot  \frac{1}{\widehat{\Ka}}\sum_{\ell:(\ell, \hw_\ell)\in\hW}\ind\{w_\ell=\emptyset\}\right]+ \tp\,,\label{eq:change_measure_pFA}\\
    \pAUE &\le \E_q\left[\ind\{\Ka\neq 0\}\cdot  \frac{1}{\Ka}\sum_{\ell:(\ell, w_\ell)\in \W}\ind\big\{\widehat{w_\ell}\notin\{w_\ell, \emptyset\}\big\}\right]+ \tp\,.\label{eq:change_measure_pAUE}
\end{align} 
 In the next subsection, we   evaluate the $ \E_q[\cdot]$ terms on the right side  in a simple special case, with the calculations for  the general case deferred to Section \ref{sec:general_case}.

     \subsection{A  Special Case}\label{sec:special_case}
  Under the new measure $q$,  we begin by analyzing the error probabilities for a special case  where $\Ka=\ka, \Kap=\kap$ are both assumed fixed, \edit{$\rl=\ru=0$} and  $\kap\le \ka\le L$.
  As a result, $ \ukap= \okap=\kap $ and $\hKa :=|\hW|=\kap$ in \eqref{eq:min_distance_decoder}. 

  The Venn diagram in Fig.~\ref{fig:special_case_venn}  illustrates the relation between different sets of  codewords in this special  case. The intersection $\W\cap \hW$ denotes the set of  transmitted codewords that are decoded correctly.  
   The  set in gray, $\W_\iMD$, is called the  \emph{initial} set of $\MD$s, which represents the  $(\ka-\kap)$ transmitted codewords that are guaranteed to be \edit{missed} simply  because $\kap\le \ka$.

   In addition to $\W_\iMD$, the remaining sets $(\W\setminus \hW) \setminus \W_\iMD=: \We$ and  $\hW\setminus\W=:\hWe$ represent the \emph{extra} codewords in error. The  sets in blue, $\W_\eMD$ and $\hW_\eFA$, represent the  extra  $\MD$s and $\FA$s. Their elements have a one-to-one correspondence  because each  active user declared silent corresponds to a silent user declared active. Therefore  these sets share the same size $\abs{\W_\eMD}=|\hW_\eFA|$. 
\begin{figure}[t!]
    \centering
    \includegraphics{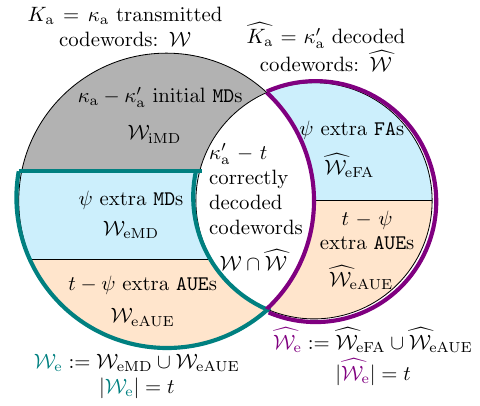}
    \caption{Venn diagram illustrating the relation between different sets of messages in the special case defined in Section \ref{sec:special_case}, where $\Ka=\ka, \Kap=\kap$ are fixed, \edit{$\rl=\ru=0$} and $\kap\le \ka$, and integers $t\in[0:\kap]$ and $ \psi\in[0:\min\{t,L-\ka\}]$. Diagram analogous to \cite[Fig.\@ 7]{ngo2022unsourced}. }
    \label{fig:special_case_venn}
    \end{figure}

   The  sets in orange, $\W_\eAUE$ and $\hW_\eAUE$, represent the extra active user errors ($\AUE$s), where  the  users who transmitted the codewords  $\W_\eAUE$ are declared active, but their codewords are mistaken for  some other codewords $\hW_\eAUE$ in their respective codebooks. Thus, these sets also have the same size $\abs{\W_\eAUE}=|\hW_\eAUE|$. Noting that  $\We=\W_\eMD\cup \W_\eAUE$ and $\hWe=\hW_\eFA\cup \hW_\eAUE$, we have $|\We|=|\hWe|$.
    In terms of notation, the  subscripts ``i'' and ``e'' denote, respectively, the initial errors due to $\hKa = \kap$ being different from $\Ka=\ka$, and the extra errors occurred in the further decoding process. 

 Based on the relation between different sets of messages illustrated in Fig.~\ref{fig:special_case_venn}, and the fact that $\Ka=\ka, \Kap=\kap $ are fixed in the special case,  we  simplify the right side of \eqref{eq:change_measure_pMD}--\eqref{eq:change_measure_pAUE} as follows: 
   \begin{align}
       \pMD &\le  
       \E_q\left[\frac{ \abs{\W_\iMD}+ \abs{\W_\eMD}}{\Ka}\right] +\tp  = 
       \frac{ (\ka-\kap)+ \E_q\abs{\W_\eMD}}{\ka} +\tp \label{eq:pMD_special_case},\\
       \pFA &
       \le \E_q\left[\frac{|\hW_\eFA|}{\Kap}\right] +\tp
       =\frac{\E_q |\W_\eMD|}{\kap} +\tp \label{eq:pFA_special_case},\\
       \pAUE &
       \le \E_q\left[\frac{|\W_\eAUE|}{\Ka}\right]
       =\frac{\E_q\left[\abs{\We}-\abs{\W_\eMD}\right]}{\ka} +\tp \label{eq:pAUE_special_case}.
\end{align}
Note that   $\E_q\abs{\W_\eMD} =\E_q\left[\E_q \abs{\W_\eMD}\big|\abs{\We}\right]$ by the law of total expectation. 
Therefore, to evaluate the right side of \eqref{eq:pMD_special_case}--\eqref{eq:pAUE_special_case}, we only need  to characterize the marginal distribution of  $\abs{\We}$, and the  distribution of  $\abs{\W_\eMD}$ conditioned on $\abs{\We}$. We first provide these distributions in Lemmas  \ref{lem:p_W_aMD_union_W_AUE} and \ref{lem:cond_p_aMD}, and then prove the lemmas. For simplicity, we will henceforth drop the subscript $q$ from $\E_q[\cdot]$, and stress that all expectations are  with respect to the new measure $q$.

\begin{lem}\label{lem:p_W_aMD_union_W_AUE}
       In the special case, where  $\Ka=\ka$ and $\Kap=\kap$ are fixed, \edit{$\rl=\ru=0$} and 
       $\kap\le\ka\le L$, it holds  for every integer $t \in [0: \kap]$ that
\begin{equation}\label{eq:p_Wa=t_special_case_lem}
       \P\left(\abs{\We} = t\right) \le \exp\left(-\frac{n}{2} \check{E}(t)\right),
   \end{equation}
   where
   \begin{align}
       &\check{E}(t) = \max_{\rho,\rho_1\in[0,1]} \big[  -\rho\rho_1\check{R}_1(t) - \rho_1\check{R}_2(t) + \check{E}_0(\rho,\rho_1) \big],\label{eq:E(t,t)_special_case_lem}\\
       &\check{R}_1(t)= \frac{2}{n} \left(t\ln M + \ln\binom{\check{\mc{R}}}{t}\right),\quad \check{\mc{R}}=L-\ka+t, \quad 
       \check{R}_2(t)= \frac{2}{n}\ln\binom{\kap}{t}, \label{eq:R1_R2_special_case_lem}\\
       &\check{E}_0(\rho,\rho_1) = \max_{\check{\lambda}>-\frac{1}{P't}} \big[ \rho_1 \check{a}(\rho, \check{\lambda}) +\ln(1-\rho_1\check{P}_1\check{b}(\rho, \check{\lambda})) \big],\label{eq:E0_special_case_lem}\\
        &\check{a}(\rho, \check{\lambda}) = \rho\ln(1+ P't\check{\lambda})+\ln\left(1+ P't \check{\chi}(\rho, \check{\lambda})\right),
       \label{eq:a0_b0_lem}\\
       &\check{b}(\rho, \check{\lambda}) =\rho\check{\lambda} - \frac{\check{\chi}(\rho, \check{\lambda})}{1+P't\check{\chi}(\rho, \check{\lambda})},\qquad \check{\chi}(\rho, \check{\lambda}) = \frac{\rho\check{\lambda}}{1+P't\check{\lambda}},\label{eq:chi_special_case_lem}\\
    &\check{P}_1 = (\ka-\kap)P'+1\,.
       \label{eq:mu0_P2_lem}
   \end{align}
   \end{lem}

\begin{lem}\label{lem:cond_p_aMD}Consider the special case, where  $\Ka=\ka$ and $\Kap=\kap$ are fixed, \edit{$\rl=\ru=0$} and 
       $\kap\le\ka\le L$. Define $\check{\mc{R}}=L-\ka+t$. For every integer  $t\in[0:\kap]$ and $ \psi\in[0:\min\{t,\check{\mc{R}}-t\}]$, let 
       \begin{align}
  \nu(t,\psi)= \binom{\check{\mc{R}}-t}{\psi}M^\psi \cdot \binom{t}{t-\psi}(M-1)^{t-\psi},\quad 
    \nu(t) = \sum_{\psi=0}^{\min\{t,\check{\mc{R}}-t\}} \nu(t,\psi).
\label{eq:u(t,psi),u(t)_special_case}
\end{align}
Then the following hold:
       \begin{enumerate}[(1)]
           \item \label{lem:num_values_of_We}\edit{Given $\We$ with size $|\We|=|\hWe|=t$, the number of distinct choices for the set  $\hWe$   equals $\nu(t)$. Furthermore, $\nu(t)\le \binom{\check{\mc{R}}}{t}M^t$.}
    \item\label{lem:cond_prob_psi} It holds that  $\P\left(|\W_\eMD|=\psi \,\bigg|\, |\We|=t\right) = \nu(t, \psi)/\nu(t)$.
       \end{enumerate}
\end{lem}
 
Substituting \eqref{eq:p_Wa=t_special_case_lem}--\eqref{eq:u(t,psi),u(t)_special_case} into \eqref{eq:pMD_special_case}--\eqref{eq:pAUE_special_case} yields Theorem \ref{thm:ngo_ach} for  the special case. We next provide the proofs of Lemmas \ref{lem:p_W_aMD_union_W_AUE} and \ref{lem:cond_p_aMD}.
\qed

\begin{proof}[Proof of Lemma \ref{lem:p_W_aMD_union_W_AUE}]
We recall from \eqref{eq:min_distance_decoder} the notation $c(\W'):=\sum_{j:(j,w_j')\in \W'} \bc_{w_j'}^{(j)}\,$. The channel output  takes the form
\begin{equation}\label{eq:def_y_special_case}
           \by = c(\W) +\bveps =c(\W_\iMD) +c(\We) +c(\W\cap \hW) +\bveps\,,
       \end{equation}
       where we have used the fact  $\W= \W_\iMD \,  \cup  \, \We  \, \cup \,  (W\cap \hW)$, \edit{where the subsets forming the union are pairwise disjoint} (see Fig.~\ref{fig:special_case_venn}). 
       Since the set decoded by the maximum-likelihood decoder can be expressed as $\hW = \hWe \cup (\W\cap \hW) $, we have 
       $$\|\by - c(\hWe)-c(\W\cap \hW)\|_2^2 \, < \, \|\by-c(\We)-c(\W\cap \hW)\|_2^2\,, $$ where we have used $|\We| = |\hWe|$. 
       Combined with 
       \eqref{eq:def_y_special_case}, this is equivalent to
\begin{equation}\label{eq:special_case_error_criterion}
    \|c(\W_\iMD)+c(\We)-c(\hWe)+\bveps\|_2^2<\|c(\W_\iMD)+\bveps\|_2^2\,   . 
       \end{equation}
       Letting $\mc{E}(\W_\iMD, \We, \hWe)$ denote the event that  $(\W_\iMD, \We, \hWe)$ satisfies \eqref{eq:special_case_error_criterion} with $|\We|=|\hWe|=t$,  we have 
       \begin{align}\label{eq:p_Wa=p_union_E}
           \P\left(\abs{\We}=t\right)
    =\P\left(\bigcup_{\We: \abs{\We}=t}\,\bigcup_{\hWe: |\hWe|=t} \mc{E}(\W_\iMD, \We, \hWe)\right), \quad \text{for}\; t\in [0:\kap].
       \end{align}
       To bound the right side of \eqref{eq:p_Wa=p_union_E}, we first note that
       \begin{align}
           &\P\left(\mc{E}(\W_\iMD, \We, \hWe)\,\bigg|\, \W_\iMD, \We, \bveps\right)\nonumber\\
           & = \P\left(\|c(\W_\iMD)+c(\We)-c(\hWe)+\bveps\|_2^2<\|c(\W_\iMD)+\bveps\|_2^2 \,\bigg| \, \W_\iMD, \We, \bveps\right)\nonumber\\
           & \stackrel{\text{(a)}}{\le }\exp\left(\check{\lambda}_0 \|c(\W_\iMD)+\bveps\|_2^2\right)\E_{\hWe}\left[\exp\left(-\check{\lambda}_0 \|c(\W_\iMD)+c(\We)-c(\hWe)+\bveps\|_2^2\right)\,\bigg|\, \W_\iMD, \We, \bveps\right]\nonumber\\
           &\stackrel{\text{(b)}}{=}\exp\left(\check{\lambda}_0\|c(\W_\iMD)+\bveps\|_2^2\right)
           (1+2P't\check{\lambda}_0)^{-\frac{n}{2}}\exp\left(-\frac{\check{\lambda}_0\|c(\W_\iMD)+c(\We)+\bveps\|_2^2}{1+2P't\check{\lambda}_0}\right), \; \forall \check{\lambda}_0>-\frac{1}{2P't}\nonumber\\
           & =: \exp\left(-\textsf{A}(\W_\iMD, \We, \bveps)\right),\label{eq:base_event_E_cond_all}
       \end{align}
       where (a) applies the  Chernoff bound in Lemma \ref{lem:chernoff}, and (b) applies the Gaussian identity in Lemma \ref{lem:gaussian_identity},  noting that  $c(\hWe)\sim \normal_n(\bzero, P't\bI)$.
We now apply Gallager's $\rho$-trick   \eqref{eq:rho_trick_cond}  twice to add two unions $\bigcup_{\hWe}$ and $\bigcup_{\We}$  to the left side of \eqref{eq:base_event_E_cond_all} and obtain an upper  bound on \eqref{eq:p_Wa=p_union_E}.

\paragraph{$\rho$-trick for   $\bigcup_{\hWe}$} We first apply \eqref{eq:rho_trick_cond} to add the union over $\hWe$ to the left side of \eqref{eq:base_event_E_cond_all},  obtaining   for every $\rho\in [0,1],$
       \begin{align}
    \P\left(\bigcup_{\hWe: |\hWe|=t}\mc{E}(\W_\iMD, \We, \hWe)\,\bigg| \, \W_\iMD, \We, \bveps\right)
        \le  \left[\binom{\check{\mc{R}}}{t} M^{t}\right]^\rho \exp\left(-\rho \textsf{A}(\W_\iMD, \We, \bveps)\right)\label{eq:apply_rho_trick_1},
       \end{align}
      where  recall that $M$ denotes the number of messages in each codebook. \edit{Given $\We$ with size $|\We|=t$, the multiplicative factor $\binom{\check{\mc{R}}}{t}M^t$ 
      is an upper bound on the number of ways to construct $\hWe$ given $\hWe$ with size $|\hWe|=t$, as we   prove in Lemma \ref{lem:cond_p_aMD}, Part \ref{lem:num_values_of_We}}. 
      %
%
    Expanding the right side of \eqref{eq:apply_rho_trick_1} \edit{and taking expectation over $\We$ yields}
       \begin{align}
    &  \P\left(\bigcup_{\hWe: |\hWe|=t}\mc{E}(\W_\iMD, \We, \hWe)\,\bigg| \, \W_\iMD, \bveps\right) \nonumber\\
    & \le \binom{\check{\mc{R}}}{t}^\rho M^{\rho t}
        (1+2P't\check{\lambda}_0)^{-\frac{n\rho}{2}}
\exp\left(\rho\check{\lambda}_0 \|c(\W_\iMD)+\bveps\|_2^2\right) \cdot \nonumber\\
& \quad\;\;  \E_{\We}\left[\exp\left(-\frac{\rho\check{\lambda}_0 \|c(\W_\iMD)+c(\We)+\bveps\|_2^2}{1+2P't\check{\lambda}_0}\right)\,\bigg|\, \W_\iMD, \bveps\right]\nonumber\\
& \stackrel{\text{(a)}}{=}  \binom{\check{\mc{R}}}{t}^\rho M^{\rho t}
        (1+2P't\check{\lambda}_0)^{-\frac{n\rho}{2}}
\exp\left(\rho\check{\lambda}_0 \|c(\W_\iMD)+\bveps\|_2^2\right)\cdot \nonumber\\
& \quad \;\; 
\left[ (1+2P't\check{\chi}_0)^{-\frac{n}{2}}\exp\left(-\frac{\check{\chi}_0\|c(\W_\iMD)+\bveps\|_2^2}{1+2P't\check{\chi}_0}\right) \right]\nonumber\\
&\stackrel{\text{(b)}}{=}  \binom{\check{\mc{R}}}{t}^\rho M^{\rho t}
\exp\left(-\frac{n\check{a}_0}{2}+\check{b}_0\|c(\W_\iMD)+\bveps\|_2^2\right),\label{eq:apply_rho_trick_1_ext}
       \end{align}
       where (a) applies the Gaussian identity in Lemma \ref{lem:gaussian_identity} and the shorthand $\check{\chi}_0:=\rho\check{\lambda}_0/(1+2P't\check{\lambda}_0)$, noting that $c(\We)\sim \normal_n(\bzero, P't\bI)$; (b) uses the 
 shorthand  $\check{b}_0: = \rho\check{\lambda}_0 - {\check{\chi}_0}/(1+2P't\check{\chi}_0)$ and $\check{a}_0:=\rho\ln(1+2P't\check{\lambda}_0)+\ln(1+2P't\check{\chi}_0)$.

\paragraph{$\rho$-trick for   $\bigcup_{\We}$}   We now apply the $\rho$-trick \eqref{eq:rho_trick_cond} again to add the  union  over $ \We$ to the left side of \eqref{eq:apply_rho_trick_1_ext}.  \edit{Given $\W_\iMD$}, since $\We\subset \W\setminus \W_{\iMD}$  and $|\W\setminus \W_{\iMD}| = \kap$, there are $\binom{\kap}{t}$ ways to construct $\We$ with size $t$. Therefore, for $\rho_1\in[0,1]$, we have
    \begin{align}
& \P\left(\bigcup_{\We:\abs{\We}=t}\bigcup_{\hWe: |\hWe|=t} \mc{E}(\W_\iMD, \We, \hWe)\,\bigg|\, \W_\iMD, \bveps \right) \nonumber\\
& \le  \binom{\kap}{t}^{\rho_1} \binom{\check{\mc{R}}}{t}^{\rho\rho_1} M^{\rho\rho_1 t}
\exp\left(-\frac{n\rho_1\check{a}_0}{2}+\rho_1\check{b}_0\|c(\W_\iMD)+\bveps\|_2^2\right) \,.
\end{align}
\edit{Taking expectation over $\W_\iMD$, and $\bveps\sim\mc{N}_n(\bzero, \bI)$,  we obtain}
\begin{align}
&\P\left(\bigcup_{\We:\abs{\We}=t}\bigcup_{\hWe: |\hWe|=t} \mc{E}(\W_\iMD, \We, \hWe)\right)\nonumber\\
       & \le \binom{\kap}{t}^{\rho_1} \binom{\check{\mc{R}}}{t}^{\rho\rho_1} M^{\rho\rho_1 t} 
\exp\left(-\frac{n\rho_1\check{a}_0}{2}\right)
\E_{\W_\iMD, \bveps}\left[\exp\left(\rho_1\check{b}_0\|c(\W_\iMD)+\bveps\|_2^2\right)\right]\nonumber\\
& = \binom{\kap}{t}^{\rho_1} \binom{\check{\mc{R}}}{t}^{\rho\rho_1} M^{\rho\rho_1 t} 
\exp\left(-\frac{n\rho_1\check{a}_0}{2}\right)
\left[1-2\rho_1 \check{P}_1\check{b}_0\right]^{-\frac{n}{2}}
\label{eq:p_double_unions_E_expanded_special}.
    \end{align}
   The last step in \eqref{eq:p_double_unions_E_expanded_special} applies the Gaussian identity in Lemma \ref{lem:gaussian_identity} to $c(\W_\iMD) +\bveps\sim \normal_n(\bzero, \check{P}_1\bI)$ with  $\check{P}_1=\abs{\W_\iMD}P'+1=(\ka-\kap)P'+1$ as defined in \eqref{eq:mu0_P2_lem}.

    Finally, substituting \eqref{eq:p_double_unions_E_expanded_special}  into \eqref{eq:p_Wa=p_union_E} \edit{and minimizing the upper bound with respect to $\rho, \rho_1\in[0,1]$ and $\check{\lambda}_0>-\frac{1}{2P't}$} yields
    \begin{align}
    &    \P\left(\abs{\We}=t\right) \le 
      \min_{\substack{\rho,\rho_1\in[0,1],\\ \check{\lambda}_0>-\frac{1}{2P't}}  }\binom{\kap}{t}^{\rho_1}\binom{\check{\mc{R}}}{t}^{\rho\rho_1} M^{\rho\rho_1 t} 
\exp\left(-\frac{n\rho_1\check{a}_0}{2}\right)
\left[1-2\rho_1\check{P}_1\check{b}_0\right]^{-\frac{n}{2}}\nonumber\\
& =\min_{\substack{\rho,\rho_1\in[0,1],\\\check{\lambda}>-\frac{1}{P't}} } \binom{\kap}{t}^{\rho_1} \binom{\check{\mc{R}}}{t}^{\rho\rho_1} M^{\rho\rho_1 t} 
\exp\left(-\frac{n\rho_1 \check{a}}{2}\right)\left[1-\rho_1\check{P}_1\check{b}\right]^{-\frac{n}{2}} = \exp\left(-\frac{n}{2}\check{E}(t)\right),\label{eq:p_Wa_special_case_proof}
    \end{align}
 where  $\check{\lambda}:=2\check{\lambda}_0$, $\check{E}(t)$ is given in \eqref{eq:E(t,t)_special_case_lem}--\eqref{eq:E0_special_case_lem},  and $ \check{a}, \check{b}, \check{\chi}$ are given in \eqref{eq:a0_b0_lem}--\eqref{eq:chi_special_case_lem}, which implies that $ \check{a}_0=\check{a}$, $\check{b}_0=\check{b}/2$ and $\check{\chi}_0=\check{\chi}/2$.
\end{proof}

   \begin{remark}[Comparisons with 
   unsourced setting]\label{rmk:ours_vs_ngo}\normalfont For unsourced random access, the $\binom{\check{\mc{R}}}{t} M^t$ factor in  \eqref{eq:apply_rho_trick_1} needs to be  replaced by  $\binom{M-\ka}{t}$, as in \cite[Eq.\@ (82)]{ngo2022unsourced}. This is because in unsourced random access,  the $t$ elements of $\hWe$ are chosen from the $(M-\ka)$  codewords that are not transmitted in the common codebook. 
   \end{remark}

\begin{proof}[Proof of Lemma \ref{lem:cond_p_aMD}]

\edit{\textbf{Proof of part \ref{lem:num_values_of_We}:}} \edit{Given $\We$ with size $|\We|=t$,} we observe from Fig.~\ref{fig:special_case_venn} that each element of $\hWe$  comes from  either:
\begin{enumerate}[i), leftmargin=0.7cm, labelindent=0.7cm]
    \item \label{item:silent_special}one of the   codebooks of the $L-\ka=\check{\mc{R}}-t$ silent users, or
    \item \label{item:active_special} the set of untransmitted codewords in one of the codebooks of the $t$  users who transmitted $\We$.
\end{enumerate}
We refer to these sets as type-\ref{item:silent_special}  and type-\ref{item:active_special}, respectively.
Noting that $\hWe=\hW_\eFA\cup \hW_\eAUE$,   each element of $\hW_\eFA$ comes from a type-\ref{item:silent_special} codebook and each element of $\hW_\eAUE$ comes from a type-\ref{item:active_special} codebook  excluding $\We$. The decoding rule \eqref{eq:min_distance_decoder} ensures that all the elements of $\hWe=\hW_\eFA\cup \hW_\eAUE$ come from distinct codebooks.  
Hence,  to have $|\hWe|=t$ with   $|\hW_\eFA|=\psi$ and $|\hW_\eAUE|=t-\psi$,  for  $\psi\in [0:\min\{t, L-\ka\}]$, $\hWe$ needs to contain  $\psi$ codewords, each from  a distinct  type-\ref{item:silent_special} codebook, and $(t-\psi)$ codewords,  each from a distinct type-\ref{item:active_special} codebook. There are 
 \begin{equation}\label{eq:nu_t_psi_special}
     \nu(t,\psi) = \binom{L-\ka}{\psi}M^\psi \cdot \binom{t}{t-\psi}(M-1)^{t-\psi}
 \end{equation}
 ways to construct such a set  $\hWe$.  \edit{Moreover, recalling $\check{\mc{R}}=L-\ka+t$,} there are
 \begin{align}\label{eq:nu_t_special_def}
\nu(t)\coloneqq\sum_{\psi=0}^{\min\{t, L-\ka\}}\nu(t,\psi) \le \sum_{\psi=0}^{\min\{t, L-\ka\}}\binom{L-\ka}{\psi} \cdot \binom{t}{t-\psi} \cdot M^{t} = \binom{\check{\mc{R}}}{t} M^t 
 \end{align}
 ways in total  to construct  $\hWe$ for any valid values of  $ \psi$, \edit{where the inequality in \eqref{eq:nu_t_special_def} holds due to Vandermonde's identity.}

\paragraph{\edit{Proof of part \ref{lem:cond_prob_psi}:}}  We characterize the conditional distribution of  $|\W_{\eMD}|$ 
by characterizing the conditional distribution of $|\hW_{\eFA}|$ since  $ |\W_{\eMD}|=|\hW_{\eFA}|$ (see Fig.~\ref{fig:special_case_venn}).  \edit{The key is to leverage the combinatorial argument established in  part \ref{lem:num_values_of_We}. Given $\We$, we know there are  $\nu(t,\psi)$ possible configurations for $\hW_{\eFA}$  with size $|\hW_{\eFA}| = \psi$, and $\nu(t)$ total possible configurations for $\hWe$ with size  $|\hWe| = t$.} 
%
Since all codewords  are exchangeable across codebooks, each possible configuration of $\hWe$ occurs with equal probability. 
We therefore have
\begin{align}
    \P\left(|\hW_\eFA| = \psi\,\bigg|\,  |\hWe|=t, \We\right) = \frac{\nu(t,\psi)}{\nu(t)}\label{eq:u_tb/ut_special_case}\,.
\end{align}
 Finally, using  $|\W_\eMD|=|\hW_\eFA|$ and $ |\We|= |\hWe| $ \edit{and taking expectation over $\We$ concludes the proof.} 
\end{proof}

  \subsection{The General Case}\label{sec:general_case}
\begin{figure}[t!]
    \centering
    \hspace*{1cm}\includegraphics[width=0.77\columnwidth]{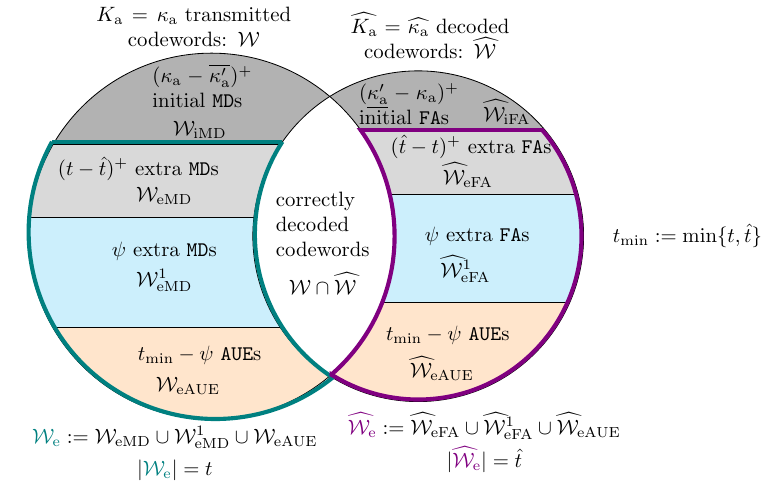}\hspace*{-1cm}
    \caption{Venn diagram illustrating the relation between different sets of messages in the general case, where $\Kap, \Ka$ are random and  \edit{$\rl, \ru\ge 0$}. Diagram analogous to \cite[Fig.\@ 8]{ngo2022unsourced}. Recall notation $x^+=\max\{x,0\}$.}
    \label{fig:general_case_venn}
    \end{figure}
    
This section provides the proof of  Theorem \ref{thm:ngo_ach} in the general case, where  $\Ka$, $\Kap$ (and hence $\hKa$) are  random, and the decoding \edit{radii $\rl, \ru\ge 0$}. As before, we use $\ka$, $\kap, \ukap, \okap$ and $\hka$ to denote the realization of random variables $\Ka$, $\Kap, \uKap, \oKap$ and $\hKa$.

Fig.~\ref{fig:general_case_venn} illustrates the relation between different sets of codewords. 
As before, let  $\W$ be the  set of transmitted codewords with size $\abs{\W}=\ka$ in a particular instance.  
Recall that given \edit{$\rl$ and $\ru$}, the decoder  computes $\kap$ in \eqref{eq:determine_Ka'} and determines  the decoded set $\hW$ in \eqref{eq:min_distance_decoder} with size $|\hW|=\hka\in[\ukap: \okap]$, where  $\ukap$ and $\okap$  are functions of $\kap$.
%
Similar to the special case, the difference between $\ka$ and $\hka$ leads to  $|\ka-\hka|$  $\MD$s or   $\FA$s, depending on whether or not $\ka > \hka$.
 However,  given $\ka $ and $\kap$ (hence $\ukap$ and $\okap$), $\hka$ is unknown and the only knowledge  we have about $\hka$ is that  $\hka\in[\ukap: \okap]$. Therefore, 
 \begin{enumerate}[a)]
     \item \label{enum:a} If $\ka >\okap$, we are guaranteed with $(\ka-\okap)$  $\MD$s initially, represented by $\W_{\iMD}$.
     \item \label{enum:b} If $\ka< \ukap $, we are guaranteed with $(\ukap - \ka)$ $\FA$s initially, represented by $\hW_{\iFA}$.
     \item \label{enum:c} If $\ka\in [\ukap: \okap]$,  the relation between $\ka$ and $\hka$ is unknown, so we cannot declare any  $\MD$s or $\FA$s initially. 
 \end{enumerate}
To ensure our results hold across  all three scenarios \ref{enum:a}, \ref{enum:b} and \ref{enum:c}, we denote $|\W_{\iMD}| $ by $(\ka-\okap)^+$ and $ |\hW_{\iFA}| $ by $ (\ukap-\ka)^+$. The $\{\cdot\}^+$ operators ensure that at most one of $\W_{\iMD}$ and $\hW_{\iFA}$ is non-empty.

As before, $\W\cap \hW$ represents the set of transmitted codewords decoded correctly. The sets 
$\We:=(\W\setminus\hW)\setminus \W_\iMD$ and 
$\hWe:=(\hW\setminus \W)\setminus \hW_\iFA$ correspond to any additional errors. $\We$ contains the extra $\MD$s and $\AUE$s, while $\hWe$ contains the  extra $\FA$s and $\AUE$s. 
In contrast to the special case where $|\We|=|\hWe|$, in the general case $\We$ and $\hWe$ may differ in size. We  denote their sizes by  $|\We|=t$ and $|\hWe|=\hatt$.
It follows   that there are $|t-\hatt|$ additional $\MD$s or $\FA$s depending on whether or not  $t> \hatt$. These extra errors are represented by $\W_\eMD$ and $\hW_\eFA$, respectively, with sizes $ |\W_{\eMD}| = (t-\hatt)^+$ and $ |\hW_{\eFA}| = (\hatt-t)^+$.

Having defined  $\W_\iMD, \hW_\iFA, \W_\eMD,$ and $\hW_\eFA,$ it is easy to see that the remaining sets  $\We\setminus\W_\eMD$ and $ \hWe\setminus \hW_\eFA$ share the same size $\tmin:=\min\{t, \hatt\}$. 
Similar to the special case, $\We\setminus\W_\eMD$ and $ \hWe\setminus \hW_\eFA$ each  contains two types of errors, coloured in blue and orange, respectively.
The sets in blue, denoted by  $\W_\eMD^1$ and $ \hW_\eFA^1$, are the  extra $\MD$s and $\FA$s. Their elements \edit{are in one-to-one correspondence with $|\W_\eMD^1| = |\hW_\eFA^1|$: each  transmitted codeword  $w\in\W_\eMD^1$ corresponds to  a unique decoded codeword  $\hw\in\hW_\eFA^1$ from a silent user's codebook, where   $w$ is mistakenly decoded as $\hw$.} 
The sets in orange, denoted by $\W_\eAUE$ and  $\hW_\eAUE$, are the  extra  $\AUE$s.  Their elements \edit{are also in one-to-one correspondence with $ |\W_\eAUE| = |\hW_\eAUE|$, because for every transmitted codeword  $w\in\W_\eAUE$, there is a unique decoded codeword  $\hw\in\hW_\eAUE$ from  the same codebook as $w$ such  that  $w$ is decoded as $\hw$.}

\begin{figure}[t!]
    \centering
    \subfloat[$\ka >\okap$ implies  $t\ge \hatt$.]{\includegraphics[width=0.48\linewidth]{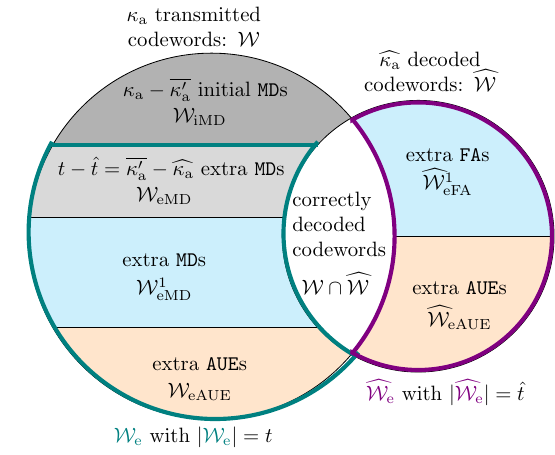}\label{fig:general_case_venn_larger_Ka}}
    \hfill
    \subfloat[$\ka<\ukap$ implies  $t\le \hatt$.]{\includegraphics[width=0.45\linewidth]{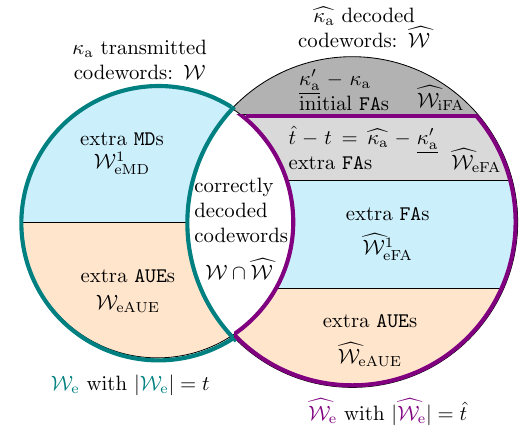}
    \label{fig:general_case_venn_larger_hKa}}\\
    \subfloat[{$\ka \in  [\ukap: \okap ]$} allows either $t\ge \hatt$ or $t< \hatt$. ]{\includegraphics[width=0.45\linewidth]{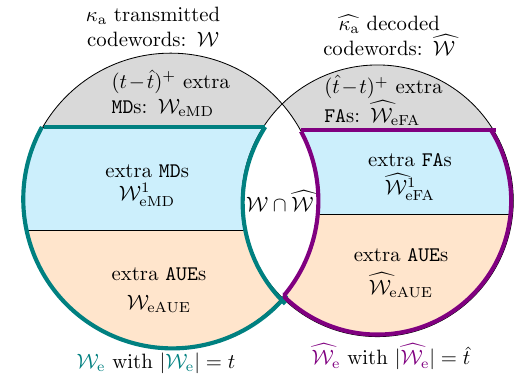}
    \label{fig:general_case_venn_ambiguous}}
    \caption{Venn diagrams illustrating the relation between different sets of messages in the general case, where $\Kap, \Ka$ are random and  \edit{$\rl, \ru\ge 0$}. Subfigures correspond to three different scenarios: $\ka >\okap$, $\ka<\ukap$, or $ \ka\in [\ukap: \okap]$.}
    \label{fig:venn_general_specialized}
    \end{figure}

\begin{remark}[Different values of $\ka$]\normalfont
    Note that $\ka> \okap $ implies $t\ge \hatt$, and 
 $\ka< \ukap$ implies $t\le \hatt$, whereas  $ \ka\in [\ukap: \okap]$ allows $t$ and $\hatt$ to have either relation. We verify the first claim  as an example: when $\ka> \okap\ge \hka$,  we have   $\ka - \hka =|\W_\iMD| + |\W_{\eMD}|  $, which implies that $|\W_{\eMD}| = (\ka - \hka) - |\W_{\iMD}| = \okap-\hka =|\We\cup (\W\cap \hW)| - |\hWe\cup (\W\cap \hW)| = t-\hatt\ge 0$ (see Fig.~\ref{fig:general_case_venn_larger_Ka}).
 The second and third 
 claims can be verified similarly.
 Fig.~\ref{fig:venn_general_specialized}  illustrates how the relation between different sets of messages  shown in  Fig.~\ref{fig:general_case_venn} specializes  to the three different cases: $\ka> \okap$, $\ka< \ukap$, or $ \ka\in [\ukap: \okap]$.
\end{remark}

We now proceed to prove  Theorem \ref{thm:ngo_ach} in the general case. First note that the  right side of \eqref{eq:change_measure_pMD}--\eqref{eq:change_measure_pAUE}  can be rewritten as:
\begin{align}
    \pMD &\le \E\left[\frac{\abs{\W_\iMD} + \abs{\W_\eMD} +\abs{\W_\eMD^1}}{\Ka}\right] +\tp \nonumber\\
    &
    =\E\left[\frac{(\Ka-\oKap)^+ + (|\We|-|\hWe|)^+ +\abs{\W_\eMD^1}}{\Ka}\right] +\tp \label{eq:pMD_general_case}\,,\\
    \pFA &\le \E\left[\frac{|\hW_\iFA|+|\hW_\eFA|+ |\hW_\eFA^1|}{\hKa}\right] + \tp \nonumber\\
    &= \E\left[\frac{(\uKap -\Ka)^+ +(|\hWe| -|\We|)^++ |\W_\eMD^1|}{ \Ka -|\We| -(\Ka-\oKap)^+  +|\hWe|+ (\uKap-\Ka)^+}\right] + \tp \label{eq:pFA_general_case}\,,\\
    \pAUE &\le \E\left[\frac{|\W_\eAUE|}{\Ka}\right] +\tp
    =\E\left[\frac{\min\{|\We|, |\hWe|\}-|\W_\eMD^1|}{\Ka}\right] +\tp\,.\label{eq:pAUE_general_case}
\end{align}
Since $\uKap$ and $\oKap$ are functions of $\Kap$, the expectations above are w.r.t.\@ $(\Ka, \Kap, |\We|, |\hWe|)$. To evaluate the right side of \eqref{eq:pMD_general_case}--\eqref{eq:pAUE_general_case}, observe that it suffices to characterize the joint  distribution of $(\Ka, \Kap, |\We|, |\hWe|)$ and the distribution of $|\W_{\eMD}^1|$ conditioned on $(\Ka, \Kap, |\We|, |\hWe|)$. We provide these distributions in Lemmas \ref{lem:general_p_t,tHat}
 and \ref{lem:cond_expectation_general_case} after introducing the necessary notation. 

 \edit{Following the definitions  in \eqref{eq:T_thm}--\eqref{eq:ut_thm} of Theorem \ref{thm:ngo_ach}, let 
 $\T$ denote the set of possible values for $t=|\We|$, and let $\hTt$ denote the set of possible values for  $\hatt=|\hWe|$ when $|\We|=t$.}
 We repeat these definitions below for  convenience:
  \begin{align}
     &\T :=\left[0: \min\{\ka,\okap\}\right] \label{eq:def_T}, \\
     & \hTt := \left[\left\lbrace t+(\ka-\okap)^+ - (\ka-\ukap)^+ \right\rbrace^+: \tu\right] \label{eq:def_barTt}, \\
     & \text{where}\quad \tu  := \min\left\lbrace\okap-(\ukap-\ka)^+, \; t+(\okap-\ka)^+ -(\ukap-\ka)^+ \right\rbrace .\label{eq:def_tu}
 \end{align}
In particular, \edit{Eq.~\eqref{eq:def_T}} follows from  $|\W\setminus\hW|=|\W_\iMD\cup \We|\le |\W|$, i.e., $(\ka-\okap)^++t\le\ka$. We verify \eqref{eq:def_barTt}--\eqref{eq:def_tu} by   noticing that  (i) $|\hW\setminus \W|=|\hW_\iFA\cup \hWe|\le |\hW| =\hka\le \okap$ leads to  $(\ukap-\ka)^++\hatt\le  \okap$, and (ii) $\hka\in[\ukap: \okap]$ implies 
 \begin{align}
     &\ukap\le \hka= \ka-t-(\ka-\okap)^+  +\hatt+ (\ukap-\ka)^+\le \okap\nonumber\\
     \Rightarrow & 
     \begin{cases}
      \hatt\le t + (\okap-\ka)^+-(\ukap-\ka)^+\\
      \hatt \ge t +(\ka-\okap)^+ - (\ka-\ukap)^+\,.
    \end{cases}  \nonumber
 \end{align}
 
 We next present Lemmas \ref{lem:general_p_t,tHat} and \ref{lem:cond_expectation_general_case} using the definitions in  \eqref{eq:def_T}--\eqref{eq:def_tu}. We also recall 
 that 
$\Ka\in[\kl: \ku]$ under the new measure specified in Section \ref{sec:change_of_measure}, and that $\Kap\in[\kl: \ku]$  due to the constraints in the decoding rule \eqref{eq:determine_Ka'}.

 \begin{lem}\label{lem:general_p_t,tHat}
    In the general case,  for every $\ka\in[\kl: \ku]$, $\kap\in[\kl: \ku]$, $ t\in \T$ and  $\hatt\in \hTt$, we have
    \begin{align}
          \P\left(\Ka=\ka, \Kap=\kap, \abs{\We}=t,  |\hWe|=\hatt\right)
    \le p_\Ka\left(\ka\right)\min\left\lbrace p(t,\hatt), \xi(\ka, \kap)\right\rbrace,\label{eq:prob_4_items_general_case}
    \end{align}
    where $p(t,\hatt)$ and $\xi(\ka, \kap)$ are defined  in \eqref{eq:p_t,tHat_def} and  \eqref{eq:xi_thm} of Theorem \ref{thm:ngo_ach}.
\end{lem}
 \begin{lem}\label{lem:cond_expectation_general_case}In the general case, define $\tmin:=\min\{t,\hatt\}$, $\mc{R}: = L-\ka +\tmin - (\ukap-\ka)^+ -(\hatt-t)^+ \ge \tmin$,  and $\psiu:=\min\{\tmin, \mc{R}-\tmin\}$ as in \eqref{eq:tmin_R_thm} and \eqref{eq:psiu_thm} of Theorem \ref{thm:ngo_ach}. For every $\ka\in[\kl: \ku]$, $\kap\in[\kl: \ku]$, $ t\in \T$, $\hatt\in \hTt$, and $\psi\in[0: \psiu]$, let 
\begin{align}\label{eq:u(t,psi),u(t)_general_case}
  \nu(\tmin,\psi)= \binom{\mc{R}-\tmin}{\psi}M^\psi \cdot \binom{\tmin}{\tmin-\psi}(M-1)^{\tmin-\psi},\quad 
    \nu(\tmin) = \sum_{\psi=0}^{\psiu} \nu(\tmin,\psi).
\end{align}
Then the following hold:
\begin{enumerate}[(1)]
\item\label{lem:nu_tmin} \edit{Given $\Ka=\ka$, $\Kap=\kap$,    $\We$  (with $|\We|=t$),  $|\hWe|=\hatt$,   $ \W_\eMD$ and   $\hW_\eFA$, the number of distinct choices for the set  $\hWe$ equals $\nu(\tmin)$. Furthermore, $\nu(\tmin)\le \binom{\mc{R}}{\tmin}M^\tmin$.} 
    \item \label{lem:cond_psi_tmin_general} It holds that 
\begin{equation}\label{eq:prob_|W21|_general_case}
    \P\left(|\W_{\eMD}^1|=\psi\,\bigg|\, \Ka=\ka, \Kap=\kap, |\We| = t, |\hWe|=\hatt\right)
    =  \nu(\tmin,\psi)/\nu(\tmin).
    \end{equation}
\end{enumerate}
\end{lem}
\edit{For the given Gaussian joint-codebook  (described in Section \ref{sec:preliminaries})}, substituting \eqref{eq:prob_4_items_general_case}--\eqref{eq:u(t,psi),u(t)_general_case}  into \eqref{eq:pMD_general_case}--\eqref{eq:pAUE_general_case}  
 \edit{establishes  the individual guarantees in} \eqref{eq:thm_epsMD}--\eqref{eq:thm_epsAUE} of Theorem \ref{thm:ngo_ach}. \edit{To show that \eqref{eq:thm_epsMD}--\eqref{eq:thm_epsAUE} hold simultaneously, we proceed similarly to \cite[Theorem 19]{polyanskiy2011feedback} and \cite[Theorem 8]{Yavas2021random} by constructing a  randomized coding scheme that uses time-sharing among at most two joint-codebooks from the Gaussian ensemble. This completes the proof of  Theorem \ref{thm:ngo_ach} in the general case.} 
 In the  remainder of this section, we  provide the proofs of Lemmas \ref{lem:general_p_t,tHat} and \ref{lem:cond_expectation_general_case}, which are analogous to the proofs of Lemmas \ref{lem:p_W_aMD_union_W_AUE} and \ref{lem:cond_p_aMD} for the special case.

\begin{proof}[Proof of Lemma \ref{lem:general_p_t,tHat}]
To prove \eqref{eq:prob_4_items_general_case}, by the law of total probability
 we need to show 
\begin{align}\label{eq:cond_prob_to_prove}
    \P\left(\Kap=\kap, \abs{\We}=t,  |\hWe|=\hatt \,\bigg|\,  \Ka=\ka\right)\le \min\{p(t,\hatt), \xi(\ka, \kap)\}.
\end{align}
As before,  let $\ukap$ and $\okap$ denote the realizations of $\uKap$ and $\oKap$ when  $\Kap=\kap$,  given $\kl, \ku$, \edit{$\rl$ and $\ru$}. 
From  the decoding rule in 
 \eqref{eq:determine_Ka'}--\eqref{eq:min_distance_decoder},   we have  with $\by\sim \normal_n(\bzero, (\ka P'+1)\bI) $, 
     \begin{align}
    &\P\left( \Kap=\kap, \abs{\We}=t,  \, |\hWe|=\hatt \,\bigg|\,  \Ka=\ka\right) \nonumber\\
    & \le \P\left( \argmax_{\kappa\in[\kl:\ku]} \, p(\by| \Ka=\kappa)=\kap,  
     \abs{\We}=t,  |\hWe|=\hatt, |\hW|\in [\ukap:\okap]\,\bigg|\, \Ka=\ka 
    \right)\nonumber\\
    &\le \min\Bigg\lbrace
    \underbrace{\P\left( \argmax_{\kappa\in[\kl:\ku]} \, p(\by|\Ka=\kappa)=\kap\right)}_{\text{term (i)}},\;
    \underbrace{\P\left(\abs{\We}=t,  |\hWe|=\hatt, |\hW|\in [\ukap:\okap]
    \,\bigg|\, \Ka=\ka \right)}_{\text{term (ii)}}\Bigg\rbrace\,,\label{eq:two_upper_bounds_pt,tHat,xi}
    \end{align} 
    \edit{
    where the first step holds because $\Kap=\kap$  implies $|\hW|\in [\ukap:\okap]$, and 
    the second step holds because $\P(A\cap B)\le \min\{\P(A), \P(B)\}$.}
We bound the terms (i) and (ii) on the right side of \eqref{eq:two_upper_bounds_pt,tHat,xi}. For the first term, using $\by\sim \normal_n(\bzero, (\ka P'+1)\bI) $, we have
    \begin{align}
    \text{term (i)} &= 
    \prob\big(p(\by|\Ka=\kap) > p(\by|\Ka=\kappa), \, \forall \kappa\in[\kl: \ku]\setminus \kap\big) \nonumber\\
    &\le \min_{\kappa\in[\kl:\ku]\setminus \kap} \P\left(p(\by|\Ka=\kap)> p(\by|\Ka=\kappa)\right)\nonumber\\
  & =\min_{\kappa\in[\kl:\ku]\setminus \kap} \P\Bigg\lbrace-\frac{n}{2}\ln(1+\kap P')-\frac{\|\by\|^2_2}{2(1+\kap P')}> -\frac{n}{2}\ln(1+\kappa P')-\frac{\|\by\|^2_2}{2(1+\kappa P')}\Bigg\rbrace\nonumber\\
  & =\min_{\kappa\in[\kl:\ku]\setminus \kap}  \P\left\lbrace\left(\frac{1}{1+\kap P'} - \frac{1}{1+\kappa P'}\right)\|\by\|_2^2< n\ln \left(\frac{1+ \kappa P'}{1+ \kap P'}\right)\right\rbrace\nonumber\\
  &\stackrel{\text{(a)}}{=} \min_{\kappa\in[\kl:\ku]\setminus \kap}\Bigg\lbrace\ind\{\kappa < \kap\}\frac{\Gamma(\frac{n}{2}, \zeta)}{\Gamma(\frac{n}{2})} + \ind\{\kappa > \kap\}\frac{\gamma(\frac{n}{2}, \zeta)}{\Gamma(\frac{n}{2})}\Bigg\rbrace = \xi(\ka, \kap)\,, \label{eq:xi_proof}
    \end{align}
where (a) uses 
$\frac{1}{1+\ka P'}\|\by\|_2^2\sim \text{Gamma}(\frac{n}{2}, 2)$ and
\begin{align}
    \zeta =\frac{n}{2(1+\ka P')}\ln \left(\frac{1+ \kappa P'}{1+\kap P'}\right) \left[\frac{1}{1+\kap P'} - \frac{1}{1+\kappa P'}\right]^{-1}.\label{eq:zeta_proof}
\end{align}
Results \eqref{eq:xi_proof}--\eqref{eq:zeta_proof} correspond to \eqref{eq:xi_thm}--\eqref{eq:zeta_thm} in Theorem \ref{thm:ngo_ach}.

Next, we bound  term (ii) in \eqref{eq:two_upper_bounds_pt,tHat,xi}.  Similar to the special case,  we have 
\begin{align}
\by=c(\W)+\bveps = 
c(\W_\iMD)+c(\We) +c(\W\cap \hW)+\bveps. \label{eq:channel_output_general}    
\end{align}
The decoder returns $ \hW_\iFA\cup\hWe\cup (\W\cap \hW)$ as its estimate, which implies that
$$\|\by-c(\hW_\iFA)-c(\hWe)-c(\W\cap \hW)\|_2^2
<\|\by-c(\hW_\iFA)-c(\We)-c(\W\cap \hW)\|_2^2. $$ Combined with 
\eqref{eq:channel_output_general}, this is equivalent to 
\begin{align}
\|c(\W_\iMD)+c(\We)-c(\hW_\iFA) -c(\hWe) +\bveps\|_2^2
<\|c(\W_\iMD)-c(\hW_\iFA)+\bveps\|_2^2.\label{eq:general_case_error_criterion}
\end{align}
Let $\mc{E}(\W_\iMD, \hW_\iFA, \We, \hWe)$ be the event that $(\W_\iMD, \hW_\iFA, \We, \hWe)$ satisfies \eqref{eq:general_case_error_criterion} \edit{with $\Ka=\ka$, $|\We|=t$, $|\hWe|=\hatt$ and $|\hW|\in [\ukap:\okap]$}. Then we can rewrite  term (ii) in \eqref{eq:two_upper_bounds_pt,tHat,xi} as
\begin{align}\label{eq:quantity_to_be_bounded_by_ptt}
    \text{term (ii)} = \P\left\lbrace\bigcup_{\We: |\We|=t}\bigcup_{\hWe: |\hWe|=\hatt}
    \mc{E}(\W_\iMD, \hW_\iFA, \We, \hWe)\right\rbrace.
    \end{align}
We follow the same steps as in \eqref{eq:p_Wa=p_union_E}--\eqref{eq:p_Wa_special_case_proof} to bound \eqref{eq:quantity_to_be_bounded_by_ptt}. First note that
\begin{align}\label{eq:base_event_general}
    &\P\left(\mc{E}(\W_\iMD, \hW_\iFA, \We, \hWe)\,\bigg|\, \W_\iMD, \hW_\iFA, \We, \W_\eMD, \hW_\eFA, \bveps\right)\nonumber\\
     & \stackrel{\text{(a)}}{\le }\exp\left(\lambda_0 \|c(\W_\iMD)-c(\hW_\iFA)+\bveps\|_2^2\right) \cdot \nonumber\\
     & \quad \quad 
     \E_{\hWe}\left[\exp\left(-\lambda_0 \|c(\W_\iMD)+c(\We)-c(\hW_\iFA)-c(\hWe)+\bveps\|_2^2\right)\,\bigg|\, \W_\iMD, \hW_\iFA, \We, \bveps\right]\nonumber\\
    &\stackrel{\text{(b)}}{=}\exp\left(\lambda_0\|c(\W_\iMD)-c(\hW_\iFA)+\bveps\|_2^2\right) \cdot \nonumber\\
    & \quad \quad
    (1+2 P'\hatt\lambda_0)^{-\frac{n}{2}}\exp\left(-\frac{\lambda_0\|c(\W_\iMD)+c(\We)-c(\hW_\iFA)+\bveps\|_2^2}{1+2 P'\hatt\lambda_0}\right), \; \forall \lambda_0>-\frac{1}{2 P'\hatt}\,,
\end{align}
where (a) applies the Chernoff bound in Lemma \ref{lem:chernoff}, and (b) applies  the Gaussian identity in Lemma \ref{lem:gaussian_identity}. 

Define the  shorthand $\chi_0=\rho\lambda_0/(1+2 P'\hatt\lambda_0), b_0=\rho\lambda_0-\chi_0/(1+2P't\chi_0)$ and $a_0=\rho \ln \left(1+2P'\hatt \lambda_0\right) + \ln \left(1+2P't\chi_0\right)$. Analogous to \eqref{eq:apply_rho_trick_1}--\eqref{eq:p_double_unions_E_expanded_special} in the special case, next we add  two unions over $\We$ and $\hWe$ to the event in \eqref{eq:base_event_general}  and use Gallager's $\rho$-trick \eqref{eq:rho_trick_cond} twice to obtain  an upper bound on \eqref{eq:quantity_to_be_bounded_by_ptt}.

\paragraph{$\rho$-trick for $\bigcup_{\hWe}$}   We apply the $\rho$-trick \eqref{eq:rho_trick_cond} to add a union over $\hWe$ to the left side of \eqref{eq:base_event_general}. \edit{As we prove in Lemma \ref{lem:cond_expectation_general_case} part \ref{lem:nu_tmin},
the number of ways to construct $\hWe$ with size $|\hWe|=\hatt$ is upper bounded by  $\binom{\mc{R}}{\tmin}M^\tmin$, where $\tmin=\min\{t,\hatt\}$ and 
 $\mc{R} = L-\ka +\tmin - (\ukap-\ka)^+ -(\hatt-t)^+ \ge \tmin$.} 
Hence applying the $\rho$-trick \edit{and taking expectation over $\We, \W_\eMD$ and $\hW_\eFA$} yields for $\rho\in[0,1]$,
\begin{align}
    &\P\left(\bigcup_{\hWe: |\hWe|=\hatt}\mc{E}(\W_\iMD, \hW_\iFA, \We, \hWe)\,\bigg|\, \W_\iMD, \hW_\iFA,  \bveps\right)\nonumber\\
    &\le \left[\binom{\mc{R}}{\tmin}M^\tmin\right]^\rho 
        (1+2 P'\hatt\lambda_0)^{-\frac{n\rho}{2}}
\exp\left(\rho\lambda_0\|c(\W_\iMD)-c(\hW_\iFA)+\bveps\|_2^2\right)  \cdot \nonumber\\
& \quad\quad \E_{\We}\left[\exp\left(-\chi_0 \|c(\W_\iMD)+c(\We)-c(\hW_\iFA)+\bveps\|_2^2\right)\,\bigg|\, \W_\iMD,\hW_{\iFA}, \bveps\right]\nonumber\\
    &\stackrel{\text{(a)}}{=}  \binom{\mc{R}}{\tmin}^\rho M^{\rho\tmin}
        (1+2 P'\hatt\lambda_0)^{-\frac{n\rho}{2}}
\exp\left(\rho \lambda_0\|c(\W_\iMD)-c(\hW_\iFA)+\bveps\|_2^2\right)  \cdot \nonumber\\
& \quad \quad\left\lbrace (1+2P' t \chi_0)^{-\frac{n}{2}}\exp\left(-\frac{\chi_0\|c(\W_\iMD)-c(\hW_\iFA)+\bveps\|_2^2}{1+2P't\chi_0}\right) \right\rbrace\nonumber\\
& = \binom{\mc{R}}{\tmin}^\rho M^{\rho\tmin}
\exp\left(-\frac{na_0}{2} +b_0  \|c(\W_\iMD)-c(\hW_\iFA)+\bveps\|_2^2\right),\label{eq:rho_trick1_general}
\end{align}
where (a) applied the Gaussian identity from Lemma \ref{lem:gaussian_identity}. 
\paragraph{$\rho$-trick for $\bigcup_{\We}$} \edit{Given $\W_\iMD$,} we recall that $\We \subset \W\setminus \W_{\iMD}$ and $|\W\setminus \W_{\iMD}| = \ka -  (\ka-\okap)^+ = \min\{\ka, \okap\}$, hence there are $\binom{\min\{\ka, \okap\}}{t}$ ways to construct $\We$ with size $|\We|=t$. 
 Applying the $\rho$-trick  \eqref{eq:rho_trick_cond} again to add a union over $\We$ to the left side of  \eqref{eq:rho_trick1_general}  
yields that for $\rho_1\in[0,1]$,
\begin{align}
    \text{term (ii)}& \stackrel{(a)}{=} \P\left(\bigcup_{\We:|\We|=t}\bigcup_{\hWe: |\hWe|=\hatt}\mc{E}(\W_\iMD, \hW_\iFA, \We, \hWe)\right)\nonumber\\
    &\le \min_{\substack{\rho, \rho_1\in[0,1],\\\lambda_0>-\frac{1}{2 P'\hatt}}}\binom{\min\{\ka, \okap\}}{t}^{\rho_1}  \cdot \nonumber\\
& \quad \;\;\E_{\W_\iMD, \hW_\iFA, \bveps}\left[\binom{\mc{R}}{\tmin}^{\rho\rho_1} M^{\rho\rho_1\tmin}
\exp\left(-\frac{n\rho_1 a_0}{2} +\rho_1 b_0 \|c(\W_\iMD)-c(\hW_\iFA)+\bveps\|_2^2\right)\right]\nonumber\\
    & \stackrel{\text{(b)}}{=} \min_{\substack{\rho, \rho_1\in[0,1],\\\lambda_0>-\frac{1}{2 P'\hatt}}} \binom{\min\{\ka,\okap\}}{t}^{\rho_1}\binom{\mc{R}}{\tmin}^{\rho\rho_1} M^{\rho\rho_1\tmin} 
    \exp\left(-\frac{n\rho_1 a_0}{2} \right)(1-2\rho_1 P_1 b_0)^{-\frac{n}{2}}\nonumber\\
    &\stackrel{\text{(c)}}{=} \min_{\substack{\rho, \rho_1\in[0,1],\\\lambda>-\frac{1}{P'\hatt}}} \binom{\min\{\ka,\okap\}}{t}^{\rho_1}
    \binom{\mc{R}}{\tmin}^{\rho\rho_1} M^{\rho\rho_1\tmin}
    \exp\left(-\frac{n\rho_1 a}{2} \right)(1-\rho_1P_1b)^{-\frac{n}{2}}=p(t, \hatt)\ , \label{eq:p_double_union_E_expanded_general}
\end{align}
where (a) applied \eqref{eq:quantity_to_be_bounded_by_ptt}; (b) used the shorthand $P_1=(\abs{\W_\iMD}+|\hW_\iFA|)P'+1=((\ka-\okap)^++(\ukap-\ka)^+)P'+1$  defined in \eqref{eq:P1_thm}; (c) used $\lambda:=2\lambda_0$ and $a,b, \chi$ as defined in \eqref{eq:a_thm}--\eqref{eq:b_mu_thm} (thus $ a_0=a, b_0=b/2, \chi_0=\chi/2$) and $p(t,\hatt)$ as defined in \eqref{eq:p_t,tHat_def}.
Substituting \eqref{eq:xi_proof} and \eqref{eq:p_double_union_E_expanded_general} back into  \eqref{eq:two_upper_bounds_pt,tHat,xi}  yields \eqref{eq:cond_prob_to_prove}, which completes the proof. 
\end{proof}

\begin{proof}[Proof of Lemma \ref{lem:cond_expectation_general_case}]\textbf{\edit{Proof of part \ref{lem:nu_tmin}:}} \edit{Given $\Ka=\ka$, $\Kap=\kap$, $\We$ (with $|\We|=t$), $|\hWe|=\hatt$,  $\W_\eMD$ and $\hW_\eFA$, the  elements of  $\hWe$ are fully determined by the remaining subsets $\hW_{\eFA}^1$ and $\hW_\eAUE$. Therefore, counting the number of distinct configurations of $\hWe$ reduces to counting the   distinct  ways to construct $\hW_{\eFA}^1 \cup \hW_\eAUE$.}

To begin, 
note that  $\hW_\eFA^1\cup \hW_\eAUE$ \emph{does not} contain codewords from the codebooks of the users:
\begin{itemize}
    \item who  transmitted $\W\cap \hW$ since these are  correctly decoded, 
    \item who  transmitted $\W_\iMD$ or $\W_{\eMD}$ because these users are declared silent, or
    \item who are  declared active via the $\FA$s in $\hW_\iFA$ or $\hW_{\eFA}$.
\end{itemize} 
As a result, 
 the elements of  $\hW_{\eFA}^1\cup \hW_{\eAUE}$ exclusively  come from the remaining $\mc{R}$ codebooks, with
 \begin{align}
     \mc{R}&=L-|\W\cap \hW| - |\W_\iMD\cup \W_{\eMD}| - |\hW_\iFA\cup \hW_{\eFA}|\nonumber\\
     &=L-\big(\ka-|\W_\eMD^1\cup \W_\eAUE| \big)- |\hW_\iFA\cup \hW_{\eFA}|\nonumber\\
     &=L-(\ka-\tmin) - \big((\ukap-\ka)^+ +(\hatt-t)^+\big)\nonumber\\
     &= L-\ka+\tmin - (\ukap-\ka)^+ -(\hatt-t)^+\ .\label{eq:def_r_general_case_simplified}
  \end{align}
These $\mc{R}$ codebooks include
\begin{enumerate}[i)]
\item \label{item:general_active} the codebooks  of the $|\W_{\eMD}^1\cup\W_{\eAUE}|=\tmin$  users who transmitted   $\W_{\eMD}^1\cup\W_{\eAUE}$, or
\item \label{item:general_silent} the codebooks of the $(\mc{R}-\tmin)$ silent users, who are \emph{not} mistakenly declared active through $\hW_\iFA$ and  $\hW_{\eFA}$. 
\end{enumerate}
As a sanity check on \ref{item:general_silent},  there are  $(L-\ka) $ silent users in total, among which
$|\hW_\iFA|+|\hW_{\eFA}|=(\ukap-\ka)^+ +(\hatt-t)^+ $ are mistaken as active, so the number of users in the category \ref{item:general_silent} is 
$(L-\ka) - \big( (\ukap-\ka)^+ +(\hatt-t)^+ \big)$, which is equal to $(\mc{R}-\tmin)$ due to \eqref{eq:def_r_general_case_simplified}.

Observe that each element of $\hW_{\eFA}^1$ comes from a type-\ref{item:general_silent} codebook, and each element of $\hW_{\eAUE}$ comes from the set of untransmitted codewords in a type-\ref{item:general_active} codebook. The number of ways to divide the set $\hW_{\eFA}^1\cup\hW_{\eAUE}$ into $\hW_{\eFA}^1$ and $\hW_{\eAUE}$ such that $|\hW_{\eFA}^1|=\psi$ and $|\hW_{\eAUE}|=\tmin-\psi$ is
\begin{align}
 \nu(\tmin,\psi)= \binom{\mc{R}-\tmin}{\psi}M^\psi \cdot \binom{\tmin}{\tmin-\psi}(M-1)^{\tmin-\psi}. \nonumber
\end{align} Thus the total number of ways to  divide  $\hW_{\eFA}^1\cup\hW_{\eAUE}$ into $\hW_{\eFA}^1$ and $ \hW_{\eAUE}$ is
\begin{align}
    \nu(\tmin )= \sum_{\psi=0}^{\min\{\tmin, \mc{R}-\tmin\}} \nu(\tmin, \psi) \le \binom{\mc{R}}{\tmin} M^\tmin\, . \nonumber
\end{align} 

\paragraph{\edit{Proof of part \ref{lem:cond_psi_tmin_general}:}}
Recall that $ |\W_{\eMD}^1|=|\hW_{\eFA}^1|$ as illustrated in Fig.~\ref{fig:general_case_venn}.  We characterize the conditional distribution of  $ |\W_{\eMD}^1|$ 
by characterizing the conditional distribution of $|\hW_{\eFA}^1|$. 
Combining $\nu(\tmin, \psi)$ and $\nu(\tmin)$ and using the exchangeability of the codewords yields 
\begin{align}
\P\left(|\W_{\eMD}^1|=\psi\,\bigg|\, \Ka=\ka, \Kap=\kap, |\We| = t, |\hWe|=\hatt,  \We, \W_\eMD, \hW_\eFA\right)
    =  \nu(\tmin,\psi)/\nu(\tmin).
\end{align}
\edit{Taking expectation over $\We, \W_\eMD$ and $\hW_\eFA$ yields 
\eqref{eq:prob_|W21|_general_case}.}
Note that \eqref{eq:u(t,psi),u(t)_general_case}--\eqref{eq:prob_|W21|_general_case} imply that 
$\P\big(\abs{\W_{\eMD}^1}=\psi \,\big| \, \Ka=\ka, \Kap=\kap, \abs{\We}=t,|\hWe|=\hatt \big)=1$ when  $\psi=0$ and $\tmin=\min\{t, \hatt\}=0$,  as one would expect.
\end{proof}

\subsection{Proof of Corollary \ref{cor:error_floors}}\label{sec:proof_error_floors}
When $E_b/N_0\to \infty$, the maximum-likelihood decoder doesn't make mistakes as long as  $\ka\in [\ukap:\okap]$. Under the new measure defined in Section \ref{sec:change_of_measure}, $\ka\in [\ukap:\okap]$ can be ensured by choosing \edit{$\rl=\ru=L$}.
However, given  fixed \edit{$\rl$ and $\ru$} that do not scale with $L$,  it is possible that  $\ka\notin  [\ukap: \okap] $. When this happens, we have one of the following two scenarios:
\begin{itemize}
     \item if $\ukap> \ka$, then $\hka\ge \ukap> \ka$, hence the decoder  commits at least $ (\ukap-\ka)$ $\FA$s in $\hW_{\iFA}$, and guarantees that  all   transmitted codewords are decoded  correctly,  i.e., $ \W\subset\hW$\ ;
    \item if $\okap< \ka$, then $ \hka\le\okap< \ka$, hence the decoder commits at least $ (\ka-\okap)$ $\MD$s in $\W_{\iMD}$, and  guarantees that the decoded set only contains  transmitted codewords, i.e., $\hW\subset \W$.
\end{itemize}
 Consequently, as ${E_b}/{N_0}\to\infty$ (or $P\to\infty$), the error floors are calculated assuming  no additional errors are incurred on top of the initial errors $\W_{\iMD}$ or $\hW_{\iFA}$ (i.e., $t=\hatt=0$).  Taking $t=\hatt=0$ in \eqref{eq:thm_epsMD}--\eqref{eq:thm_epsAUE} and noticing that $\lim_{P'\to\infty}\lim_{P\to \infty}\P(\|\bc_{w_\ell}^{(\ell)}\|_2^2> nP) = 0$ in \eqref{eq:tilde_p}, we obtain \eqref{eq:error_floor_MD}--\eqref{eq:error_floor_AUE} in Corollary \ref{cor:error_floors}.  \qed
    \section{Proof of Theorem \ref{thm:gen_sc_gamp_SE}}\label{app:gen_sc_gamp_proof}

We first describe in Section \ref{subsec:gen_rect_AMP} an abstract AMP iteration  with matrix-valued iterates. The abstract AMP iteration, defined for an i.i.d. Gaussian matrix, does not correspond to a specific signal/observation model.  The state evolution result for this  abstract AMP   was established  in \cite{javanmard2013state}; see also \cite[Sec.\@ 6.7]{feng2022unifying} and  \cite{Ger21}. We use this result to prove Theorem \ref{thm:gen_sc_gamp_SE}, by showing that the spatially coupled AMP in \eqref{eq:alg_sc_Zt}--\eqref{eq:alg_sc_Xt} for the linear model in \eqref{eq:model_B=1} is an instance of the abstract AMP iteration with matrix-valued iterates.

\subsection{Abstract  AMP recursion for i.i.d.~Gaussian matrices} \label{subsec:gen_rect_AMP}
The abstract AMP and the state evolution result in this subsection is similar to the one in \cite[Section V.A]{cobo2023bayes}, but we present it here for completeness.

Let $\bM \in \reals^{\tn \times L}$ be a random matrix with entries $M_{i\ell} \stackrel{\text{i.i.d.}}{\sim} \normal(0, {1}/{\tn})$. Consider the following recursion with iterates $\be^{t+1} \in \reals^{\tn \times \ell_E}$ and $\bh^{t+1} \in \reals^{L \times \ell_H}$, defined  for $t \ge 0$:
\begin{align}
    \bh^{t+1} & = \bM^{\top} g_t(\be^{t}, \, \bgamma)
    \, - \, f_t(\bh^t, \, \bbeta) \sD^{\top}_t \, ,\label{eq:Ht_update} \\
    \be^{t+1} &  = \bM f_{t+1}(\bh^{t+1}, \, \bbeta) \, - \, g_{t}(\be^{t}, \, \bgamma) \sB_{t+1}^{\top},\label{eq:Et_update} 
 \end{align}
 where $\bgamma\in\reals^{\tn\times \k},\,  \bbeta\in\reals^{L\times \k}$ are matrices independent of $\bM$.
 The recursion is initialized with some $\bi^0 \equiv  f_0(\bh^0, \, \bbeta) \in \reals^{L \times \ell_E}$ and $\be^0 = \bM \bi^0$.
Here, the functions $f_{t+1}: \reals^{L \times (\ell_H +\k)} \to \reals^{\tn \times \ell_E}$ and $g_t: \reals^{\tn \times (\ell_E +\k)} \to \reals^{L \times \ell_H}$ are defined component-wise as follows, for $t \ge 0$. 
With  the sets $[\tn]$ and $[L]$ partitioned into $\{\mc{I}_\sfr \}_{\sfr\in [\R]}$ and $\{\mc{L}_{\sfc}\}_{\sfc\in [\C]}$ as defined below \eqref{eq:Zt_result}, 
 we assume that  there exist functions $\tf_{t+1} : \reals^{\ell_H} \times \reals^{\k} \times [\Lc] \to \reals^{\ell_E}$ and 
$\tg_t: \reals^{\ell_E} \times \reals^{\k} \times [\Lr] \to \reals^{\ell_H}$ such that 
\begin{equation}
    \begin{split}
            & f_{t+1,\ell}(\bh, \bbeta) = \tf_{t+1}(\bh_\ell, \, \bbeta_\ell, \ \sfc ),  \quad \text{ for } \ell \in \mc{L}_{\sfc}, \quad \bh \in \reals^{L \times \ell_H}, \ \bbeta \in \reals^{L\times \k}, \\
            & g_{t,i}(\be, \bgamma) = \tg_{t}(\be_i, \, \bgamma_i, \ \sfr ), \quad \text{ for } i \in \mc{I}_{\sfr}, \quad \be \in \reals^{\tn \times \ell_E}, \ \bgamma \in \reals^{\tn\times \k}.
    \end{split}
    \label{eq:ftj_gti}
\end{equation}
(Here we note that $\bh_\ell, \bbeta_\ell, \be_i, \bgamma_i$ are all row vectors.) In words, the functions $f_{t+1}, g_t$ act row-wise on their inputs and depend only on the indices of the column block/row block that their inputs come from. 
The matrices $\sD_t \in \reals^{\ell_H \times \ell_E}$ and $\sB_{t+1} \in \reals^{\ell_E \times \ell_H}$ in  \eqref{eq:Ht_update}--\eqref{eq:Et_update}  are defined as:
\begin{equation}
\begin{split}
    \sD_t = \frac{1}{\tn} \sum_{i=1}^{\tn} \tg_t'(\be^t_i, \bgamma_i, \sfr),\quad 
    \sB_{t+1} = \frac{1}{\tn} \sum_{\ell=1}^{L} \tf_{t+1}'(\bh^{t+1}_\ell, \bbeta_\ell, \sfc),  \label{eq:abstract_Dt_Bt}
\end{split}
\end{equation}
where  $\tg_t'$ and $\tf_{t+1}'$ denote the Jacobians of $\tg_t$ and $\tf_{t+1}$ with respect to their first argument.

\paragraph{Assumptions}
\begin{enumerate}[label={\bfseries (B\arabic*)}, start=0]
\item As the dimensions $n, L \to \infty$, the aspect ratio $\frac{L}{n} \to \mu >0$, and we emphasize that $\ell_E, \ell_H$ as well as $\Lr, \Lc$ are positive integers that do not scale with  $\n, L$. 

\item The functions $\tf_t( \cdot \, , \cdot \, , \sfc)$ and 
$\tg_t(\cdot \, , \cdot \, , \sfr)$ are Lipschitz for $t \ge 1$, $\sfr \in [\R]$ and $\sfc \in [\C]$.

\item  For  $\sfc \in [\Lc]$ and $\sfr \in [\Lr]$, let $\nu(\bbeta_\sfc)$ and $\pi(\bgamma_\sfr)$ denote the  empirical distributions of $\bbeta_\sfc  \equiv  (\bbeta_\ell)_{\ell \in [\mc{L}_{\sfc}]}  \in \reals^{(L/\C)\times \k}$ and 
$\bgamma_\sfr \equiv (\bgamma_i)_{i \in [\mc{I}_{\sfr}]} \in \reals^{(\tn/\R)\times \k}$, respectively. Then, for some $m \in [2, \infty)$, there exist $\k$-dimensional vector random variables $\bar{\bbeta}_{\sfc} \sim \nu_{\sfc}$   and $\bar{\bgamma}_{\sfr} \sim \pi_{\sfr}$, 
with $\int_{\reals^k}\|\bx\|^m d\nu_{\sfc}(\bx)$, $\int_{\reals^k}\|\bx\|^m d\pi_{\sfr}(\bx)<\infty$
such that $d_m(\nu(\bbeta_\sfc),\nu_{\sfc}) \rightarrow 0$ and $d_m(\pi(\bgamma_\sfr), \pi_{\sfr}) \rightarrow 0$ almost surely.  Here $d_m(P,Q)$ is the $m$-Wasserstein distance between distributions $P, Q$ defined on the same Euclidean probability space.

\item  For $\sfc \in [\C]$, we assume there exists a symmetric non-negative definite $\hat{\bSigma}^{0, \sfc} \in \reals^{\ell_E \times \ell_E}$ such that we almost surely have 
\begin{align}
\k\mu \cdot\lim_{L \to \infty}   \frac{1}{(L / \C)}(\bi^0_{\sfc})^{\top} \bi^0_{\sfc} = \hat{\bSigma}^{0, \sfc},
\label{eq:hSig0}
\end{align}
where $\bi^0_\sfc  \equiv  (\bi^0_\ell)_{\ell\in [\mc{L}_{\sfc}]}  \in \reals^{(L/\C)\times \ell_E}$.
\end{enumerate}

 Theorem \ref{thm:SE_gen} below states that for each $t \ge 1$, the empirical distribution of the rows of $\be^t\in \reals^{\tn \times \ell_E}$ converge to the law of a Gaussian random vector $\sim \normal_{\ell_E}(\bzero, \bSigma^t)$. Similarly, the empirical distribution of the rows of $\bh^t\in \reals^{L \times \ell_H}$ converge to the law of a Gaussian random vector $\sim \normal_{\ell_H}(\bzero, \bOmega^t)$. Here the covariance matrices $\bSigma^t  $ and $\bOmega^t $ are   defined by the state evolution recursion, described below.

\paragraph{State evolution} The state evolution is initialized with $\bSigma^{0} = \frac{1}{\C} \sum_{\sfc \in [\Lc]}
\hat{\bSigma}^{0, \sfc}$, where $\hat{\bSigma}^{0, \sfc}$ is defined in \eqref{eq:hSig0}. Then for $t\ge 1$,  we recursively compute  $\bOmega^{t+1} \in \reals^{\ell_H \times \ell_H}$ and $\bSigma^{t+1} \in \reals^{\ell_E \times \ell_E}$ given $\bOmega^t,\bSigma^t$ as follows:
\begin{align}
    \bOmega^{t+1} = \frac{1}{\R} \sum_{\sfr \in [\Lr]}  \hat{\bOmega}^{t+1, \sfr} \, ,
    \label{eq:Omega_t1}
\end{align}
where for $\sfr \in 
[\Lr]$,
\begin{align}
    &\hat{\bOmega}^{t+1, \sfr} = \E\left\{ \tg_t(\tilde{\bU}^t, \bar{\bgamma}_\sfr,  \, \sfr)  \, \tg_t(\tilde{\bU}^t, \bar{\bgamma}_\sfr,  \, \sfr)^{\top} \right\}, \quad
    \tilde{\bU}^t \sim \normal(\bzero, \bSigma^t) \,  \perp \, \bar{\bgamma}_\sfr \sim \pi_{\sfr}.
    \label{eq:Gt_r}
\end{align}
Similarly, 
\begin{align}
    \bSigma^{t+1} = \frac{1}{\C} \sum_{\sfc \in [\Lc]}   \hat{\bSigma}^{t+1, \sfc} \, ,
    \label{eq:Sigma_t1}
\end{align}
where for $\sfc \in  [\Lc]$,
\begin{align}
    &\hat{\bSigma}^{t+1, \sfc} = \k\mu\cdot \E\left\{ \tf_{t+1}({\bU}^{t+1}, \bar{\bbeta}_\sfc,  \, \sfc)  \, \tf_{t+1}({\bU}^{t+1}, \bar{\bbeta}_\sfc,  \, \sfc)^{\top} \right\}, \quad 
    {\bU}^{t+1} \sim \normal(\bzero, \bOmega^{t+1}) \,  \perp \, \bar{\bbeta}_\sfc \sim \nu_{\sfc}.
    \label{eq:Gt1_c}
\end{align}
The random variables $\tilde{\bU}^t, \, \bar{\bgamma}_\sfr, \, {\bU}^{t+1}, \, \bar{\bbeta}_\sfc$ are treated as column vectors in \eqref{eq:Gt_r} and \eqref{eq:Gt1_c}, and $\k\mu = L/\tn$.

\begin{thm}
Consider the abstract AMP recursion in \eqref{eq:Ht_update}--\eqref{eq:Et_update} under the Assumptions \textbf{(B0)}--\textbf{(B3)} above. Let $\xi: \reals^{\ell_H} \times \reals^{\k} \to \reals$ and $\zeta: \reals^{\ell_E} \times \reals^{\k} \to \reals$ be any pseudo-Lipschitz functions of order $m$, where $m$ is specified in Assumption \textbf{(B2)}. Then, for  $t \ge 1$,  $\sfr \in [\Lr]$ and $\sfc \in [\Lc]$, we almost surely have:
\begin{align}
 & \lim_{L \to \infty} \,  \frac{1}{L/\C} \sum_{\ell \in  \mc{L}_\sfc} \xi(\bh^t_\ell\, , \bbeta_\ell)  = \E\{ \xi( {\bU}^t \, , \bar{\bbeta}_\sfc) \}, 
 \quad   {\bU}^t \sim \normal(\bzero, \bOmega^t) \,  \perp \, \bar{\bbeta}_\sfc \sim \nu_{\sfc} \, ,     \label{eq:Ht_result} 
\\
 &  \lim_{\tn \to \infty} \,  \frac{1}{\tn/\R} \sum_{i \in \mc{I}_\sfr } \zeta(\be^t_i \, , \bgamma_i)  = \E\{ \zeta( \tilde{\bU}^{t} \, , \bar{\bgamma}_\sfr) \}, 
\quad  \tilde{\bU}^t \sim \normal(\bzero, \bSigma^t) \,  \perp \, \bar{\bgamma}_\sfr \sim \pi_{\sfr} \, .
     \label{eq:Et_result}
\end{align}
\label{thm:SE_gen}
\end{thm}

Theorem \ref{thm:SE_gen} can be proved by applying \cite[Thm.\@ 1]{javanmard2013state}, which gives an analogous state evolution result for an AMP recursion defined via a \emph{symmetric} Gaussian matrix; see \cite[Section V.A]{cobo2023bayes}.

\subsection{Reduction of Spatially Coupled Matrix AMP to Abstract AMP}
\label{subsec:proof_main}
We define the i.i.d.\@ Gaussian matrix $\bM \in \reals^{\tn \times L}$ in terms of the spatially coupled matrix  $\bA$ in \eqref{eq:construct_sc_A} as follows.
For $i \in [\tn], \, \ell \in [L]$,
\begin{equation}
    M_{i\ell} =
    \begin{cases}
        A_{i\ell}/\sqrt{\Lr \, W_{\sfr(i), \sfc(\ell)} }, & \quad \text{ if } W_{\sfr(i), \sfc(\ell)} \neq 0, \\
        \stackrel{\text{indep.}}{\sim}\normal(0, 1/{\tn}), & \quad \text{ otherwise }.
    \end{cases}
    \label{eq:M_def}
\end{equation}
From \eqref{eq:construct_sc_A}, we have that $M_{i\ell} \stackrel{\text{i.i.d.}}{\sim} \normal(0, {1}/{\tn})$ for $i \in [\tn], \, \ell \in [L]$. Using the matrix $\bM$, we define an abstract AMP recursion for the form in \eqref{eq:Ht_update}--\eqref{eq:ftj_gti}, with iterates 
$\be^t \in \reals^{\tn \times \k\Lr}$ and $\bh^{t+1} \in \reals^{L \times \k\Lc}$. This is done via the following choice of functions $\tf_t : \reals^{\k\Lc} \times \reals^{\k} \times [\Lc] \to \reals^{\k \Lr}$ and 
$\tg_t: \reals^{\k \Lr} \times \reals^{\k} \times [\Lr] \to \reals^{\k\Lc}$,  for $t \ge 0$. In \eqref{eq:ftj_gti}, we  take $\bbeta := \bX $ and $\bgamma := \bEps$, and let
\begin{align}
    & \tf_{t}(\bh_\ell, \, \bX_\ell, \ \sfc )  
    = \sqrt{\Lr} \, \left[\left(\eta_{t-1,\sfc}(\bX_\ell- \bh_{\ell, \sfc}) - \bX_\ell\right)^\top \, [\sqrt{W_{1\sfc}}, \ldots, \sqrt{W_{\Lr \sfc}}]\right]^\top 
   \quad\text{for } \ell \in \mc{L}_{\sfc}, \ \bh_\ell \in \reals^{\k\Lc}, \label{eq:ft_choice}  \\
    & \tg_{t}(\be_i, \, \bEps_i, \ \sfr )  
    = \sqrt{\Lr}
\,     \left[\left(\be_{i,\sfr}-\bEps_i\right)^\top \, [\bQ^t_{\sfr 1}\sqrt{W_{\sfr 1}}, \,  \bQ^t_{\sfr 2}\sqrt{W_{\sfr 2}}, \ldots, \bQ^t_{\sfr \Lc}\sqrt{W_{\sfr \Lc}}] \right]^\top 
\quad  \text{for } i \in \mc{I}_{\sfr}, \ \be_i \in \reals^{\k\Lr}.
\label{eq:gt_choice}
\end{align}
Here $\bh_\ell=[\bh_{\ell,1}, \bh_{\ell,2}, \ldots, \bh_{\ell,\Lc}]$ and  $\bh_{\ell, \sfc}\in \reals^{\k}$ denotes the  $\sfc$th length-$k$  section  of the vector $\bh_\ell \in \reals^{\k\Lc}$. Similarly, $\be_{i, \sfr}\in\reals^{\k}$ denotes the  $\sfr$th length-$k$ section of the vector $\be_{i} \in \reals^{\k\Lr}$.  
We note that the vectors in  \eqref{eq:ft_choice}--\eqref{eq:gt_choice} are treated as column vectors even though they represent the $j$th row of the functions $f_t, g_t$, as given in \eqref{eq:ftj_gti}. 
$\bQ^t_{\sfr \sfc}\in\reals^{\k\times \k}$ is the ($\sfr, \sfc$)th $k\times k$ submatrix of $\bQ^t\in\reals^{\k\Lr \times \k\Lc}$  defined in \eqref{eq:Q_t_def}.

Consider the AMP algorithm \eqref{eq:Ht_update}--\eqref{eq:Et_update} defined via the matrix $\bM$ in \eqref{eq:M_def}, the functions in \eqref{eq:ft_choice}--\eqref{eq:gt_choice}, and the matrices $\sB_t, \sD_t$. The algorithm is initialized with $\bi^0 \equiv  f_0(\bh^0, \, \bbeta)$, whose $\ell$th row is given by
\begin{align}
    \bi^0_\ell = \sqrt{\R} \, \left[(\bX_\ell^0 - \bX_\ell)^\top \, \left[\sqrt{W_{1\sfc}}, \ldots, \sqrt{W_{\Lr \sfc}}\right]\right]^\top, 
    \quad\text{for }\ell\in\mc{L}_\sfc, \  \bi_\ell^0\in\reals^{k\R}.
    \label{eq:i0_init}
\end{align}
We set $\be^0 = \bM \bi^0$.

With the choice of functions above, the state evolution recursion \eqref{eq:Omega_t1}--\eqref{eq:Gt1_c} reduces to the following. Given $\bSigma^t \in \reals^{\k\R \times \k\R}$ and $\bOmega^{t} \in \reals^{\k\C \times \k\C}$, the $(\sfc,\sfc')$th $k\times k$ submatrix of $ \hat{\bOmega}^{t+1, \sfr}\in \reals^{\k\C \times \k\C} $ for $ \sfr\in [\R]$ takes the form
\begin{align}
\left[ \hat{\bOmega}^{t+1, \sfr} \right]_{\sfc  \sfc'} 
&= \R \,  \bQ^t_{\sfr\sfc}{}^\top\E\left\{ \left(\tilde{\bU}^t_{\sfr} - \bbE\right)\left(\tilde{\bU}^t_{\sfr} - \bbE\right)^{\top}\right\}\bQ^t_{\sfr\sfc'}
\,  \sqrt{W_{\sfr \sfc} W_{\sfr \sfc'}}, 
 \quad \text{for}\quad\sfc, \, \sfc' \in [\C].
\label{eq:hOmegat1_def}
\end{align}
%
Here $\tilde{\bU}^t_{\sfr}\in \reals^{\k}$ denotes the $\sfr$th length-$k$ section of $\tilde{\bU}^t \in \reals^{\k\R}$ with $ \tilde{\bU}^t \sim \normal(\bzero,  \bSigma^t)  \perp  \bbE \sim p_{\bbE}$.
Both $\tilde{\bU}^t_{\sfr}$ and $\bbE$ are treated as column vectors.
Similarly, the $(\sfr, \sfr')$th $k\times k$ submatrix of $ \hat{\bSigma}^{t+1, \sfc} \in \reals^{\k\R \times \k\R}$ for $\sfc\in [\C]$ takes the form 
%
%
\begin{align}
  \left[  \hat{\bSigma}^{t+1, \sfc} \right]_{\sfr\sfr'}
   =
 \k\mu\,\R \E \left\{ \left(\eta_{t+1,\sfc}(\bbX + {\bU}^{t+1}_{\sfc}) - \bbX\right) 
  \left(\eta_{t+1,\sfc}(\bbX + {\bU}^{t+1}_{\sfc}) - \bbX\right)^\top\vphantom{^\top} \right\} &\sqrt{W_{\sfr \sfc} W_{\sfr'\sfc}} \, ,  \nonumber\\
 & \text{for} \quad \sfr, \, \sfr'\in [\R].
  \label{eq:hSigmat1_def}
\end{align}
Here ${\bU}^t_{\sfc}\in \reals^{\k}$ denotes the $\sfc$th length-$k$ section of ${\bU}^t \in \reals^{\k\Lc}$ with ${\bU}^t \sim \normal(\bzero, \bOmega^t) \perp  \bbX \sim p_{\bbX}$. 
The state evolution is initialized with $\hat{\bSigma}^{0, \sfc} = k\mu \lim_{L \to \infty}   \frac{1}{(L / \C)}(\bi^0_{\sfc})^{\sT} \bi^0_{\sfc}$, for $\sfc \in [\C]$, with $\bi^0 \in \reals^{L \times k\R}$  defined in \eqref{eq:i0_init}. Using Assumption $\textbf{(A0)}$ (see \eqref{eq:init_limits_mat}), its entries are given by 
\begin{align}
  \left[  \hat{\bSigma}^{0, \sfc} \right]_{\sfr\sfr'} =
 k\mu\R \, \bXi_{\sfc} \sqrt{W_{\sfr \sfc} W_{\sfr' \sfc}} \, , \quad \sfr\sfr' \in [\R].
  \label{eq:hSigma0}
\end{align}
We then have $\bSigma^0 = \frac{1}{\C} \sum_{\sfc \in [\C]} \hat{\bSigma}^{0, \sfc}$.

To prove Theorem \ref{thm:gen_sc_gamp_SE}, we show that for $t \ge 0$ and $\sfr \in [\R]$, $\sfc \in [\C]$:
\begin{align}
 &   \left[\bSigma^{t+1}\right]_{\sfr\sfr} \,  = \,  \k\muin\sum_{\sfc=1}^\Lc W_{\sfr\sfc}\bPsi_\sfc^{t+1},  \quad   \text{and}\quad 
 \left[  \bOmega^{t+1} \right]_{\sfc \sfc} \, = \, \left[\sum_{\sfr=1}^\Lr W_{\sfr\sfc}{\left[\bPhi_{\sfr}^t\right]}^{-1}\right]^{-1},
  \label{eq:SE_equiv}
\end{align}
where the right side of the equalities consists of the state evolution parameters defined in \eqref{eq:SC_SE_Phi_t}--\eqref{eq:G_c_t}. This implies that $\bU_{\sfc}^{t+1}\stackrel{\rm d}{=}{\bG}_\sfc^t$ for $\sfc\in[\C]$ and $\tilde{\bU}_{\sfr}^t + \bbE \stackrel{\rm d}{=} \tilde{\bG}_\sfr^t$ for $\sfr \in [\R]$, where ${\bG}_\sfc^t\sim \normal_d(\bzero, \bT_{\sfc}^t)$ and $\tilde{\bG}_{\sfr}^t\sim \mathcal{N}_d\left(\0, \bPhi_\sfr^t\right)$, with $\bPhi_\sfr^t$ and  $\bT_\sfc^t$ defined in \eqref{eq:SC_SE_Phi_t} and  \eqref{eq:G_c_t}. 
We then show that, for $t\geq0$,
\begin{align}
& \be^t_{i, \sfr} = -\bZ^t_i + \bEps_i, 
\quad \text{ for } i \in \mc{I}_{\sfr}, \ \sfr \in [\Lr], \label{eq:eti_claim}\\ 
& \bh^{t+1}_{\ell, \sfc} = \bX_\ell - \bS^t_\ell, \quad \text{ for } \ell \in \mc{L}_{\sfc} , \ \sfc \in [\Lc].
\label{eq:ht1j_claim}
\end{align}
Here $\bZ_i^t, \bX_\ell^t$ denote the $i$th and $\ell$th rows, respectively, of the spatially coupled AMP iterates $\bZ^t$ and $\bX^t$. Theorem \ref{thm:gen_sc_gamp_SE} follows by substituting \eqref{eq:SE_equiv}--\eqref{eq:ht1j_claim} into  Theorem  \ref{thm:SE_gen}, the state evolution result for the abstract AMP.

We now prove  \eqref{eq:SE_equiv} and then \eqref{eq:eti_claim}--\eqref{eq:ht1j_claim}.
\paragraph{Proof of \eqref{eq:SE_equiv}} For $t = 0$, from \eqref{eq:hSigma0} we have
\begin{equation}
\begin{split}
      \left[  \hat{\bSigma}^{0, \sfc} \right]_{\sfr \sfr}
  =  k\mu\R \, \bXi_{\sfc} W_{\sfr \sfc}.
   \end{split}
\end{equation}
Therefore
\begin{equation}
   \left[  \bSigma^{0} \right]_{\sfr \sfr}
  =  k\mu \frac{\R}{\C} \sum_{\sfc \in [\C]} \bXi_{\sfc} W_{\sfr \sfc}.
   \label{eq:Sig0_step}
\end{equation}
Noting from \eqref{eq:init_limits_mat}--\eqref{eq:Psic0} that $\bPsi_c^0 = \bXi_\sfc$ for $\sfc \in [\C]$ and  $\muin =  \frac{\R}{\C}\mu$, we see that the first claim in \eqref{eq:SE_equiv} holds for $t=0$. Assume towards induction that for some $t \ge 0$ we have 
\begin{align}
 \left[\bSigma^{t}\right]_{\sfr\sfr} \,  = \,  k\muin\sum_{\sfc=1}^\Lc W_{\sfr\sfc}\bPsi_\sfc^{t},  \quad   \text{  and }   \quad  \left[  \bOmega^{t} \right]_{\sfc \sfc} \, = \, \left[\sum_{\sfr=1}^\Lr W_{\sfr\sfc}{[\bPhi_{\sfr}^{t-1}]}^{-1}\right]^{-1},
 \label{eq:SE_equiv_hyp}
\end{align}
(The second part of the induction hypothesis \eqref{eq:SE_equiv_hyp} is ignored for $t=0$.) From \eqref{eq:hOmegat1_def} we have: 
\begin{align}
\left[ \bOmega^{t+1} \right]_{\sfc \sfc} & = 
\frac{1}{\R} \sum_{\sfr \in [\R]} \left[ \hat{\bOmega}^{t+1, \sfr} \right]_{\sfc \, \sfc}
=  \frac{\R}{\R} \sum_{\sfr \in [\R]} \,  \btQ^t_{\sfr\sfc}{}^\top\E\left\{ \left(\tilde{\bU}^t_{\sfr} - \bbE\right)\left(\tilde{\bU}^t_{\sfr} - \bbE\right)^{\top}\right\}\btQ^t_{\sfr\sfc}
\,  W_{\sfr \sfc}, \nonumber \\
& = \sum_{\sfr \in [\R]} \,  W_{\sfr \sfc}\,\btQ^t_{\sfr\sfc}{}^\top\left(\E\left\{ \tilde{\bU}^t_{\sfr}\tilde{\bU}^t_{\sfr}{}^{\top}\right\} - \E\left\{\bbE\tilde{\bU}^t_{\sfr}{}^{\top}\right\} - \E\left\{\tilde{\bU}^t_{\sfr}\bbE^{\top}\right\} + \E\left\{\bbE\bbE^{\top}\right\}\right)\btQ^t_{\sfr\sfc}
\,  , \nonumber \\
& = \sum_{\sfr \in [\R]} W_{\sfr \sfc}\,  \btQ^t_{\sfr\sfc}{}^\top\left(\left[  \bSigma^{t} \right]_{\sfr \sfr} + \sigma^2\bI\right)\btQ^t_{\sfr\sfc}
\,  ,
\label{eq:Omegat1_step}
\end{align}
where $\tilde{\bU}^t_{\sfr}, \bbE\in \reals^k$ are treated as column vectors. Recalling from \eqref{eq:SC_SE_Phi_t} and \eqref{eq:Q_t_def} that
\begin{align}
    \bPhi_\sfr^t = \left[\Sigma^{t} \right]_{\sfr \sfr} + \sigma^2\bI,
 \qquad   \btQ^t_{\sfr\sfc} =  [\bPhi_\sfr^t]^{-1}\left[\sum_{\sfr'=1}^\R W_{\sfr' \sfc}[\bPhi_{\sfr'}^t]^{-1}\right]^{-1},
\end{align}
the expression for $\left[ \bOmega^{t+1} \right]_{\sfc \sfc}$ in \eqref{eq:Omegat1_step} can be rewritten as
\begin{align}
    \left[ \bOmega^{t+1} \right]_{\sfc \sfc}=&\sum_{\sfr \in [\R]}  W_{\sfr \sfc}\,  \left(\left[\sum_{\sfr' \in [\R]} W_{\sfr' \sfc}[\bPhi_{\sfr'}^t]^{-1}\right]^{-1}\right)^{\top}\left([\bPhi_\sfr^t]^{-1}\right)^{\top} \bPhi_\sfr^t[\bPhi_\sfr^t]^{-1}\left[\sum_{\sfr' \in [\R]} W_{\sfr' \sfc}[\bPhi_{\sfr'}^t]^{\ -1}\right]^{-1}\nonumber\\
    =& \sum_{\sfr \in [\R]}  W_{\sfr \sfc}\,  \left[\sum_{\sfr' \in [\R]} W_{\sfr' \sfc}[\bPhi_{\sfr'}^t]^{\ -1}\right]^{-1}[\bPhi_\sfr^t]^{-1} \bPhi_\sfr^t[\bPhi_\sfr^t]^{-1}\left[\sum_{\sfr' \in [\R]} W_{\sfr' \sfc}[\bPhi_{\sfr'}^t]^{\ -1}\right]^{-1}\nonumber\\
    =& \left[\sum_{\sfr' \in [\R]} W_{\sfr' \sfc}[\bPhi_{\sfr'}^t]^{-1}\right]^{-1} \sum_{\sfr \in [\R]}  W_{\sfr \sfc}\,  \left[ \bPhi_\sfr^t \right]^{-1}\left[\sum_{\sfr' \in [\R]} W_{\sfr' \sfc}[\bPhi_{\sfr'}^t]^{-1}\right]^{-1} \nonumber \\
    =& \left[\sum_{\sfr \in [\R]} W_{\sfr' \sfc}[\bPhi_{\sfr'}^t]^{-1}\right]^{-1},\label{eq:Omegat1_cc}
\end{align}
where we have used the fact that $\bPhi_\sfr^t$ is a sum of covariance matrices and is therefore symmetric, as is its inverse. This implies that $\bU_\sfc^{t+1}\stackrel{\rm d}{=}{\bG}_\sfc^t$.

Next, from \eqref{eq:hSigmat1_def} we have that 
\begin{align}
    \left[  \bSigma^{t+1} \right]_{\sfr \sfr} \, 
&= \frac{k\mu}{\C} \sum_{\sfc \in [\C]} \left[  \hat{\bSigma}^{t+1, \sfc} \right]_{\sfr \sfr} \nonumber\\
&= \frac{k\mu\R}{\C}  \sum_{\sfc \in [\C]} W_{\sfr \sfc} \, \E \left\{ \left(\eta_{t+1,\sfc}(\bbX + \bU^{t+1}_{\sfc}) - \bbX\right) \left(\eta_{t+1,\sfc}(\bbX + \bU^{t+1}_{\sfc}) - \bbX\right)^\top \right\} \, \nonumber\\
 \label{eq:Sigt1_rr}
&= k\muin\sum_{\sfc \in [\C]}W_{\sfr \sfc}\bPsi_\sfc^{t+1},
\end{align}
where we recall that $\bU^{t+1}_{\sfc} \sim \normal\left(0, [\bOmega^{t+1}]_{\sfc\sfc}\right)$ and the definition of $\bPsi_\sfc^{t}$ in \eqref{eq:SC_SE_Psi_t}. 
With the expressions in \eqref{eq:Omegat1_cc} and \eqref{eq:Sigt1_rr}, we complete the induction step and hence, prove \eqref{eq:SE_equiv}. 

\paragraph{Proof of \eqref{eq:eti_claim} and \eqref{eq:ht1j_claim}}  At $t=0$,   the algorithm is initialized with $\bi^0$ defined in \eqref{eq:i0_init},  and   $\be^0 = \bM  \bi^0 \in  \reals^{\tn \times k\R}$. For $\sfr \in [\R]$ and $i \in \mc{I}_\sfr$, consider the $\sfr$th section of the  $i$th row of $\be^0$, $\be_i^0\in\reals^{k\R}$. Writing the sum over $\ell$ from $1$ to $L$ as a double sum, we have that:
\begin{align}
& \be^0_{i, \sfr}  = \sum_{\sfc = 1}^{\C}  \sum_{\ell \in \mc{L}_\sfc}  M_{i \ell } \sqrt{R \, W_{\sfr\sfc}} \, \left(\bX_\ell^0 - \bX_\ell\right).  \label{eq:e0_ir} 
\end{align}

Recalling the definition of $M_{i\ell}$ in \eqref{eq:M_def},  we  have that  
\begin{align}
\be^0_{i, \sfr} =  \sum_{\ell=1}^L A_{i\ell}\left(\bX_\ell^0 -\bX_\ell\right) = \left(\bA\bX^0\right)_i - \left(\bA\bX\right)_i = -\bZ^0_i + \bEps_i, \label{eq:e0i}
\end{align} 
where the last equality follows from the initialization of $\bZ^0$ in \eqref{eq:Z_init}.

Assume towards induction that for some $t \ge 0$, we have:
\begin{align}
& \be^t_{i, \sfr} = -\bZ^t_i + \bEps_i,
\quad \text{ for } i \in \mc{I}_{\sfr}, \ \sfr \in [\Lr], \label{eq:eti_hyp}\\ 
& \bh^{t}_{\ell, \sfc} = \bX_\ell - \bX_\ell^{t-1} - \bV_\ell^{t-1}, \quad \text{ for } \ell \in \mc{L}_{\sfc} , \ \sfc \in [\Lc].
\label{eq:htj_hyp}
\end{align}
(For $t=0$, we only assume \eqref{eq:eti_hyp}, as shown in \eqref{eq:e0i}.) 

From \eqref{eq:Ht_update}, the $\ell$th row of $\bh^{t+1}$, for $\ell \in \mc{L}_{\sfc}$, is given by 
\begin{equation}
\bh^{t+1}_\ell = \sum_{i=1}^{\tn} M_{i\ell} \,  \tg_t(\be_i^t,\bEps_i, \sfr(i) )
\, - \, \tf_t(\bh_\ell^t, \bX_\ell, \sfc) {\sD}_t^{\top},
\label{eq:htj_exp}
\end{equation}
where from \eqref{eq:abstract_Dt_Bt},
\begin{align}
\sD_t = \frac{1}{\tn} \sum_{i=1}^{\tn} \tg_t'(\be^t_i, \bEps_i, \sfr(i)) = \frac{1}{\R} \sum_{\sfr=1}^\R \frac{1}{\tn/\R}\sum_{i \in \mc{I}_\sfr}\tg_t'(\be^t_i, \bEps_i, \sfr)\in \reals^{k\Lc \times k\Lr }.
\end{align}
Here $\tg_t'(\be^t_i, \bEps_i, \sfr(i))$ denotes the Jacobian with respect to  the first argument $\be^t_i\in \reals^{k\R}$. Using the definition of $\tg_t$ in \eqref{eq:gt_choice}, the $k\times k$ submatrices of its Jacobian are:  
\begin{align}
[ \tg_t'(\be^t_i, \bEps_i,  \, \sfr) ]_{\mathsf{j} \mathsf{m}} 
= \begin{cases}
    \sqrt{R} \sqrt{W_{\sfr \mathsf{j}}} \, \btQ^t_{\sfr \mathsf{j}}, & \text{ for } \mathsf{j} \in [\C],  \, \mathsf{m} = \sfr, \\
     \bzero, & \text{ otherwise.}
\end{cases}
\end{align}
Using this, we obtain that $[ {\sD}_t ]_{\sfc \sfr}\in\reals^{k\times k}$ is given by
\begin{align}
   &[ {\sD}_t ]_{\sfc \sfr}
   =  
    \frac{1}{\sqrt{\R}} 
        \sqrt{W_{\sfr \sfc}} \, \btQ_{\sfr\sfc}^t,\quad  \sfc \in [\C],   \sfr \in [\R]. \label{eq:barDt_exp01}
\end{align}
Using \eqref{eq:barDt_exp01} in \eqref{eq:htj_exp} along with the definition of $\tf_t$ from \eqref{eq:ft_choice}, we obtain
\begin{equation}
\begin{split}
   &\bh^{t+1}_{\ell, \sfc}  =  \sum_{ \sfr = 1}^{\R} \sum_{i \in \mc{I}_\sfr} M_{i\ell }  \sqrt{\Lr}
\,     \left(\be^t_{i,\sfr}-\bEps_i\right) \sqrt{W_{\sfr \sfc}}\btQ^t_{\sfr \sfc} \, - \, \left(\eta_{t-1,\sfc}(\bX_\ell - \bh_{\ell, \sfc}^t) - \bX_\ell\right) \, \sum_{\sfr=1}^\R W_{\sfr \sfc}\btQ_{\sfr\sfc}^t{}^\top
    \end{split}
    \label{eq:ht1_jc0}
\end{equation}
By the induction hypothesis, we have that
$\be^t_{i, \sfr} = -\bZ^t_i + \bEps_i$, for $i \in \mc{I}_{\sfr}$, and $\bh^{t}_{\ell, \sfc} = \bX_\ell - \bX_\ell^{t-1} - \bV_\ell^{t-1}$ for $\ell \in \mc{L}_{\sfc}$. (For $t=0$,  we only assume the hypothesis on $\be^t$, and the formula in \eqref{eq:ht1_jc0} holds for $\bh^{1}_{\ell, \sfc}$ with  the term $\eta_{t-1,\sfc}(\bX_\ell - \bh_{\ell, \sfc}^t)$  replaced by $\bX_\ell^0$.) 
Moreover, substituting the expression for $M_{i\ell}$ in \eqref{eq:M_def} and noting from the definition of $\btQ_{\sfr\sfc}^t$ in \eqref{eq:Q_t_def} that $\sum_{\sfr=1}^\R W_{\sfr \sfc}\btQ_{\sfr\sfc}^t{}^\top=\boldsymbol{I}_{d\times d}$ we obtain:
\begin{align}
    \bh^{t+1}_{\ell, \sfc} & = \sum_{i=1}^{\tn} A_{i\ell} \left(\be^t_{i,\sfr(i)}-\bEps_i\right) \btQ^t_{\sfr(i) \sfc}
    \, - \, \left(\eta_{t-1,\sfc}(\bX_\ell - \bh_{\ell, \sfc}^t) - \bX_\ell\right)\nonumber\\
    & = -\sum_{i=1}^{\tn} A_{i\ell} \bZ_i^t \btQ^t_{\sfr(i) \sfc}
    \, - \, \left(\eta_{t-1,\sfc}(\bX_\ell - \bh_{\ell, \sfc}^t) - \bX_\ell\right) \nonumber\\
    & = \bX_\ell \, - \, \eta_{t-1,\sfc}(\bX_\ell - \bh_{\ell, \sfc}^t)\, - \, \sum_{i=1}^{\tn} A_{i\ell} \bZ_i^t \btQ^t_{\sfr(i) \sfc}
    \nonumber\\
    & = \bX_\ell 
    \, - \, \eta_{t-1,\sfc}(\bX_\ell^{t-1} + \bV_\ell^{t-1})-\sum_{i=1}^{\tn} A_{i\ell} \bZ_i \btQ^t_{\sfr(i) \sfc}\nonumber\\
    & = \bX_\ell - \bX_\ell^{t} - \bV_\ell^{t}.
    \label{eq:ht1_step}
\end{align}
where the last equality follows from the definition of $\bV^t$ and $\bX^t$ in \eqref{eq:SC_AMP_V_t} and \eqref{eq:alg_sc_Xt}.

Next consider the $i$th row of $\be^{t+1}$, for $i \in \mc{I}_\sfr$, which from \eqref{eq:Et_update}, is given by 
\begin{equation}
e^{t+1}_i = \sum_{\ell=1}^L M_{i\ell} \, \tf_{t+1}(\bh_\ell^{t+1}, \bX_\ell, \sfc(\ell)) \, - \, \tg_{t}(\be^{t}_i, \bEps_i, \sfr) \,  {\sB}_{t+1}^{\top},
\label{eq:eti_exp}
\end{equation}
where from \eqref{eq:abstract_Dt_Bt},  
\begin{align}
\sB_{t+1} = \frac{1}{\tn} \sum_{\ell=1}^L \tf_{t+1}'(\bh^{t+1}_\ell, \bX_\ell, \sfc(\ell))=k\mu\frac{1}{\C} \sum_{\sfc=1}^{\C}\frac{1}{L/\C} \sum_{\ell\in \mc{L}_\sfc}\tf_{t+1}'(\bh^{t+1}_\ell, \bX_\ell, \sfc)\, \in \reals^{k\Lr \times k\Lc}.
\end{align}
From the definition of $\tf_t$ in \eqref{eq:ft_choice}, the $k\times k$ partitions of the Jacobian $\tf_{t+1}'(\bh^{t+1}_\ell, \bX_\ell, \sfc(\ell))$ for $\ell\in\mc{L}_c$ are: 
\begin{equation}
   [\tf_{t+1}'(\bh_\ell^{t+1}, \bX_\ell,  \, \sfc)]_{\mathsf{j}   \mathsf{m}} = \begin{cases}
        -\sqrt{\Lr} \sqrt{W_{\mathsf{j} \sfc}}\, \left[\eta'_{t,\sfc}(\bX_\ell - \bh_{\ell, \sfc}^{t+1})\right],  &    \text{ for } \mathsf{j} \in [\Lr],  \mathsf{m}= \sfc, \\
       \bzero, &  \text{ otherwise},
   \end{cases}
\end{equation}
where $\left[\eta'_{t,\sfc}(\bX_\ell - \bh_{\ell, \sfc}^{t+1})\right]\in\reals^{k\times k}$ represents the Jacobian of $\eta_{t,\sfc}(\bX_\ell - \bh_{\ell, \sfc}^{t+1})$. Therefore,
\begin{align}
   [ {\sB}_{t+1} ]_{\sfr \sfc}  =  
   -\frac{1}{\delta} \frac{1}{\C} 
        \sqrt{\Lr} \sqrt{W_{\sfr \sfc}}\, \frac{1}{L/\C}\sum_{\ell\in\mc{L}_c}\{ \eta'_{t,\sfc}(\bX_\ell - \bh_{\ell, \sfc}^{t+1}) \} ,  &    \text{ for } \sfr \in [\Lr],  \sfc \in [\Lc].
   \label{eq:barBt_exp_mat}
\end{align}

We now use \eqref{eq:barBt_exp_mat} in \eqref{eq:eti_exp} along with the definition of $\tg_{t}$ in \eqref{eq:gt_choice}. For $i \in \mc{I}_{\sfr}$, $\sfr \in [\Lr]$, we have:
\begin{align}
\be^{t+1}_{i, \sfr} &= \sum_{\sfc = 1}^{\C}  \sum_{\ell \in \mc{L}_\sfc}  M_{i\ell} \sqrt{R \, W_{\sfr\sfc}}
\, \left(\eta_{t,\sfc}(\bX_\ell - \bh_{\ell, \sfc}^{t+1}) - \bX_\ell\right) \nonumber \\
& \qquad  + \left(\be_{i,\sfr}^t-\bEps_i\right) \, \frac{k\mu\Lr}{\Lc} \sum_{\sfc =1}^{\Lc} W_{\sfr \sfc}\btQ^t_{\sfr \sfc} \,  \left\{\frac{1}{L/\C}\sum_{\ell\in\mc{L}_c} \left(\eta'_{t,\sfc}(\bX_\ell - \bh_{\ell, \sfc}^{t+1}) \right)^\top\right\}. \label{eq:et_ir}
\end{align}
Next, by the induction hypothesis we have that 
$\be^{t}_{i, \sfr} = -\bZ_i^t + \bEps_i$ for $i \in \mc{I}_\sfr$, and we have shown above  that $\bh^{t+1}_{\ell, \sfc} = \bX_\ell - \bX_\ell^{t} - \bV_\ell^{t} $  for $\ell \in \mc{L}_{\sfc}$. 
Using these in \eqref{eq:et_ir} along with $\muin =  \frac{\R}{\C}\mu$,  we obtain: 
\begin{align}
   \be^{t+1}_{i, \sfr} &  =  \sum_{\ell=1}^L A_{i\ell} \, \left(\eta_{t,\sfc(\ell)}(\bX_\ell^{t} + \bV_\ell^{t}) - \bX_\ell\right)
   \, - \,k\muin\bZ_i^t \,  \sum_{\sfc =1}^{\Lc} W_{\sfr \sfc}\btQ^t_{\sfr \sfc} \,  \left\{\frac{1}{L/\C}\sum_{\ell\in\mc{L}_c} \left(\eta'_{t,\sfc}(\bX_\ell - \bh_{\ell, \sfc}^{t+1}) \right)^\top\right\}\nonumber \\
   & = -\left(\bY_i -\bEps_i\right) + \sum_{\ell=1}^L A_{i\ell} \, \eta_{t,\sfc(\ell)}(\bX_\ell^{t} + \bV_\ell^{t})
   \, - \btZ_i^{t+1}\nonumber\\
   & = \bEps_i - \left( \bY_i - \sum_{\ell=1}^L A_{i\ell} \, \bX_\ell^{t+1}
   \, + \,\btZ_i^{t+1}\right)\nonumber\\
   & = \bEps_i - \bZ_i^{t+1},
\end{align}
where the second equality follows from the definition of $\btZ^t$ in \eqref{eq:Z_tilde} and the last equality follows from \eqref{eq:alg_sc_Zt}.

This completes the proof of the claims in \eqref{eq:eti_claim} and \eqref{eq:ht1j_claim}, and Theorem \ref{thm:gen_sc_gamp_SE} follows from Theorem \ref{thm:SE_gen}, the state evolution result for abstract AMP iterations. \qed

\section{Conclusion and Future Work}
This paper studied  the Gaussian multiple  access channel with random user activity  in the many-user  regime, where the number of users $L$ scales linearly with the code length $n$. 
Each user has their own codebook, and decoding may result in  three types of errors: \edit{missed detection} ($\MD$), false alarm ($\FA$), and active user error ($\AUE$).   We derived two  achievability bounds on the error probabilities  ($\pMD, \pFA$ and $\pAUE$), in terms of the  user payload $\mu$ and  the distribution of the number of  active users. The first  bound is based on a finite-length analysis of Gaussian random codebooks with  maximum-likelihood decoding. The second bound was  obtained by characterizing the asymptotic performance of spatially coupled Gaussian codebooks with AMP decoding. In the special case where the users are always active, our asymptotic achievability bound is strictly tighter than the existing bounds  by Kowshik \cite{kowshik2022improved} and Zadik et  al.\@ \cite{zadik2019improved} for moderate or large user payloads (e.g., hundreds of bits). 

 We also proposed an efficient CDMA-type  scheme with a spatially coupled signature matrix and AMP decoding. The AMP decoder can be tailored to take advantage of prior information on the users' messages or activity patterns. Precise asymptotic guarantees   on its error performance  were derived by establishing a general state evolution result for spatially coupled AMP with matrix-valued iterates.  Although we considered spatially coupled Gaussian signature matrices, using recent  results on AMP universality \cite{Wan22, tan2024group}, the decoding algorithm and all the theoretical results remain valid for a much broader class of `generalized white noise' matrices. This class includes spatially coupled matrices with sub-Gaussian entries, so the results apply to the popular setting of random binary-valued signature sequences \cite{shamaiVerdu01, guo2005randomly}.  Extending our CDMA-type scheme to handle other  priors, such as modulation \cite{hsieh2022near},  coding \cite{liu2024LDPC}, channel  fading \cite{gao2024unsourced},  or common alarm messages \cite{Ngo2024unsourced_alarm} is a promising direction for future work.  In particular, the  AMP decoder with the thresholding denoiser \eqref{eq:thres_denoiser} has low complexity and
could be integrated with a decoder for an outer code, similarly to \cite{ebert2023sparse, liu2024LDPC}, to tackle coded random access.

This paper investigated random access in the conventional GMAC setting where users have distinct codebooks. It would be interesting to adapt the proposed coding schemes  to \emph{unsourced}  random access \cite{polyanskiy2017perspective, ngo2022unsourced}, where users share a common codebook and the decoder returns a list of messages without needing to identify the senders. 
Several works \cite{Amalladinne2020Coded,fengler2021SPARCS,amalladinne2022unsourced} have investigated schemes for  unsourced random access based on sparse regression codes (SPARCs) and AMP decoding. A coded CDMA-type scheme with AMP decoding  could significantly reduce  the   computational and memory complexity  compared to unsourced random access schemes based on SPARCs.

\appendix

\section{Implementation Details}
\label{appx:implementation}
In our  experiments in Section \ref{sec:sparc_ach_numerical}, we find that the asymptotic error bounds in Theorem \ref{thm:sparc_asymp_errors} and Corollary \ref{cor:asymp_error_alpha=1} 
associated  with $\mc{F}_\Bayes$ cannot be computed efficiently for user payloads $k>8$ bits, while the error bounds associated  with $\mc{F}_\marginal$ can be computed in an efficient and numerically stable way  for significantly larger payloads $k$. The exact threshold on $k$ is determined by the precision level employed. For instance, with Python's  64-bit (double) precision, we can handle up to $k< 63$  bits. This observation is consistent with those in \cite{hsieh2022near, kowshik2022improved}.

Since $M=2^k$,  when computing the error bounds $\epsMD, \epsFA$ or $\epsAUE$ for larger  $k$ we may encounter numerical issues  due to the terms with $M$ or $(M-1)$ in  their exponents (see 
\eqref{eq:sparc_epsMD_def}--\eqref{eq:sparc_epsAUE_def} and  \eqref{eq:asymp_pAUE_sparc_alpha=1}). To avoid such problems, we use the  approximation
$(1-\epsilon)^n \approx e^{-\epsilon n}$ when $n$ is large and $\epsilon$ is small.

For all the  experiments in  Sections \ref{sec:sparc_ach_numerical}, numerical integration is used to evaluate any expectation terms $\E_z[\cdot]$ over the standard Gaussian  $z$, including the expectation  in the mutual information term in $\mc{F}_\marginal$ (see \eqref{eq:pot_fn_entrywise}--\eqref{eq:MI_entrywise}) or the expectation  in  $\epsAUE$ (see 
\eqref{eq:sparc_epsAUE_def} 
or  \eqref{eq:asymp_pAUE_sparc_alpha=1}). Monte Carlo sampling is used to estimate any expectation terms $\E_{\bz}[\cdot]$ over the multivariate standard Gaussian $\bz$, such as the expectation in the mutual information term in $\mc{F}_\Bayes$ (see \eqref{eq:pot_fn_sectionwise}--\eqref{eq:MI_sectionwise}).

\section*{Acknowledgment}
We thank the anonymous reviewers for  their feedback which helped improve the paper.
 
{\small{
\bibliographystyle{IEEEtran}
 \bibliography{gmac} 
 }}

\end{document}